\documentclass[10pt,reqno,a4paper]{amsart}
\usepackage[british]{babel}
\usepackage{amssymb,amstext,amsthm,eucal,bbm,mathrsfs,amscd}
\usepackage{bbold}%,relsize}
\usepackage{stmaryrd}
\usepackage[hmargin=1in]{geometry}
\usepackage{array,color,slashed}
\usepackage[utf8]{inputenc}
\usepackage[small]{eulervm}
\usepackage{tgpagella}
\usepackage{enumitem,multicol,cite}
\usepackage[unicode]{hyperref}
\hypersetup{%
  pdftitle   = {Killing superalgebras for lorentzian six-manifolds},
  pdfkeywords = {Spencer cohomology, Lie superalgebras, superconformal symmetry, six dimensions},
  pdfauthor  = {Paul de Medeiros, José Figueroa-O'Farrill, Andrea Santi},
  pdfcreator = {\LaTeX\ with package \flqq hyperref\frqq}
}
\theoremstyle{plain}
\newtheorem{lemma}{Lemma}
\newtheorem{proposition}[lemma]{Proposition}
\newtheorem{theorem}[lemma]{Theorem}
\newtheorem{corollary}[lemma]{Corollary}
\newtheorem{definition}[lemma]{Definition}
\theoremstyle{definition}
\newtheorem*{remark}{Remark}
\newcommand{\half}{\tfrac{1}{2}}
\newcommand{\fg}{\mathfrak{g}}
\newcommand{\fa}{\mathfrak{a}}

\newcommand{\fh}{\mathfrak{h}}
\newcommand{\fk}{\mathfrak{k}}

\newcommand{\fp}{\mathfrak{p}}
\newcommand{\fphat}{\mathfrak{\hat p}}
\newcommand{\fkhat}{\mathfrak{\hat k}}
\newcommand{\fr}{\mathfrak{r}}

\newcommand{\fso}{\mathfrak{so}}
\newcommand{\fsp}{\mathfrak{sp}}
\newcommand{\RR}{\mathbb{R}}
\newcommand{\HH}{\mathbb{H}}
\newcommand{\CC}{\mathbb{C}}
\newcommand{\ZZ}{\mathbb{Z}}
\renewcommand{\SS}{\mathbb{S}}
\newcommand{\be}{\boldsymbol{e}}
\newcommand{\bep}{\boldsymbol{\epsilon}}
\newcommand{\Cl}{C\ell}
\newcommand{\1}{\mathbb{1}}
\newcommand{\sbar}{\overline{s}}
\newcommand{\eD}{\mathscr{D}}

\newcommand{\eE}{\mathscr{E}}
\newcommand{\eL}{\mathscr{L}}

\newcommand{\eR}{\mathscr{R}}
\newcommand{\eT}{\mathscr{T}}
\newcommand{\eU}{\mathscr{U}}
\newcommand{\eV}{\mathscr{V}}
\DeclareMathOperator{\Id}{Id}
\DeclareMathOperator{\tr}{tr}

\DeclareMathOperator{\Spin}{Spin}
\DeclareMathOperator{\AdS}{AdS}
\DeclareMathOperator{\Sph}{S}
\DeclareMathOperator{\SL}{SL}
\DeclareMathOperator{\Sp}{Sp}

\DeclareMathOperator{\End}{End}
\DeclareMathOperator{\Hom}{Hom}

\DeclareMathOperator{\diag}{diag}
\definecolor{orange}{rgb}{1.0, 0.5, 0.0}
\allowdisplaybreaks[1]
\begin{document}
\title{Killing superalgebras for lorentzian six-manifolds}
\author[de~Medeiros]{Paul de Medeiros}
\author[Figueroa-O'Farrill]{José Figueroa-O'Farrill}
\author[Santi]{Andrea Santi}
\address[PdM]{Department of Mathematics and Physics,
  University of Stavanger, 4036 Stavanger, Norway}
\address[JMF]{Maxwell Institute and School of Mathematics, The University of
  Edinburgh, Edinburgh EH9 3FD, Scotland}
  \address[AS]{Dipartimento di Matematica, Università di Bologna, Piazza di Porta San Donato 5, 40126, Bologna, Italy}
%\email{paul.demedeiros@uis.no,j.m.figueroa@ed.ac.uk,asanti.math@gmail.com}
\date{\today}
\begin{abstract}
  We calculate the Spencer cohomology of the $(1,0)$ Poincaré
  superalgebras in six dimensions: with and without R-symmetry.  As the cases of four and eleven dimensions taught us, we may read off from this calculation a Killing
  spinor equation which allows the determination of which geometries
  admit rigidly supersymmetric theories in this dimension.  We prove
  that the resulting Killing spinors generate a Lie superalgebra and
  determine the geometries admitting the maximal number of such Killing
  spinors. They are divided in two branches. One branch consists of the 
  lorentzian Lie groups with bi-invariant metrics and, as a special
  case, it includes the lorentzian Lie groups with a self-dual Cartan
  three-form which define the maximally supersymmetric backgrounds of $(1,0)$ Poincaré supergravity in six dimensions. The notion of Killing
  spinor on the other branch does not depend on the choice of a three-form
  but rather on a one-form valued in the R-symmetry algebra. In this case, we
  obtain  three different (up to local isometry) maximally
  supersymmetric backgrounds, which are distinguished by the causal type
  of the one-form.
\end{abstract}
\maketitle
\tableofcontents

\section{Introduction}
\label{sec:introduction}

There has been considerable interest over recent years in the systematic
exploration of curved backgrounds that support some amount of rigid
(conformal) supersymmetry. The primary motivation being that quantum
field theories on such backgrounds are often amenable to the powerful
techniques of supersymmetric localisation, typically revealing
interesting new insights and exact results \cite{Pestun:2007rz,
  Kapustin:2009kz, Drukker:2010nc, Jafferis:2010un, Jafferis:2011zi,
  Kallen:2012cs, Hosomichi:2012he, Kallen:2012va, Kim:2012ava}.   

By far the most successful strategy in this direction was initiated by
Festuccia and Seiberg \cite{Festuccia:2011ws}, originally for rigidly
supersymmetric backgrounds in four dimensions but subsequently
generalised \cite{Jia:2011hw, Samtleben:2012gy, Klare:2012gn,
  Dumitrescu:2012ha, Cassani:2012ri, Liu:2012bi, deMedeiros:2012sb,
  Samtleben:2012ua} to other dimensions in both euclidean and lorentzian
signatures. Their method takes advantage of the existence of some
locally supersymmetric supergravity theory coupled to one or more field
theory supermultiplets. In any such theory, it is possible to take a
certain rigid limit in which the Planck mass tends to infinity and the
degrees of freedom from the gravity supermultiplet are effectively
frozen out. What remains after taking this limit is a rigidly
supersymmetric field theory on a bosonic supersymmetric background of
the original supergravity theory. The Killing spinor equations which
characterise this supersymmetric background are simply read off from the
supersymmetry variation of the gravitino in the rigid limit. It is
important to emphasise that these supersymmetric backgrounds need not
solve the supergravity field equations. For example, in four dimensions,
the old minimal off-shell formulation of Poincaré supergravity contains
auxiliary fields which are all set to zero by the field equations.
However, many interesting rigidly supersymmetric backgrounds of this
theory are not solutions because they are supported by one or more
non-zero auxiliary fields \cite{Festuccia:2011ws, deMedeiros:2016srz}.

The precise details of the rigid supersymmetry supported by any bosonic
supersymmetric supergravity background are encoded by its Killing
superalgebra \cite{FigueroaO'Farrill:2004mx, FigueroaO'Farrill:2007ar,
  FigueroaO'Farrill:2007ic, FigueroaO'Farrill:2008ka,
  FigueroaO'Farrill:2008if,FigueroaO'Farrill:2012fp}. The Killing
superalgebra is a Lie superalgebra whose odd part consists of all the
Killing spinors supported by the background and whose even part contains
Killing vectors which preserve the background. For supergravity theories
with a non-trivial R-symmetry, the even part of the Killing superalgebra
may also contain R-symmetries which preserve the background. While the
appearance of the Killing superalgebra may seem somewhat peripheral in
relation to the rigid limit described above, it is clearly an object of
fundamental significance in the description of rigid supersymmetry and
in understanding special geometrical properties of the backgrounds which
support it.

So much so that, somewhat in the spirit of the Erlangen program, one
might prefer to take the classification of Killing superalgebras as the
central question, with no prior knowledge of supergravity, and then
deduce as a by product all the possible rigidly supersymmetric
backgrounds (which may or may not correspond to backgrounds of some
known supergravity theory). This is certainly the philosophy we have
adopted in some of our previous works, which has led to the
classification of Killing superalgebras for maximally supersymmetric
lorentzian backgrounds in dimensions eleven
\cite{Figueroa-OFarrill:2015rfh, Figueroa-O'Farrill:2015utu} and four
\cite{deMedeiros:2016srz}. The key property of Killing superalgebras
that permits such a classification is the fact that they are all
filtered deformations (in a certain technical sense which we review in
\S\ref{sec:KSAfd}) of some subalgebra of the Poincaré superalgebra,
possibly extended by R-symmetries. As one might expect, there is a
natural cohomology theory (a generalised version of Spencer cohomology)
which governs these filtered deformations at the infinitesimal level,
and the essence of the classification is the calculation of a certain
Spencer cohomology group in degree two. In dimensions eleven and four
\cite{Figueroa-OFarrill:2015rfh, Figueroa-O'Farrill:2015utu,
  deMedeiros:2016srz}, this calculation actually prescribes a Killing
spinor equation which is in precise agreement with the Killing spinor
equation that characterises bosonic supersymmetric backgrounds of
minimal Poincaré supergravity in these respective dimensions (more
accurately, in the \lq old minimal' off-shell formulation in four
dimensions \cite{deMedeiros:2016srz}). So, at least in these cases, all
the rigidly supersymmetric backgrounds are indeed backgrounds of a known
Poincaré supergravity theory.

In this paper, we shall extend these considerations to look at Killing
superalgebras for lorentzian backgrounds in six dimensions. There are
several reasons that make dimension six especially interesting. Recall
that the Lie superalgebra $\mathfrak{osp}(6,2|N)$ is isomorphic to the
$N$-extended conformal superalgebra of $\RR^{5,1}$, and that conformal
superalgebras do not exist in higher dimensions (at least, not in the
traditional sense of Nahm \cite{Nahm:1977tg}). Furthermore the
$\mathfrak{sp}(N)$ R-symmetry subalgebra of $\mathfrak{osp}(6,2|N)$ is
nonabelian, for any $N>0$. Now let $N=1$. By omitting the dilatations
and special conformal transformations in the even part of
$\mathfrak{osp}(6,2|1)$, together with the special conformal
supercharges in the odd part, we obtain a Lie superalgebra that we will
denote by ${\hat \fp}$. If we also omit the $\mathfrak{sp}(1)$
R-symmetry in ${\hat \fp}$, we recover the ordinary $(1,0)$ Poincaré
superalgebra in six dimensions (without R-symmetry) that we will denote
by $\fp$.

Our first goal in this paper will be to calculate the relevant Spencer
cohomology groups for both $\fp$ and ${\hat \fp}$, see Theorems
\ref{thm:spencer} and \ref{thm:spencer-R}. In marked contrast with the
situation in dimensions eleven and four, where the inclusion of
R-symmetry is immaterial, here in dimension six we will find that the
relevant Spencer cohomology groups for $\fp$ and ${\hat \fp}$ are
different. In both cases, we then go on to use the explicit expression
for a particular component of the Spencer cocycle representative to
prescribe an appropriate Killing spinor equation. For the case where the
R-symmetry is not gauged, this Killing spinor equation is given by
Definition~\ref{def:KillingSpinor} in Section~\ref{sec:killing-spinors}
while, for the case of gauged R-symmetry, the Killing spinor equation is
defined by \eqref{eq:KillingSpinorIndices} in
Section~\ref{sec:killing-superalgebras}. In both cases, on a
six-dimensional spin manifold $M$ equipped with a lorentzian metric $g$,
we find that the extra background data needed to define this Killing
spinor equation consists of a three-form $H$ and an
$\mathfrak{sp}(1)$-valued one-form $\varphi$. (For the case of gauged
R-symmetry, one must also specify a flat $\mathfrak{sp}(1)$ connection
$C$.) The important distinction is that $H$ must be self-dual for $\fp$
whereas, for ${\hat \fp}$, its anti-self-dual component $H^-$ need not
be zero. It is important to stress that the Killing spinor equation we
deduce from Spencer cohomology agrees with the Killing spinor equation
for bosonic supersymmetric backgrounds of $(1,0)$ Poincaré supergravity
in six dimensions {\emph{only}} when $H$ is self-dual and $\varphi =0$.
(For the case of gauged R-symmetry, one can remove the flat connection
$C$ by an appropriate choice of gauge.) It is an intriguing question as
to whether our more general Killing spinor equation can be recovered
from supergravity, perhaps via superconformal compensators.

We then proceed to the construction of Killing superalgebras based on
these Killing spinors.  The geometric content of the cocycle conditions
for Spencer cohomology is that the Dirac current of a Killing spinor
(derived from Spencer cohomology) is a Killing vector and that the Lie
derivative along the Dirac current annihilates the Killing spinor
itself.  In order to prove that Killing spinors generate a Lie
superalgebra, the only additional requirement is that the Lie derivative
along the Dirac current of any Killing spinor preserves the space of
Killing spinors.  This is guaranteed if the connection $\eD$ defining
the notion of a Killing spinor is invariant along the flow generated by
the Dirac current of any Killing spinor.

An equivalent condition for the invariance of $\eD$ is the invariance of
the other geometric data defining $\eD$: the three-form $H$ and
$\fsp(1)$-valued one-form $\varphi$. For $\fp$, we establish the
existence of a Killing superalgebra provided $H$ is closed and $\varphi$
is coclosed, see Theorem \ref{thm:KSAI}. For ${\hat \fp}$, if
$\varphi = 0$, we find that a Killing superalgebra exists provided $H$
is closed and $H^-$ is parallel with respect to the metric connection
with skew-symmetric torsion given by $H^+$, see Theorem \ref{thm:KSAII}.

Finally, we present in Theorem  \ref{thm:maxsusy} the classification (up
to local isometry) of all backgrounds which admit the maximal number of
Killing spinors. In addition to Minkowski space $\RR^{5,1}$, we find
that there are two distinct branches of maximally supersymmetric
backgrounds. All backgrounds on the first branch are conformally flat
and have $H=0$ with $\varphi = a \otimes R$, where $a$ is a non-zero
parallel one-form and $R$ is a non-zero element of $\mathfrak{sp}(1)$.
Up to local isometry, there are three different backgrounds on this
branch which depend only on the causal type of $a$:
\begin{itemize}
\item $\AdS_5 \times \RR$, if $a$ is spacelike;
\item $\RR \times \Sph^5$, if $a$ is timelike;
\item the symmetric plane wave with metric $g_-$ in \eqref{eq:ppwave}, if $a$ is null.
\end{itemize}
All backgrounds on the second branch have $\varphi =0$ and a non-zero
$H$ which is identified with the parallel Cartan three-form of a
six-dimensional Lie group $M$ with bi-invariant lorentzian metric $g$.
Up to local isometry, the list of different backgrounds on this branch
is as follows:
\begin{itemize}
\item $\AdS_3\times \Sph^3$;
\item $\AdS_3 \times \RR^3$;
\item $\RR^{2,1}\times \Sph^3$;
\item the symmetric plane wave with metric $g_-$ in \eqref{eq:ppwave}.
\end{itemize}
In the first three cases, $H$ can be any linear combination (with
non-zero coefficients) of the volume forms on the respective $\AdS_3$
and $\Sph^3$ factors. If $H$ is non-zero and self-dual, only the first
and fourth cases above are viable, and we recover precisely the
classification \cite{CFOSchiral,Gutowski:2003rg} of maximally
supersymmetric backgrounds of $(1,0)$ Poincaré supergravity in six
dimensions.

%in \cite{Meessen}

In conclusion, we verify that all of these maximally supersymmetric
backgrounds do indeed admit a Killing superalgebra and that different
backgrounds have different associated Killing superalgebras, in the
sense of filtered deformations.

This paper is organised as follows. In Section~\ref{sec:conventions} we
introduce our six-dimensional spinor conventions, set the notation and
prove a number of algebraic results that we will use in the rest of the
paper. In Section~\ref{sec:complex-1-0-poincare} we introduce the
Spencer cohomology complexes associated to the $(1,0)$ Poincaré
superalgebra $\fp$ and its extension $\fphat$ by the R-symmetry. The
relevant cohomology groups are computed in
Section~\ref{sec:calculation-h2-2fp} for $\fp$ and in
Section~\ref{sec:calc-h2-2hatfp} for $\fphat$. From these calculations
we extract the Killing spinor equations and in
Section~\ref{sec:killing-spinors} we show that, subject to some
additional conditions on the geometric data given by the Spencer
cohomology, these Killing spinors generate a Lie superalgebra. This is
revisited in Section~\ref{sec:killing-superalgebras} in a slightly
different formalism, paying particularly close attention to the case of
gauged R-symmetry. Finally, in Section~\ref{sec:maxim-supersymm-back} we
determine the geometries admitting the maximal number of Killing
spinors. These are then the candidate six-dimensional lorentzian
manifolds on which to construct rigidly supersymmetric theories with
eight real supercharges.

\section{Conventions}
\label{sec:conventions}

Let $(V,\eta)$ be a six-dimensional (``mostly plus'') lorentzian vector
space.  We may choose a $\eta$-orthonormal basis
$(\be_0,\be_1,\dots,\be_5)$ for $V$ relative to which
$\eta(\be_\mu,\be_\nu) = \eta_{\mu\nu} = \diag(-1,+1,\dots,+1)$.

We will let $\flat : V \to V^*$ and $\sharp: V^* \to V$ denote
the musical isomorphisms:
\begin{equation}
  v^\flat(w) = \eta(v,w) \qquad\text{and}\qquad \xi(v) =
  \eta(\xi^\sharp,v).
\end{equation}
It follows, as usual, that $\flat$ and $\sharp$ are mutual inverses.  We
let $\fso(V)$ be the Lie algebra of $\eta$-skew-symmetric
endomorphisms of $V$:
\begin{equation}
  \fso(V) = \left\{L : V \to V ~ \middle | ~ \eta(L v, w) = - \eta(v, Lw)
    \quad\forall v,w \in V\right\}.
\end{equation}
There is a vector space (in fact, an $\fso(V)$-module) isomorphism
$\fso(V) \cong \Lambda^2 V$.  If $L \in \fso(V)$, we define $\omega_L
\in \Lambda^2V$ by
\begin{equation}
  L v = - \iota_{v^\flat}\omega_L.
\end{equation}
Conversely, if $\omega \in \Lambda^2V$, we define $L_\omega \in \fso(V)$
by the same relationship: namely,
\begin{equation}
  L_\omega v = - \iota_{v^\flat} \omega.
\end{equation}
It then follows that these two maps are mutual inverses: $L_{\omega_L} =
L$ and $\omega_{L_\omega} = \omega$.  Relative to the basis $(\be_\mu)$ for
$V$, we find that
\begin{equation}
  \omega_L = \tfrac12 L^{\mu\nu} \be_\mu \wedge
  \be_\nu\qquad\text{where}\qquad L \be_\mu = \be_\nu L^\nu{}_\mu.
\end{equation}

We define the Clifford algebra $\Cl(V)$ by the Clifford relations
\begin{equation}
  v \cdot v = \eta(v,v) \1.
\end{equation}
As a real unital associative algebra, $\Cl(V) \cong \HH(4)$, whereas we
have an isomorphism of Lie groups $\Spin(V) \cong \SL(2,\HH)$.  There
is, up to isomorphism, a unique irreducible Clifford module $\Sigma$
which is quaternionic and of dimension $4$.  We prefer to think of
$\Sigma$ as an $8$-dimensional complex vector space with an invariant
quaternionic structure.  As a representation of $\Spin(V)$ it breaks up
as $\Sigma = \Sigma_+ \oplus \Sigma_-$, where $\Sigma_\pm$ are
irreducible representations: either quaternionic of dimension $2$ or,
equivalently, complex $4$-dimensional with an invariant quaternionic
structure.  Let $\Delta$ denote the fundamental representation of
$\Sp(1)$: it can be thought of as a complex $2$-dimensional
representation with an invariant quaternionic structure.  The tensor
product $\Sigma_+ \otimes_\CC \Delta$ is the complexification of a real
representation of $\fso(V)$ we call $S$.  In other words, $S\otimes\CC =
\Sigma_+ \otimes_{\CC} \Delta$.  In practice we prefer to work with $S
\otimes \CC$; although we will not mention this explicitly.

There is a dual pairing $\left<-,-\right>$ between $\Sigma_+$ and
$\Sigma_-$ relative to which,
\begin{equation}
  \left<v \cdot s_+, s_-\right> = - \left<s_+, v \cdot s_-\right>
\end{equation}
for all $v \in V$ and $s_\pm \in \Sigma_\pm$.  We may extend it to a
\emph{symmetric} inner product on $\Sigma = \Sigma_+ \oplus
\Sigma_-$, also denoted $\left<-,-\right>$, in such a way that
$\Sigma_\pm$ are (maximally) isotropic subspaces.  We will use the
notation $\sbar = \left<s,-\right>$, so that $\sbar_1 s_2 = \left<s_1,
  s_2\right>$.

If we let $\bep_A$, $A = 1,2$, denote a basis for $\Delta$, any $s \in
S$ can be written as $s = s^A \bep_A$;  we will often just work with the
components $s^A \in \Sigma_+$.  On $\Delta$ we have an invariant
symplectic structure $\epsilon$, normalised to $\epsilon_{12} =
\epsilon^{12} = 1$. In \S\ref{sec:killing-spinors}, we will also make
use of the skew-symmetric bilinear form $\left(-,-\right)$ on
$\Sigma\otimes_{\CC}\Delta$ given by the tensor product of
$\left<-,-\right>$ and $\epsilon$.

We use the Northeast convention to raise and lower
indices with $\epsilon$:
\begin{equation}
  u_A = \epsilon_{AB} u^B \qquad\text{and}\qquad u^A = u_B
  \epsilon^{BA},
\end{equation}
from where it follows that $\epsilon_{AB} \epsilon^{AC} = \delta_B^C$.
Since $\Delta$ is $2$-dimensional, $\Lambda^2 \Delta^*$ is one-dimensional
and spanned by $\epsilon$.  Any bivector
$B\in\otimes^2\Delta\cong\otimes^2\Delta^*$ can be decomposed
into $B_{AB} = (B_0)_{AB} + \epsilon_{AB} b$, where $(B_0)_{AB} =
(B_0)_{BA}$ and $b = \tfrac12 B_{AB}\epsilon^{AB}$.  The notation
suggests that if we think of the bivector as an endomorphism of
$\Delta$, then $B^A{}_B = (B_0)^A{}_B + b \delta^A_B$ and $B_0$ is
traceless.

The basis element $\be_\mu$ for $V$ is sent to the endomorphism of $\Sigma$
denoted $\Gamma_\mu$, where
\begin{equation}
  \Gamma_\mu \Gamma_\nu = \Gamma_{\mu\nu} + \eta_{\mu\nu}.
\end{equation}
We define the volume in the Clifford algebra by $\Gamma_7 :=
\Gamma_{012345}$ and the projection operators $P_\pm = \tfrac12( \1 \pm
\Gamma_7)$. Then $\Gamma_7 s^A = s^A$ and hence $P_+ s^A = s^A$.

We can make two kinds of bilinears from $s \in S$: the \emph{Dirac
  current} $\kappa=\kappa(s,s)\in V$ with components
\begin{equation}
  \kappa^\mu = \epsilon_{AB} \sbar^A \Gamma^\mu s^B \iff \sbar^A
  \Gamma^\mu s^B = \tfrac12 \epsilon^{AB} \kappa^\mu,
\end{equation}
and a family $\omega$ of $3$-forms given by
\begin{equation}
\label{eq:fam3}
  \omega_{\mu\nu\rho}{}^{AB} = \sbar^A \Gamma_{\mu\nu\rho} s^B.
\end{equation}

\begin{lemma}
$\omega \in \Lambda^3_+ V \otimes \odot^2 \Delta$.
\end{lemma}

\begin{proof}
  Only the self-duality needs proof.  We calculate
  \begin{equation}
    \begin{split}
      \omega_{\mu\nu\rho}{}^{AB} &= \sbar^A \Gamma_{\mu\nu\rho} s^B\\
          &= \sbar^A \Gamma_{\mu\nu\rho}\Gamma_7  s^B\\
          &= \tfrac16 \varepsilon_{\mu\nu\rho\lambda\sigma\tau}\sbar^A \Gamma^{\lambda\sigma\tau} s^B\\
          &= \tfrac16 \varepsilon_{\mu\nu\rho\lambda\sigma\tau} \omega^{\lambda\sigma\tau\,AB}\\
          &= (\star\omega)_{\mu\nu\rho}{}^{AB}.
    \end{split}
  \end{equation}
\end{proof}

A very useful identity is the \emph{Fierz identity}, which says
\begin{equation}
  \label{eq:Fierz}
  s^A \sbar^B = -\tfrac18 \left(\epsilon^{AB} \kappa + \omega^{AB}\right)P_-,
\end{equation}
for all $s = s^A \bep_A\in S$.
An immediate consequence of this identity is that the Dirac current of
$s$ Clifford-annihilates $s$.

\begin{lemma}
  \label{lem:kappakills}
  $(\sbar^A \Gamma^\mu s^B) \Gamma_\mu s^C = 0$.
\end{lemma}

\begin{proof}
  Since $(\sbar^A \Gamma^\mu s^B) \Gamma_\mu s^C = \tfrac12 \epsilon^{AB}
  \kappa \cdot s^C$, it is enough to show that $\kappa \cdot s^C = 0$.
  We calculate
  \begin{equation}
    \begin{split}
      (\sbar^A \Gamma^\mu s^B) \Gamma_\mu s^C &= \Gamma_\mu (s^C \sbar^A) \Gamma^\mu s^B\\
      &= -\tfrac18 \Gamma_\mu \left( \epsilon^{CA} \kappa +
        \omega^{CA}\right) \Gamma^\mu s^B\\
      &= \tfrac12 \epsilon^{CA} \kappa \cdot s^B,
    \end{split}
  \end{equation}
  where we have used the useful identities
  \begin{equation}
    \label{eq:gammatraces}
    \Gamma^\mu \kappa \Gamma_\mu = -4 \kappa \qquad\text{and}\qquad
    \Gamma^\mu \omega^{AB} \Gamma_\mu = 0.
  \end{equation}
  We now contract both sides with $\epsilon_{AB}$ to arrive at
  \begin{equation}
    \kappa \cdot s^C = \tfrac12 \epsilon^{CA} \epsilon_{AB} \kappa \cdot
    s^B = - \tfrac12 \kappa \cdot s^C,
  \end{equation}
  which shows that $\kappa\cdot s^C = 0$.
\end{proof}

The $3$-form $\omega^{AB}$ associated to $s$ also Clifford-annihilates $s$.  In fact, more
generally, we have the following

\begin{lemma}
  \label{lem:omegakills}
  Let $\Xi \in \Lambda_+^3V$ and $s \in \Sigma_+$.  Then $\Xi \cdot s = 0$.
\end{lemma}

\begin{proof}
  The identity
  \begin{equation}
    \label{eq:CliffHodge}
    \Gamma_{\mu\nu\rho} \Gamma_7 = \tfrac1{3!}
    \varepsilon_{\mu\nu\rho\lambda\sigma\tau}
    \Gamma^{\lambda\sigma\tau},
  \end{equation}
  implies that if $\Xi \in \Lambda^3V$, then $\Xi \Gamma_7 = -
  \star\Xi$, so that if $s \in \Sigma_+$,
\begin{equation}
  \Xi \cdot s = \Xi \Gamma_7 \cdot s = -(\star\Xi) \cdot s
  \implies (\Xi + \star\Xi) \cdot s = 0.
\end{equation}
If $\Xi$ is self-dual, the result follows.
\end{proof}

In the same way one can show that the Clifford product of two self-dual
(or antiself-dual) $3$-forms vanishes.

\begin{lemma}
  \label{lem:wwzero}
  If $\Xi_1,\Xi_2 \in \Lambda^3_\pm V$, then $\Xi_1 \Xi_2 = 0$.
\end{lemma}

\begin{proof}
  The identity \eqref{eq:CliffHodge} says that if $\Xi \in
  \Lambda^3_\pm V$, then $\Xi \Gamma_7 = \mp\Xi$.  Now we calculate
  \begin{equation}
    \Xi_1 \Xi_2 = \Xi_1 \Gamma_7^2 \Xi_2 = (\Xi_1
    \Gamma_7)(- \Xi_2 \Gamma_7) = -\Xi_1 \Xi_2.
  \end{equation}
\end{proof}

The Dirac current and the $3$-form satisfy the following properties.

\begin{lemma}
  \label{lem:linindep}
  Let $\kappa$ and $\omega^{AB}$ be the Dirac current and $3$-form
  associated to a non-zero $s \in S$.  Then
  \begin{enumerate}[label=(\roman*)]
  \item $\kappa$ is a non-vanishing null vector, and
  \item the $3$-forms $\omega^{AB}$, $A,B=1,2$, are linearly independent (in particular they do not vanish).
  \end{enumerate}
\end{lemma}

\begin{proof}
  The Dirac current $\kappa$ is a null vector from Lemma
  \ref{lem:kappakills}. By the Fierz identity \eqref{eq:Fierz} we have
  \begin{equation}
    s^A \sbar^B-s^B \sbar^A = -\tfrac14 \epsilon^{AB} \kappa P_-,
  \end{equation}
  and $s=s^A\bep_A$ is a decomposable element of the tensor product
  $\Sigma_+\otimes_\CC\Delta$ if $\kappa=0$. This is not possible, since
  $s$ is a real spinor. The non-vanishing of any $\omega^{AB}$ is proved similarly.

  Now, making use of the Fierz identity, it is a simple matter to deduce
  the identity
  \begin{equation}
    \label{eq:Omega2IdentityIntro}
    \omega^{AB}_{\mu [ \nu}{}^\tau \omega^{CD}_{\rho\sigma ]\tau} = -\tfrac{1}{3} \kappa_\mu (  \epsilon^{A(C} \omega^{D)B}_{\nu\rho\sigma} + \epsilon^{B(C} \omega^{D)A}_{\nu\rho\sigma} )~.
  \end{equation}
  Consequently,
  \begin{equation}
    \label{eq:Omega2Identity111222Intro}
    \omega^{11}_{\mu [ \nu}{}^\tau \omega^{12}_{\rho\sigma ]\tau} = -\tfrac{1}{3} \kappa_\mu \omega^{11}_{\nu\rho\sigma} \; , \quad \omega^{22}_{\mu [ \nu}{}^\tau \omega^{12}_{\rho\sigma ]\tau} = \tfrac{1}{3} \kappa_\mu \omega^{22}_{\nu\rho\sigma} \; , \quad \omega^{11}_{\mu [ \nu}{}^\tau \omega^{22}_{\rho\sigma ]\tau} = -\tfrac{2}{3} \kappa_\mu \omega^{12}_{\nu\rho\sigma}~,
  \end{equation}
  are all non-vanishing. Now, assume $\omega^{11}$ is a linear
  combination of $\omega^{12}$ and $\omega^{22}$ and substitute the
  former on the LHS of the first equation in
  \eqref{eq:Omega2Identity111222Intro}. Using again
  \eqref{eq:Omega2IdentityIntro} with $A=C=1$ and $B=D=2$, we get that
  $\omega^{11}$ and $\omega^{22}$ are proportional. Plugging this back
  into the last equation in \eqref{eq:Omega2Identity111222Intro} implies
  $\omega^{12}=0$, which is absurd.
\end{proof}

Finally (for now), we have two additional algebraic relations between the
Dirac current and the $3$-form.

\begin{lemma}
  \label{lem:kwzero}
  Let $\kappa$ and $\omega^{AB}$ be the Dirac current and $3$-form
  associated to $s \in S$.  Then
  \begin{enumerate}[label=(\roman*)]
  \item $\iota_\kappa \omega^{AB} = 0$, and
  \item $\kappa^\flat \wedge \omega^{AB} = 0$.
  \end{enumerate}
\end{lemma}

\begin{proof}
  To prove the first identity, we compute
  \begin{equation}
    \begin{split}
      \kappa^\rho \omega_{\mu\nu\rho}{}^{AB} &= \kappa^\rho \sbar^A
      \Gamma_{\mu\nu\rho} s^B\\
      &=  \kappa^\rho \sbar^A \left(\Gamma_{\mu\nu} \Gamma_\rho -
        \eta_{\nu\rho} \Gamma_\mu + \eta_{\mu\rho} \Gamma_\nu\right) s^B\\
      &= \sbar^A \Gamma_{\mu\nu} \kappa \cdot s^B - \kappa_\nu \sbar^A
      \Gamma_\mu s^B + \kappa_\mu \sbar^A \Gamma_\nu s^B.
    \end{split}
  \end{equation}
  The first term vanishes because of Lemma~\ref{lem:kappakills} and the last
  two terms precisely cancel each other.  The second identity follows
  from the first due to the self-duality of $\omega^{AB}$:
  \begin{equation}
    \kappa^\flat \wedge \omega^{AB} = \kappa^\flat \wedge \star
    \omega^{AB} = \star (\iota_\kappa \omega^{AB}) = 0.
  \end{equation}
\end{proof}

\section{Spencer complexes associated to the $(1,0)$ Poincaré superalgebra}
\label{sec:complex-1-0-poincare}

The $d{=}6$ $(1,0)$ Poincaré superalgebra is the $\ZZ$-graded Lie
superalgebra
\begin{equation}
  \label{eq:psa}
  \fp = \fp_{-2} \oplus \fp_{-1} \oplus \fp_0 = V \oplus S \oplus \fso(V),
\end{equation}
with nonzero Lie brackets
\begin{equation}
  \label{eq:10ps}
  [L,M] = LM - ML,\qquad [L,v] = Lv, \qquad [L,s] = \tfrac12 \omega_L
  \cdot s \qquad\text{and}\qquad [s,s] = \kappa,
\end{equation}
for all $L,M \in \fso(V)$, $v\in V$ and $s \in S$.  We will also
consider the extended Poincaré superalgebra $\fphat$ where in degree
zero we have $\fphat_0 = \fso(V) \oplus \fr$, where $\fr \cong \fsp(1)$
is the R-symmetry.

In this section we describe the (generalised) Spencer complexes
associated to $\fp$ and $\fphat$.  As we briefly recall, they govern
filtered subdeformations of these graded Lie superalgebras.

\subsection{Spencer cohomology and filtered deformations}
\label{sec:spencer-fdefs}

Many Lie (super)algebras of geometric origin are generated by sections of
vector bundles on a manifold which are parallel with respect to some
connection.  A classical example is provided by the Lie algebra of
isometries of a riemannian manifold.  This Lie algebra is generated by
Killing vector fields which, as shown in \cite{MR0084825, MR0250643}, are
in one-to-one correspondence with parallel sections of a certain vector
bundle relative to the so-called Killing transport connection.  A
similar construction holds for the Lie algebra of conformal
transformations of a conformal manifold, generated by conformal Killing
vectors, corresponding to parallel sections of a vector bundle relative
to the conformal Killing transport connection \cite{MR0250643}.  A less
classical example is eleven-dimensional supergravity, where the Killing
spinors, which are parallel relative to the connection defined by the
gravitino variation, generate a Lie superalgebra, known as the Killing
superalgebra of the background \cite{FigueroaO'Farrill:2004mx}.

What these Lie (super)algebras have in common is that they are filtered
deformations of a graded subalgebra of the Lie (super)algebra associated
to the ``flat model''.  The flat model depends on the context: it is
euclidean space in riemannian geometry, the round sphere in the
conformal context, and the Minkowski vacuum in the supergravity context.
Indeed, the isometry algebra of euclidean space, the conformal algebra
of the round sphere or the supersymmetry algebra of Minkowski spacetime
are graded Lie (super)algebras.  This means that the underlying vector
space admits a $\ZZ$-grading and that the Lie bracket has degree zero,
so that it respects the grading.  The effect of ``turning on curvature''
is two-fold: firstly, it breaks the symmetry to a graded subalgebra and,
simultaneously, it deforms the subalgebra by introducing terms of
positive degree in the Lie bracket.  This is the algebraic manifestation
of the well-known mantra that translations no longer commute in the
presence of curvature. The resulting Lie (super)algebra is no longer
graded, but since the new terms in the Lie bracket have positive degree,
it is now filtered.

Deformations of algebraic structures, such as Lie (super)algebras, are
typically governed by a cohomology theory.  In the case of Lie
(super)algebras, it is the cohomology of the Chevalley--Eilenberg
complex of the Lie superalgebra with coefficients in the adjoint module
\cite{MR0024908,MR0422372,MR871615}.  In the case of a a graded Lie
(super)algebra, the Chevalley--Eilenberg differential has zero degree
and hence the complex splits in the direct sum of sub-complexes labelled
by the degree.  In studying filtered deformations of graded Lie
(super)algebras, we are interested in deforming the Lie bracket by terms
of positive degree.  Moreover, for graded Lie (super)algebras which are
zero in positive degree (such as the Poincaré superalgebra
\eqref{eq:psa}), we may pass to the subcomplex relative to the
degree-zero subalgebra.

A first step in this process is the calculation of the cohomology of the 
positive-degree subcomplex of the Chevalley--Eilenberg complex of the 
negative-degree subalgebra of the flat model with coefficients in the
adjoint module.  In other words, in the present context, we consider the
subcomplex of the Chevalley--Eilenberg complex $C^\bullet(\fg_-,\fg)$ of
degree $d>0$, which we denote by $C^{d,\bullet}(\fg_-,\fg)$, where $\fg$
stands for either $\fp$ or $\fphat$, and which are defined below.

\subsection{Spencer complex of $\fp$}
\label{sec:spencer-complex-fp}

In the first instance, we will calculate the cohomology of the Spencer complex
\begin{equation}
  \label{eq:complex}
  \begin{CD}
    C^{2,1}(\fp_-,\fp) @>\partial>> C^{2,2}(\fp_-,\fp) @>\partial>>
    C^{2,3}(\fp_-,\fp),
  \end{CD}
\end{equation}
where the spaces of cochains are
\begin{equation}
  \begin{split}
    C^{2,1}(\fp_-,\fp) &= \Hom(V,\fso(V))\\
    C^{2,2}(\fp_-,\fp) &= \Hom(\Lambda^2V,V) \oplus \Hom(V\otimes S, S)
    \oplus \Hom(\odot^2S, \fso(V))\\
    C^{2,3}(\fp_-,\fp) &= \Hom(V \otimes \odot^2S, V) \oplus
    \Hom(\odot^3 S, S)
  \end{split}
\end{equation}
and the differentials are such that if $\lambda \in C^{2,1}(\fp_-,\fp)$, then
\begin{equation}
  \partial\lambda(v,w) = \lambda_v w - \lambda_w v,
  \qquad \partial\lambda(v,s) = \tfrac12 \omega_{\lambda_v} \cdot s
  \qquad\text{and}\qquad
  \partial\lambda(s,s) = - \lambda_{[s,s]},
\end{equation}
and if $\psi = \alpha + \beta + \gamma \in C^{2,2}(\fp_-,\fp)$ with
\begin{equation}
  \alpha : \Lambda^2 V \to V, \qquad \beta: V \otimes S \to S
  \qquad\text{and}\qquad
  \gamma: \odot^2 S \to \fso(V),
\end{equation}
then
\begin{equation}
  \partial\psi(v,s,s) = [v,\gamma(s,s)] - 2 [s,\beta_v s] -
  \alpha([s,s],v) \qquad\text{and}\qquad
  \tfrac13 \partial\psi(s,s,s) = [s,\gamma(s,s)] - \beta_{[s,s]}s.
\end{equation}

\begin{lemma}
\label{lem:h21}
  The Spencer differential $\partial: C^{2,1}(\fp_-,\fp) \to
  C^{2,2}(\fp_-,\fp)$ is injective, so that in particular
  $H^{2,1}(\fp_-,\fp) = 0$.
\end{lemma}

\begin{proof}
  Let $\lambda \in C^{2,1}(\fp_-,\fp)$.  If $\partial\lambda = 0$, then
  in particular $\partial\lambda(s,s) = - \lambda_{[s,s]} = 0$.  Since
  $[S,S]=V$, then $\lambda = 0$.
\end{proof}

Moreover it is not just that
$\partial: C^{2,1}(\fp_-,\fp) \to C^{2,2}(\fp_-,\fp)$ is injective, but
that the component $C^{2,1}(\fp_-,\fp) \to \Hom(\Lambda^2 V, V)$ is an
isomorphism.  This follows because if $\partial\lambda(v,w) = 0$, the
tensor $T(u,v,w) := \eta(u,\lambda_v w )$ satisfies
$T(u,v,w) = T(u,w,v)$ in addition to $T(u,v,w) = - T(w,v,u)$, so that it
is identically zero.

This fact allows us to compute the cohomology by determining the kernel
of the Spencer differential on the normalised cochains satisfying
$\alpha = 0$.  In other words, we have the following

\begin{proposition}
  \label{prop:normalisation}
  There is an isomorphism (of $\fso(V)$-modules)
  \begin{equation}
    H^{2,2}(\fp_-,\fp) \cong \left\{\beta + \gamma \in C^{2,2}(\fp_-,\fp)
      \middle | \partial(\beta + \gamma) = 0\right\}.
  \end{equation}
\end{proposition}

In Section~\ref{sec:calculation-h2-2fp} we calculate this cohomology.

\subsection{Spencer complex of $\fphat$}
\label{sec:spenc-compl-extended}

We will also calculate the cohomology of the Spencer complex
\begin{equation}
  \label{eq:complexR}
  \begin{CD}
    C^{2,1}(\fphat_-,\fphat) @>\partial>> C^{2,2}(\fphat_-,\fphat) @>\partial>>
    C^{2,3}(\fphat_-,\fphat),
  \end{CD}
\end{equation}
where the spaces of cochains are
\begin{equation}
  \begin{split}
    C^{2,1}(\fphat_-,\fphat) &= \Hom(V,\fso(V)) \oplus \Hom(V,\fr)\\
    C^{2,2}(\fphat_-,\fphat) &= \Hom(\Lambda^2V,V) \oplus \Hom(V\otimes S, S)
    \oplus \Hom(\odot^2S, \fso(V)) \oplus \Hom(\odot^2S,\fr)\\
    C^{2,3}(\fphat_-,\fphat) &= \Hom(V \otimes \odot^2S, V) \oplus
    \Hom(\odot^3 S, S)
  \end{split}
\end{equation}
and the differentials are such that if $\phi = \lambda + \mu \in
C^{2,1}(\fphat_-,\fphat)$, with $\lambda: V \to \fso(V)$ and $\mu : V
\to \fr$, then
\begin{equation}
  \partial\phi(v,w) = \lambda_v w - \lambda_w v,
  \qquad \partial\phi(v,s) = \tfrac12 \omega_{\lambda_v} \cdot s + \mu_v(s)
  \qquad\text{and}\qquad
  \partial\phi(s,s) = - \lambda_{[s,s]} - \mu_{[s,s]},
\end{equation}
and if $\psi = \alpha + \beta + \gamma + \rho\in C^{2,2}(\fphat_-,\fphat)$ with
\begin{equation}
  \alpha : \Lambda^2 V \to V, \qquad \beta: V \otimes S \to S, \qquad  \gamma: \odot^2 S \to \fso(V)
  \qquad\text{and}\qquad
  \rho: \odot^2 S \to \fr,
\end{equation}
then
\begin{equation}
  \partial\psi(v,s,s) = [v,\gamma(s,s)] - 2 [s,\beta_v s] - \alpha([s,s],v)
 \quad\text{and}\quad
  \tfrac13 \partial\psi(s,s,s) = [s,\gamma(s,s)] + [s,\rho(s,s)] - \beta_{[s,s]}s.
\end{equation}

As before, we see that $\partial: C^{2,1}(\fphat_-,\fphat) \to
C^{2,2}(\fphat_-,\fphat)$ is injective and hence we calculate
$H^{2,2}(\fphat_-,\fphat)$ by calculating the kernel of the Spencer
differential on the space of normalised cochains in
$C^{2,2}(\fphat_-,\fphat)$: those which have $\alpha = 0$ and for
which
\begin{equation}
  \rho(s,s)^A{}_B = \tfrac16 \omega_{\mu\nu\rho}{}^{CD}
  \rho^{\mu\nu\rho}{}_{CD}{}^A{}_B,
\end{equation}
where $\rho\in\Lambda^3_-V\otimes\odot^2 \Delta\otimes\fr$.

In Section~\ref{sec:calc-h2-2hatfp} we calculate this cohomology.

\section{Calculation of $H^{2,2}(\fp_-,\fp)$}
\label{sec:calculation-h2-2fp}

We now compute the Spencer cohomology group $H^{2,2}(\fp_-,\fp)$
corresponding to the unextended $(1,0)$ Poincaré superalgebra.

By Proposition~\ref{prop:normalisation}, the Spencer cohomology
$H^{2,2}(\fp_-,\fp)$ is isomorphic to the solutions $\beta : V \otimes S
\to S$ and $\gamma : \odot^2 S \to \fso(V)$ of the following two
cocycle conditions:
\begin{align}
  \gamma(s,s) v + 2 [s, \beta_v s] &= 0 \label{eq:co1}\\
  \tfrac12 \omega_{\gamma(s,s)} s + \beta_{[s,s]} s &= 0. \label{eq:co2}
\end{align}
The first equation will determine $\gamma$ in terms of $\beta$, so that
the actual variables are the coefficients of $\beta$.  It pays to
understand this space and to label its components.  First of all, we may
view $\beta$ as a map $V \to \End(S)$, sending $v \in V$ to $\beta_v \in
\End(S)$, where (for $v = \be_\mu$)
\begin{equation}
  (\beta_\mu s)^B = \beta_\mu^{(0)\,B}{}_C s^C + \tfrac12
  \beta^{(2)}_{\mu\rho\sigma}{}^B{}_C \Gamma^{\rho\sigma} s^C.
\end{equation}
We find it convenient to rename the different components of $\beta$:
\begin{equation}
  \begin{split}
    A_\mu \delta^A{}_B  &:  V \to \Lambda^0 V \otimes \Lambda^2 \Delta \\
    C_\mu{}^A{}_B  &:  V \to \Lambda^0 V \otimes \odot^2 \Delta \\
    H_{\mu\rho\sigma} \delta^A{}_B &: V \to \Lambda^2V \otimes \Lambda^2 \Delta\\
    G_{\mu\rho\sigma}{}^A{}_B &: V \to \Lambda^2V \otimes \odot^2 \Delta,
  \end{split}
\end{equation}
so that
\begin{equation}
  (\beta_\mu s)^B = A_\mu s^B + C_\mu{}^B{}_C s^C + \tfrac12
  H_{\mu\rho\sigma} \Gamma^{\rho\sigma} s^B + \tfrac12
  G_{\mu\rho\sigma}{}^B{}_C \Gamma^{\rho\sigma} s^C.
\end{equation}

\subsection{Solving the first cocycle condition}
\label{sec:solv-first-cocycle}

We take the inner product of $\be_\mu$ with the first cocycle equation \eqref{eq:co1}
applied to $v= \be_\nu$ and obtain
\begin{equation}
  0 = \gamma(s,s)_{\mu\nu} + 2 \epsilon_{AB} \sbar^A \Gamma_\mu (\beta_\nu
  s)^B = \gamma(s,s)_{\mu\nu} +2 \kappa_\mu A_\nu + 2 \kappa^\rho
  H_{\nu\mu\rho} + G_{\nu\rho\sigma\,AB} \omega_\mu{}^{\rho\sigma\,AB}.
\end{equation}
The skew-symmetric part gives $\gamma(s,s)_{\mu\nu}$, whereas the
symmetric part gives an equation for $\beta$:
\begin{equation}
  \kappa_\mu A_\nu + \kappa_\nu A_\mu + \kappa^\rho H_{\nu\mu\rho} + \kappa^\rho H_{\mu\nu\rho} + \tfrac12 G_{\nu\rho\sigma}
  \omega_\mu{}^{\rho\sigma} + \tfrac12 G_{\mu\rho\sigma} \omega_\nu{}^{\rho\sigma}= 0,
\end{equation}
where we have suppressed the $\odot^2\Delta$ indices in $G$ and
$\omega$.  This equation is true for all $s$ and hence the terms in the
two independent bilinears (the Dirac current $\kappa$ and the family
$\omega$ of self-dual $3$-forms) are separately zero, giving two
equations:
\begin{align}
  \kappa^\rho (\eta_{\rho\mu} A_\nu + \eta_{\rho\nu} A_\mu +
  H_{\mu\nu\rho} + H_{\nu\mu\rho}) &= 0 \label{eq:coc1k}\\
  \omega^{\rho\sigma\tau} (\eta_{\tau\mu} G_{\nu\rho\sigma} +
  \eta_{\tau\nu} G_{\mu\rho\sigma} ) &= 0. \label{eq:coc1w}
\end{align}
Abstracting $\kappa$ from the first equation and contracting first with
$\eta^{\mu\nu}$ and then with $\eta^{\nu\rho}$ one finds that $A = 0$
and plugging that back into the equation one finds that
$H \in \Lambda^3 V$.  Since $\omega$ is self-dual, it is only the
antiself-dual projection of
$\eta_{\tau\mu} G_{\nu\rho\sigma} + \eta_{\tau\nu} G_{\mu\rho\sigma}$
which must vanish, yielding the equation
\begin{equation}
  \eta_{\tau\nu} G_{\mu\rho\sigma} + \eta_{\tau\mu} G_{\nu\rho\sigma} + 
  \eta_{\rho\nu} G_{\mu\sigma\tau} + \eta_{\rho\mu} G_{\nu\sigma\tau} + 
  \eta_{\sigma\nu} G_{\mu\tau\rho} + \eta_{\sigma\mu} G_{\nu\tau\rho} =
  \tfrac12 \varepsilon_{\tau\rho\sigma\nu}{}^{\phi\psi} G_{\mu\phi\psi}
  + \tfrac12 \varepsilon_{\tau\rho\sigma\mu}{}^{\phi\psi} G_{\nu\phi\psi}.
\end{equation}
Contracting with $\eta^{\mu\nu}$, we find that
\begin{equation}
  G_{[\rho\sigma\tau]} = \tfrac16
  \varepsilon_{\rho\sigma\tau}{}^{\lambda\mu\nu} G_{\lambda\mu\nu}
  \implies G_{[\rho\sigma\tau]} \in \Lambda^3_+ V,
\end{equation}
whereas contracting with $\eta^{\nu\sigma}$ one finds
\begin{equation}
  -5 G_{\mu\rho\tau} + \eta_{\rho\mu} G_\sigma{}^\sigma{}_\tau -
  \eta_{\tau\mu} G_\sigma{}^\sigma{}_\rho = 3 G_{[\mu\rho\tau]}.
\end{equation}
Skew-symmetrising, one finds that $G_{[\mu\rho\tau]} = 0$ and hence
\begin{equation}
  G_{\mu\rho\tau} = \eta_{\rho\mu}\varphi_\tau - \eta_{\tau\mu} \varphi_\rho,
\end{equation}
where $\varphi_\rho = \tfrac15 G_\sigma{}^\sigma{}_\rho$.  Plugging this
back into the equation~\eqref{eq:coc1w} for the family of self-dual
$3$-forms, we see that it is identically satisfied. We arrived at the
following

\begin{proposition}
  \label{prop:solfcc}
  The solution of the first cocycle equation is
  \begin{equation}
    \label{eq:solfcc}
    \begin{split}
      (\beta_\mu s)^A &= C_\mu{}^A{}_B s^B + \tfrac12
      H_{\mu\rho\sigma}\Gamma^{\rho\sigma} s^A + \varphi^{\rho\,A}{}_B
      \Gamma_{\mu\rho} s^B\\
      \gamma(s,s)_{\mu\nu} &= 2 \kappa^\rho H_{\rho\mu\nu} - 2
      \varphi^\rho{}_{AB} \omega_{\mu\nu\rho}{}^{AB},
    \end{split}
  \end{equation}
  for some $H\in\Lambda^3 V$ and $\varphi\in V\otimes \odot^2\Delta$.
\end{proposition}

\subsection{Solving the second cocycle condition}
\label{sec:solv-second-cocycle}

We now consider the second cocycle condition~\eqref{eq:co2}. Using that the Dirac current $\kappa$ Clifford
annihilates $s$ (Lemma~\ref{lem:kappakills}), we may
rewrite this condition as follows:
\begin{equation}
  \label{eq:co2too}
  \kappa^\rho E_\rho{}^C{}_D s^D + \tfrac16 \omega_{\mu\nu\rho}{}^{AB}
  \Omega^{\mu\nu\rho}{}_{AB} s^C = 0,
\end{equation}
where
\begin{equation}
  \begin{split}
    E_\rho{}^C{}_D &:= H_{\rho\mu\nu} \Gamma^{\mu\nu} \delta^C{}_D +
    C_\rho{}^C{}_D + \varphi_\rho{}^C{}_D\\
    \Omega^{\mu\nu\rho}{}_{AB} &:= - 3 \varphi^{[\rho}{}_{AB}\Gamma^{\mu\nu]^-}.
  \end{split}
\end{equation}
Note that we are taking the antiself-dual projection in the RHS of the
last equation, that is, we have just introduced a family $\Omega$ of
antiself-dual $3$-forms.

We now polarise equation~\eqref{eq:co2too}
\begin{equation}
  \epsilon_{AB} \sbar_1^A \Gamma^\rho s_2^B E_\rho{}^C{}_D s_3^D +
  \tfrac16 \sbar_1^A \Gamma_{\mu\nu\rho} s_2^B
  \Omega^{\mu\nu\rho}{}_{AB} s_3^C + \text{cyclic} =  0,
\end{equation}
set $s_1 = s_2 = s$ and rearrange to arrive at
\begin{equation}
  \left( \kappa^\rho E_\rho{}^C{}_D + \tfrac16
    \omega_{\mu\nu\rho}{}^{AB} \Omega^{\mu\nu\rho}{}_{AB} \delta^C{}_D + 2
    \epsilon_{AD} E_\rho{}^C{}_B s^B \sbar^A \Gamma^\rho + \tfrac13
    \Omega^{\mu\nu\rho}{}_{AD} s^C\sbar^A \Gamma_{\mu\nu\rho} \right) s_3^D = 0.
\end{equation}
We may abstract $s_3$, keeping in mind that $\Gamma_7 s_3 = s_3$, and
use the Fierz identity~\eqref{eq:Fierz} to arrive at
\begin{multline}
  \label{eq:coc2bis}
  \left( \kappa^\rho E_\rho{}^C{}_D + \tfrac16
    \omega_{\mu\nu\rho}{}^{AB} \Omega^{\mu\nu\rho}{}_{AB} \delta^C{}_D +
    \tfrac14 E_\rho{}^C{}_D \kappa \Gamma^\rho + \tfrac14 E_\rho{}^C{}_B
    \omega^B{}_D \Gamma^\rho \right. \\
  \left. + \tfrac1{24} \Omega^{\mu\nu\rho\,C}{}_D \kappa
  \Gamma_{\mu\nu\rho}- \tfrac1{24} \Omega^{\mu\nu\rho}{}_{AD}
  \omega^{CA} \Gamma_{\mu\nu\rho}\right) P_+ = 0.
\end{multline}
The terms in $\kappa$ and $\omega$ must vanish separately, since this
expression is true for all $s \in S$.  The $\kappa$ terms give
\begin{equation}
  \kappa^\sigma \left( E_\sigma{}^C{}_D + \tfrac14 E_\rho{}^C{}_D
    \Gamma_\sigma \Gamma^\rho + \tfrac1{24}
    \Omega^{\mu\nu\rho\,C}{}_D \Gamma_\sigma \Gamma_{\mu\nu\rho}
  \right) P_+ = 0,
\end{equation}
which, abstracting $\kappa$, substituting for $E$ and $\Omega$ and
simplifying, reduces to
\begin{equation}
  \left( H_{\sigma\mu\nu} \Gamma^{\mu\nu} \delta{}^C{}_D +
    (C + 3\varphi)_\sigma{}^C{}_D + \tfrac14
    H^{\mu\nu\rho}\Gamma_{\mu\nu} \Gamma_\sigma \Gamma_\rho
    \delta^C{}_D + \tfrac14 ( C + 3\varphi)^{\rho\,C}{}_D \Gamma_\sigma
    \Gamma_\rho \right) P_+ = 0.
\end{equation}
The terms in $\Lambda^2\Delta$ and in $\odot^2\Delta$ vanish separately,
yielding the following two equations:
\begin{align}
  \left( H_{\sigma\mu\nu} \Gamma^{\mu\nu} + \tfrac14 H^{\mu\nu\rho}
  \Gamma_{\mu\nu} \Gamma_\sigma \Gamma_\rho \right) P_+ &= 0 \label{eq:cc2kl}\\
  \left( C + 3 \varphi \right)^\sigma{}^C{}_D (\eta_{\rho\sigma} +
  \tfrac14 \Gamma_\rho \Gamma_\sigma) P_+ &= 0. \label{eq:cc2ks}
\end{align}
Simplifying the first equation we arrive at
\begin{equation}
  \tfrac34 \left( H_{\mu\nu\rho} - (\star H)_{\mu\nu\rho}\right)
  \Gamma^{\mu\nu} P_+ = 0 \implies H \in \Lambda^3_+ V,
\end{equation}
whereas simplifying the second equation (and omitting the
$\odot^2\Delta$ indices) we arrive at
\begin{equation}
  \tfrac14 \left( C^\sigma + 3 \varphi^\sigma \right) \left( 5
    \eta_{\rho\sigma} + \Gamma_{\rho\sigma} \right) P_+ = 0 \implies C =
  - 3 \varphi.
\end{equation}

It remains to consider the $\omega$ terms in
equation~\eqref{eq:coc2bis}, but before doing so we make the following
observation.  Since $H$ is self-dual, Lemma~\ref{lem:omegakills} says
that $H \cdot s = 0$.  Similarly, Lemma~\ref{lem:kappakills} says that
$\kappa \cdot s = 0$, hence
\begin{equation}
  \kappa^\rho H_{\rho\mu\nu} \Gamma^{\mu\nu} s = (H \cdot \kappa +
  \kappa \cdot H) \cdot s = 0.
\end{equation}
Comparing with the second cocycle condition~\eqref{eq:co2too}, we notice
that $H$ drops out of \eqref{eq:co2too} and we may conclude
that the $H$-dependent terms in the $\omega$-dependent terms in
equation~\eqref{eq:coc2bis} are identically zero.\footnote{This also follows
combinatorially by employing identities which follow from
Lemma~\ref{lem:wwzero}.}

The remaining $\omega$-dependent terms in equation~\eqref{eq:coc2bis}
are given by
\begin{equation}
  \label{eq:omegatoo}
  \left( \tfrac16 \omega_{\mu\nu\rho}{}^{AB} \Omega^{\mu\nu\rho}{}_{AB}
    \delta^C{}_D + \tfrac14 E_\rho{}^C{}_B \omega^B{}_D \Gamma^\rho -
    \tfrac1{24} \Omega^{\mu\nu\rho}{}_{AD} \omega^{CA} \Gamma_{\mu\nu\rho}\right) P_+ = 0,
\end{equation}
where now
\begin{equation}
  E_\rho{}^C{}_D = -2 \varphi_\rho{}^C{}_D \qquad\text{and}\qquad
  \Omega^{\mu\nu\rho}{}_{AB} = - 3 \varphi^{[\rho}{}_{AB} \Gamma^{\mu\nu]^-}.
\end{equation}
Making use of the Clifford identity
\begin{equation}
  \Gamma^{\mu\nu} \omega^{CA} \Gamma_{\mu\nu\rho} = 4 \omega^{CA} \Gamma_\rho,
\end{equation}
we simplify equation~\eqref{eq:omegatoo}:
\begin{equation}
  -\tfrac12 \left( \omega_{\mu\nu\rho}{}^{AB} \varphi^\rho{}_{AB}
    \Gamma^{\mu\nu} \delta^C{}_D + \varphi^{\rho\,C}{}_B \omega^B{}_D
    \Gamma_\rho - \varphi^\rho{}_{AD} \omega^{CA} \Gamma_\rho \right)
  P_+ = 0.
\end{equation}
Since $\omega P_+ = 0$ (Lemma~\ref{lem:omegakills}), we may replace
$\omega \Gamma_\rho$ by the anticommutator
\begin{equation}
  \omega \Gamma_\rho + \Gamma_\rho \omega = \omega_{\rho\mu\nu}
  \Gamma^{\mu\nu},
\end{equation}
resulting in the equation
\begin{equation}
  -\tfrac12 \left( \omega_{\mu\nu\rho}{}^{AB} \varphi^\rho{}_{AB}
    \delta^C{}_D + \varphi^{\rho\,C}{}_B \omega_{\rho\mu\nu}{}^B{}_D -
    \varphi^\rho{}_{AD} \omega_{\rho\mu\nu}{}^{CA} \right) \Gamma^{\mu\nu}
  P_+ = 0.
\end{equation}
The representation of $\fso(V)$ on $\Sigma_+$ is faithful, so we may
drop the $\Gamma^{\mu\nu}P_+$ and taking out some common factors, we
arrive at
\begin{equation}
  \omega_{\mu\nu\rho}{}^{AB} \varphi^\rho{}_{EF} \left( \delta^E{}_A
    \delta^F{}_B \delta^C{}_D + \delta^F{}_B \epsilon^{EC} \epsilon_{DA}
  - \delta^E{}_A \delta^F{}_D \delta^C{}_B\right) = 0,
\end{equation}
which can be seen to be identically zero using the identity
\begin{equation}
  \epsilon^{EC} \epsilon_{DA} = \delta^E{}_D \delta^C{}_A - \delta^E{}_A \delta^C{}_D.
\end{equation}

In summary, we have proved the following

\begin{theorem}
  \label{thm:spencer}
  There is an isomorphism of representations of $\fso(V) \oplus \fsp(1)$
  \begin{equation}
    H^{2,2}(\fp_-,\fp) \cong \left(\Lambda^3_+ V \otimes \Lambda^2
      \Delta \right) \oplus \left(V \otimes \odot^2\Delta\right)~.
  \end{equation}
  The cohomology class corresponding to elements $H \in \Lambda^3_+V$ and $\varphi
  \in V \otimes \odot^2\Delta$ is represented by the cocycle $\beta +
  \gamma \in C^{2,2}(\fp_-,\fp)$, where
  \begin{equation}
    \begin{split}
      (\beta_\rho s)^A &= \tfrac12 H_{\rho\mu\nu} \Gamma^{\mu\nu} s^A -
      3 \varphi_\rho{}^A{}_B s^B + \varphi^{\sigma\,A}{}_B
      \Gamma_{\rho\sigma} s^B,\\
      \gamma(s,s)_{\mu\nu} &= 2 \kappa^\rho H_{\rho\mu\nu} - 2
      \varphi^\rho{}_{AB} \omega_{\rho\mu\nu}{}^{AB}.
    \end{split}
  \end{equation}
\end{theorem}

\section{Calculation of $H^{2,2}(\fphat_-,\fphat)$}
\label{sec:calc-h2-2hatfp}

We now compute the Spencer cohomology group $H^{2,2}(\fphat_-,\fphat)$
for the $(1,0)$ Poincaré superalgebra extended by the R-symmetry.  The
first cocycle condition does not change by the introduction of the
R-symmetry and hence we can re-use the results of the previous
calculation (see Proposition~\ref{prop:solfcc}) and go directly to solving
the second cocycle condition:
\begin{equation}
  \label{eq:sccR}
  \tfrac14 \gamma(s,s)_{\mu\nu}\Gamma^{\mu\nu} s^A + \tfrac16 \omega_{\mu\nu\rho}{}^{CD}
  \rho^{\mu\nu\rho}{}_{CD}{}^A{}_B s^B + \kappa^\mu (\beta_\mu s)^A = 0.
\end{equation}

We again write it as
\begin{equation}
  \label{eq:co2tooR}
  \kappa^\rho E_\rho{}^A{}_B s^B + \tfrac16 \omega_{\mu\nu\rho}{}^{CD}
  \Omega^{\mu\nu\rho}{}_{CD}{}^A{}_B s^B = 0,
\end{equation}
where
\begin{equation}
  \begin{split}
    E_\rho{}^A{}_B &:= H_{\rho\mu\nu} \Gamma^{\mu\nu} \delta^A{}_B +
    C_\rho{}^A{}_B + \varphi^{\sigma\, A}{}_B \Gamma_{\rho\sigma}\\
    \Omega^{\mu\nu\rho}{}_{CD}{}^A{}_B &:=
    \rho^{\mu\nu\rho}{}_{CD}{}^A{}_B - 3
    \varphi^{[\rho}{}_{CD}\Gamma^{\mu\nu]^-} \delta^A{}_B.
  \end{split}
\end{equation}
Here $\rho\in\Lambda^3_-V\otimes\odot^2 \Delta\otimes\fr$ and in the last term of the RHS of the last equation, we are taking the antiself-dual projection.

Following the same procedure as in
Section~\ref{sec:solv-second-cocycle}, we polarise and arrive at two
equations for endomorphisms: one for the $\kappa$-dependent terms and
one for the $\omega$-dependent terms:
\begin{align}
  \left(\kappa^\rho E_\rho{}^A{}_B + \tfrac14 E_\rho{}^A{}_B \kappa
  \Gamma^\rho - \tfrac1{24} \Omega^{\mu\nu\rho}{}_{CB}{}^{AC} \kappa
  \Gamma_{\mu\nu\rho}\right) P_+ &= 0 \label{eq:k-scc-R}\\
  \left( \tfrac16 \omega_{\mu\nu\rho}{}^{CD}
  \Omega^{\mu\nu\rho}{}_{CD}{}^A{}_B + \tfrac14 E_\rho{}^A{}_D
  \omega^D{}_B \Gamma^\rho - \tfrac1{24}
  \Omega^{\mu\nu\rho}{}_{CB}{}^A{}_D \omega^{DC}
  \Gamma_{\mu\nu\rho}\right) P_+ &= 0. \label{eq:w-scc-R}
\end{align}
Abstracting $\kappa$ from the first equation we arrive at
\begin{equation}
  \label{eq:k-scc-R-bis}
  \left( 5 E_{\rho}{}^A{}_B + E^{\sigma\,A}{}_B
  \Gamma_{\rho\sigma} - \tfrac16 \Omega^{\mu\nu\sigma}{}_{CB}{}^{AC} \Gamma_\rho
  \Gamma_{\mu\nu\sigma}\right) P_+ = 0.
\end{equation}
Substituting for $E$ and $\Omega$ and simplifying we end up with
\begin{equation}
\label{eq:k-scc-R-tris}
  \left(3 \left(H_{\rho\mu\nu} - \widetilde H_{\rho\mu\nu}\right)
  \Gamma^{\mu\nu} \delta^A{}_B - \rho_{\rho\mu\nu\,CB}{}^{AC}
  \Gamma^{\mu\nu}+ 5 (C + 3 \varphi)_\rho{}^A{}_B + (C
  + 3\varphi)^{\sigma\, A}{}_B \Gamma_{\rho\sigma} \right) P_+ = 0,
\end{equation}
where we used the notation $\widetilde H=\star H$.
It is convenient to decompose $\rho$ into its irreducible components
relative to the R-symmetry.  Lowering indices, we take $\rho
\in \Lambda^3_- V \otimes \odot^2\Delta \otimes \odot^2\Delta$. Omitting
the $V$ indices but not the $\Delta$ indices, we have
\begin{equation}
  \label{eq:rho-decomp}
  \rho_{AB\,CD} = \xi_{ABCD} + (\zeta_{AC} \epsilon_{BD} + \zeta_{BC}
  \epsilon_{AD} + \zeta_{AD} \epsilon_{BC} + \zeta_{BD}
  \epsilon_{AC}) + \theta (\epsilon_{AC} \epsilon_{BD} + \epsilon_{BC}
  \epsilon_{AD}),
\end{equation}
where $\xi \in \Lambda^3_- V \otimes \odot^4\Delta$, $\zeta \in
\Lambda^3_-V \otimes \odot^2 \Delta$ and $\theta \in \Lambda^3_- V$.
It follows that
\begin{equation}
  \rho_{\rho\mu\nu\,CB}{}^{AC} \Gamma^{\mu\nu} =
  -4\zeta_{\rho\mu\nu}{}^A{}_B \Gamma^{\mu\nu} + 3
  \theta_{\rho\mu\nu} \delta^A{}_B \Gamma^{\mu\nu}.
\end{equation}
Plugging this into equation~\eqref{eq:k-scc-R-tris}, and separating the
equation into terms of different types under the R-symmetry and
$\fso(V)$, we arrive at
\begin{equation}
  C = - 3 \varphi, \qquad \zeta = 0 \qquad\text{and}\qquad H -
  \widetilde H = \theta. 
\end{equation}
It is convenient to decompose $H = H^+ + H^-$ into self-dual
and antiself-dual parts. Clearly $H^-=\tfrac12 \theta$.  As observed in
Section~\ref{sec:solv-second-cocycle}, the self-dual part $H^+$ of $H$
does not enter the second cocycle condition, so when we solve
equation~\eqref{eq:w-scc-R}, we may replace $H$ by $\tfrac12 \theta$.
The $\varphi$-dependent terms in the equation are just as in
Section~\ref{sec:solv-second-cocycle} and, as shown there, they cancel
identically.  This leaves an equation only for $\rho$:
\begin{equation}
  \label{eq:scc-final}
  \tfrac16 \omega_{\mu\nu\rho}{}^{CD} \left(
    \rho^{\mu\nu\rho}{}_{CD}{}^A{}_B + \tfrac34 \theta^{\rho\sigma\tau}
    \delta^A{}_D \epsilon_{BC} \Gamma_{\sigma\tau} \Gamma^{\mu\nu} -
    \tfrac1{24} \rho^{\lambda\sigma\tau}{}_{CB}{}^A{}_D
    \Gamma^{\mu\nu\rho}\Gamma_{\lambda\sigma\tau}\right) P_+ = 0.
\end{equation}
Using the decomposition~\eqref{eq:rho-decomp} of $\rho$ and the fact
that $\zeta = 0$, we may rewrite this equation in terms of $\xi$ and
$\theta$:
\begin{multline*}
  \tfrac16 \omega_{\mu\nu\rho}{}^{CD} \left(
    \xi^{\mu\nu\rho}{}_{BCD}{}^A + 2 \theta^{\mu\nu\rho} \epsilon_{BC}\delta^A{}_D + \tfrac34 \theta^{\rho\sigma\tau}
     \epsilon_{BC} \delta^A{}_D \Gamma_{\sigma\tau} \Gamma^{\mu\nu} \right.\\ -
     \left. \tfrac1{24} \xi^{\lambda\sigma\tau}{}_{BCD}{}^A
     \Gamma^{\mu\nu\rho}\Gamma_{\lambda\sigma\tau} +
     \tfrac1{24}\theta^{\lambda\sigma\tau} \epsilon_{BC}\delta^A{}_D
     \Gamma^{\mu\nu\rho}\Gamma_{\lambda\sigma\tau}\right) P_+ = 0.
\end{multline*}
Breaking up into different types under the R-symmetry, we arrive at two separate equations, one for
$\theta$ and one for $\xi$:
\begin{equation}
  \begin{split}
  \tfrac16 \omega_{\mu\nu\rho}{}^{CD} \left(
    \xi^{\mu\nu\rho}{}_{BCD}{}^A - \tfrac1{24} \xi_{\lambda\sigma\tau\, BCD}{}^A
     \Gamma^{\mu\nu\rho}\Gamma^{\lambda\sigma\tau} \right) P_+ &= 0,\\
  \tfrac16 \omega_{\mu\nu\rho}{}^A{}_B \left(
    2 \theta^{\mu\nu\rho} + \tfrac34 \theta^\rho{}_{\sigma\tau} \Gamma^{\sigma\tau} \Gamma^{\mu\nu} +
     \tfrac1{24}\theta_{\lambda\sigma\tau}
     \Gamma^{\mu\nu\rho}\Gamma^{\lambda\sigma\tau}\right) P_+ &= 0.
  \end{split}
\end{equation}
Each of these equations have terms in the $\Lambda^0 V$ and $\Lambda^2V$
components of $\End(\Sigma_+)$, which must vanish separately.  The
$\Lambda^0V$ component of the first equation says that
\begin{equation}
  \tfrac14 \omega_{\mu\nu\rho}{}^{CD} \xi^{\mu\nu\rho}_{BCD}{}^A = 0
  \implies \xi = 0,
\end{equation}
whereas the $\Lambda^0V$ component of the second equation vanishes.  The
$\Lambda^2V$ component of the second equation is
\begin{equation}
 -\tfrac38 \omega_{\mu\nu\rho}{}^A{}_B \theta^\rho{}_{\sigma\tau}\left(
    \eta^{\nu\sigma} \eta^{\mu\alpha} \eta^{\tau\beta} + \tfrac14
    \epsilon^{\mu\nu\sigma\tau\alpha\beta} \right) \Gamma_{\alpha\beta} P_+ = 0.
\end{equation}
Since $\Sigma_+$ is a faithful $\fso(V)$-representation we may write this equation as
\begin{equation}
  \tfrac3{16} \omega_{\mu\nu\rho}{}^A{}_B \theta^\rho{}_{\sigma\tau}\left(
    \eta^{\nu\sigma} \eta^{\mu\alpha} \eta^{\tau\beta} -
    \eta^{\nu\sigma} \eta^{\mu\beta} \eta^{\tau\alpha} + \tfrac12
    \epsilon^{\mu\nu\sigma\tau\alpha\beta} \right) = 0.
\end{equation}
We now abstract $\omega$, remembering that this projects onto the
antiself-dual component of the resulting expression and arrive at:\footnote{
This equation for $\theta$ defines an $\fso(V)$-equivariant map
$\Phi:\Lambda^3_-V \to \Lambda^2V \otimes \Lambda^3_-V$.  There is a
one-dimensional space of such maps, spanned by the transpose of the
$\fso(V)$ action $\mu: \fso(V) \otimes \Lambda^3_+V \to \Lambda^3_+V$,
so that $\Phi = c \mu^T$ for some $c \in \RR$.  Let $\theta \in
\Lambda^3_-V$ be in the kernel of this map.  Then for all $L \in
\fso(V)$ and $\Xi \in \Lambda^3_+V$, we have
\begin{equation}
0 = \left<\Phi(\theta), L \otimes \Xi\right> = 
c \left<\theta, \mu(L \otimes \Xi)\right> = 
c \left<\theta, L \Xi\right> =
- c \left<L \theta, \Xi\right>,
\end{equation}
which implies that if $c\neq 0$ then $\theta$ is $\fso(V)$-invariant and
hence $\theta = 0$. However, we will now see that $c=0$ and hence $\theta$
remains unconstrained.}
\begin{equation}
  \eta^{\alpha[\mu} \theta^{\nu\rho]^-\beta} - \eta^{\beta[\mu}
  \theta^{\nu\rho]^-\alpha} - \tfrac12 \theta^{[\rho}{}_{\sigma\tau}
  \epsilon^{\mu\nu]^-\sigma\tau\alpha\beta} = 0.
\end{equation}
Expanding this out and simplifying, we arrive at the equation
\begin{equation}
  \theta^{[\rho}{}_{\sigma\tau} \epsilon^{\mu\nu\alpha\beta]\sigma\tau}
  = 0.
\end{equation}
Taking the Hodge dual of the LHS yields
\begin{equation}
  \theta^{\rho}{}_{\sigma\tau} \epsilon^{\mu\nu\alpha\beta\sigma\tau}
  \epsilon_{\rho\mu\nu\alpha\beta\pi} = 4! \left(\delta^\tau_\rho
    \delta^\sigma_\pi - \delta^\sigma_\rho
    \delta^\tau_\pi \right) \theta^\rho{}_{\sigma\tau} = 0~,
\end{equation}
which is identically satisfied.  In other words, we find that
$\theta$ is unconstrained.

In summary, we have proved the following extension of
Theorem~\ref{thm:spencer}:

\begin{theorem}
  \label{thm:spencer-R}
  There is an isomorphism of representations of $\fso(V) \oplus \fsp(1)$
  \begin{equation}
    H^{2,2}(\fphat_-,\fphat) \cong \left(\Lambda^3 V \otimes \Lambda^2
      \Delta \right) \oplus \left(V \otimes \odot^2\Delta\right)~.
  \end{equation}
  The cohomology class corresponding to elements $H \in \Lambda^3V$ and $\varphi
  \in V \otimes \odot^2\Delta$ is represented by the cocycle $\beta +
  \gamma + \rho\in C^{2,2}(\fphat_-,\fphat)$, where
  \begin{equation}
    \begin{split}
      (\beta_\rho s)^A &= \tfrac12 H_{\rho\mu\nu} \Gamma^{\mu\nu} s^A -
      3 \varphi_\rho{}^A{}_B s^B + \varphi^{\sigma\,A}{}_B
      \Gamma_{\rho\sigma} s^B,\\
      \gamma(s,s)_{\mu\nu} &= 2 \kappa^\rho H_{\rho\mu\nu} - 2
      \varphi^\rho{}_{AB} \omega_{\rho\mu\nu}{}^{AB},\\
      \rho(s,s)^A{}_B &= \tfrac13 \omega_{\mu\nu\rho}{}^A{}_B
      (H^{\mu\nu\rho} - \widetilde H^{\mu\nu\rho}),
    \end{split}
  \end{equation}
  with $\widetilde H$ the Hodge dual of $H$.
\end{theorem}
Note that the components $\beta$ and $\gamma$ in Theorems
\ref{thm:spencer} and \ref{thm:spencer-R} coincide. In particular, this
will lead to a uniform notion of a Killing spinor (see Definition
\ref{def:KillingSpinor} in \S \ref{sec:killing-spinors}). We also note
that Theorem \ref{thm:spencer-R} is of a more general scope than Theorem
\ref{thm:spencer}, since it includes $3$-forms which are not necessarily
self-dual. As we will shortly see, Theorem \ref{thm:spencer} is adequate
for the construction of a Killing superalgebra on a lorentzian
six-dimensional spin manifold endowed with a closed self-dual $3$-form
and a family $\varphi\in\Omega^1(M;\fsp(1))$ of coclosed $1$-forms, but
it is precisely the introduction of R-symmetry transformations and
Theorem \ref{thm:spencer-R} which allow to extend this construction (at
least partially) to the non self-dual case.

\section{The Killing superalgebra}
\label{sec:killing-spinors}

By analogy with the results in
\cite{Figueroa-OFarrill:2015rfh,Figueroa-O'Farrill:2015utu,deMedeiros:2016srz}
on four- and eleven-dimensional supergravities, we may read off from
Theorem~\ref{thm:spencer} (or Theorem~\ref{thm:spencer-R}) the form of a
Killing spinor equation.  In this section we define these Killing
spinors and investigate the conditions under which they generate a Lie
superalgebra.  On the six-dimensional manifolds admitting such Killing
spinors, the Lie superalgebra they generate can be interpreted as the
supersymmetry algebra for rigidly supersymmetric field theories.
Although we do not construct such theories in this paper, the method would parallel the well-known construction of
supersymmetric field theories in Minkowski space via an appropriate superspace
formalism.  Of course, in general, this construction would require a better
understanding of the representation theory of the resulting Lie
superalgebras than we possess at present.

Let us start with the following definition \ref{def:KillingSpinor}. We
recall that in our conventions $S$ is an irreducible representation of
$\Spin(V)$ of quaternionic dimension $2$.  Since the introduction of the
R-symmetry results in relaxing the self-duality of the $3$-form $H$, we
will work in the more general case, specialising to the self-dual case
if and when necessary.

\begin{definition}
 \label{def:KillingSpinor}
 Let $(M,g)$ be a lorentzian six-dimensional spin manifold, with
 associated spinor bundle $\SS\to M$ with typical fiber $S$. Let
 $H\in\Omega^3(M)$ be a $3$-form and $\varphi$ a $1$-form on $M$ with
 values in $\fsp(1)$. We say that a section $\varepsilon$ of $\SS$ is a
 \textbf{Killing spinor} if it obeys
  \begin{equation}
    \label{eq:KillingSpinor}
    \eD_X \varepsilon:= \nabla_X \varepsilon - \iota_X H \cdot \varepsilon + 3\varphi(X)\cdot\varepsilon-X\wedge \varphi\cdot\varepsilon=0\;,
	\end{equation}
for all $X\in\mathfrak{X}(M)$.
\end{definition}

We write ${\mathfrak{X}}(M) = \Gamma (TM)$ to identify vector fields with sections of the tangent bundle on $M$ and ${\mathfrak{S}}(M) = \Gamma (\SS)$ to identify spinor fields with sections of the spinor bundle $\SS\to M$.

Note that any non-zero Killing spinor is nowhere vanishing as it is parallel with respect to a connection on the spinor bundle.

In this section we investigate under which conditions such Killing
spinors generate a Lie superalgebra. We know from
\cite{Figueroa-O'Farrill:2013aca} that if $\varphi=0$ and $dH = 0$ then
this is the case. In this paper we will not assume $\varphi=0$ and give
differential constraints separately on $H$ and $\varphi$ which guarantee
the existence of the Killing superalgebra.

In practice, we shall work with complexified bundles and forms in what
follows, although we will not mention this explicitly. In particular we
note that
%In view of the $\fso(V)$-equivariant identification $S\otimes\CC \cong\Sigma_+ \otimes_{\CC} \Delta$, 
the (complexification of the) spinor bundle $\SS$ has a
``Grassmann-like'' decomposition
%\label{eq:GrDec}
\begin{equation}
  \SS=\$_+ \otimes \mathcal H\;,
\end{equation}
where
$\$_+$ is the bundle of positive-chirality spinors and $\mathcal H=M\times\Delta
\to M$ a trivial rank-two complex  vector bundle. (The bundle $\mathcal H$ is trivial since the action of the structure group $\Spin(V) \cong \SL(2,\HH)$ on $S\otimes\CC =
\Sigma_+ \otimes_{\CC} \Delta$ is trivial on $\Delta$.) The Levi-Civita
connection $\nabla$ is easily seen to be compatible with this
decomposition, that is
\begin{equation}
  \nabla(\varepsilon_+\otimes \zeta)=\nabla \varepsilon_+\otimes
  \zeta+\varepsilon_+\otimes\overline\nabla \zeta
\end{equation} 
for all $\varepsilon_+\in\Gamma(\$_+)$ and $\zeta\in\Gamma(\mathcal H)$,
where $\overline\nabla$ is a flat connection on $\mathcal H$. We will
also work with differential forms which take values in
$\mathfrak{sl}(2,\CC)$ and, whenever necessary, use the Cartan-Killing
form to identify the latter with its dual.

\subsection{Preliminaries}
\label{sec:KSAprel}

We collect here a series of auxiliary differential and algebraic relations, which will
be needed in the proof of the main Theorems \ref{thm:KSAI} and \ref{thm:KSAII}. 

Let $\varepsilon$ be a non-zero section  of $\SS$. It has associated the following differential forms:
\begin{itemize}
\item $\omega^{(1)}\in\Omega^1(M)$ given by 
  \begin{equation}
    \label{eq:form1}
    \omega^{(1)}(X_1)=
    \left(\varepsilon, X_1\cdot \varepsilon\right)\;,
  \end{equation}
\item a family of self-dual $3$-forms $\omega^{(3)}\in\Omega_+^3(M;\fsp(1))$ given by
  \begin{equation}
    \label{eq:form3}
    \omega^{(3)}_A(X_1,X_2,X_3)=\left(\varepsilon, (X_1\wedge X_2\wedge
      X_3\right)\cdot A\cdot\varepsilon)\;,
  \end{equation}
\item $\omega^{(5)}\in\Omega^5(M)$ given by
  \begin{equation}
    \label{eq:form5}
    \omega^{(5)}(X_1,\ldots,X_5)=
    \left(\varepsilon, (X_1 \wedge \cdots\wedge X_5) \cdot \varepsilon\right)\;,
  \end{equation}
\end{itemize}
where $X_1,\ldots,X_5\in\mathfrak{X}(M)$ and $A\in\fsp(1)$. We note that
the $1$-form $\omega^{(1)}$ is the $g$-dual of the Dirac current
$\kappa=\kappa(\varepsilon,\varepsilon)\in\mathfrak{X}(M)$ whereas
$\omega^{(5)}=-\star\omega^{(1)}$. The family of self-dual $3$-forms has
already been introduced at a purely linear algebra level in
\eqref{eq:fam3}. We here adorn it with a superscript, to emphasise that
it is a $3$-form and avoid any confusion with the other spinor
bilinears.

\begin{proposition}
  \label{prop:auxiliary}
  Let $\varepsilon$ be a Killing spinor on $(M,g,H,\varphi)$. Then
  \begin{equation}
    \label{eq:covder}
    \begin{split}
      \nabla_X\omega^{(1)}&=2\imath_X\imath_\kappa H-2\imath_X\imath_{\varphi}\omega^{(3)}\;,\\
      \nabla_X\omega^{(3)}_A(X_1,X_2,X_3) &=
      -6\operatorname{skew}g(X,X_1)\tr(A\varphi(X_2))\omega^{(1)}(X_3) -
      6\operatorname{skew}\omega^{(3)}_A(\imath_{X}\imath_{X_1} H,X_2,X_3)\\
      &\;\;\;\;+3\omega^{(3)}_{[\varphi(X),A]}(X_1,X_2,X_3)+3\operatorname{skew}\omega^{(3)}_{[\varphi(X_1),A]}(X_2,X_3,X) \\
      &\;\;\;\;-3\operatorname{skew}g(X,X_1)(\imath_{[\varphi,A]}\omega^{(3)})(X_2,X_3)-\omega^{(5)}(X,X_1,X_2,X_3,\tr(A\varphi))\;,\\
      \nabla_X\omega^{(5)}&=2X\wedge\omega^{(1)}\wedge\star H-2
      X\wedge\star(\imath_{\varphi}\omega^{(3)})\;,
    \end{split}	
  \end{equation}
  for all $X,X_1,X_2,X_3\in\mathfrak{X}(M)$ and $A\in\fsp(1)$, where
  $\operatorname{skew}=\operatorname{skew}_{X_1,X_2,X_3}$ is
  skew-symmetrization on $X_1,X_2,X_3$ with weight one.
\end{proposition}

\begin{proof}
  For any Killing spinor $\varepsilon$ and $X,Y\in\mathfrak{X}(M)$ we compute
  \begin{equation}
    \label{eq:diff1}
    \begin{split}
      \nabla_{X}\omega^{(1)}(Y)&=2\left(\varepsilon,Y\cdot\nabla_{X}\varepsilon\right)\\
      &=2\left(\varepsilon, Y\cdot \imath_X H\cdot\varepsilon \right) -6\left(\varepsilon, Y\cdot
        \varphi(X)\cdot\varepsilon\right) +2\left(\varepsilon, Y\cdot (X\wedge \varphi)\cdot\varepsilon\right)\\
      &=2\left(\varepsilon,\imath_Y\imath_XH\cdot\varepsilon\right)+2\left(\varepsilon,Y\wedge X\wedge\varphi\cdot\varepsilon\right)\;,
    \end{split}
  \end{equation}
  where last identity is a consequence of the decompositions
  \begin{equation}
    \label{eq:decompo}
    \begin{split}
      \odot^2 S&=\Lambda^1V\oplus(\Lambda^3_+ V\otimes\odot^2 \Delta)\;,\\
      \Lambda^2 S&=(\Lambda^1 V\otimes\odot^2\Delta)\oplus\Lambda^3_+V\;.
    \end{split}
  \end{equation}
  This shows the first equation in \eqref{eq:covder} and applying $\star$ on both sides of it readily yields the last equation too. Similarly, for all $X,X_1,X_2,X_3\in\mathfrak{X}(M)$ and $A\in\fsp(1)$, we compute
  \begin{equation}
    \label{eq:nabla3}
    \begin{split}
      \nabla_{X}\omega^{(3)}_A(X_1,X_2,X_3)&=2\left(\varepsilon,(X_1\wedge X_2\wedge X_3)\cdot A\cdot\nabla_{X}\varepsilon\right)\\
      &=2\left(\varepsilon,(X_1\wedge X_2\wedge X_3)\cdot \imath_XH\cdot A\cdot\varepsilon \right)
      -6\left(\varepsilon,(X_1\wedge X_2\wedge X_3)\cdot A\cdot\varphi(X)\cdot\varepsilon     \right)\\
      &\;\;\;\;+2\left(\varepsilon,(X_1\wedge X_2\wedge X_3)\cdot A\cdot (X\wedge\varphi)\cdot\varepsilon
      \right)\\
      &=2\mathfrak{S}\left(\varepsilon,X_1\wedge X_2\wedge\imath_{X_3}\imath_X H\cdot A\cdot\varepsilon \right)-3\left(\varepsilon,(X_1\wedge X_2\wedge X_3)\cdot [A,\varphi(X)]\cdot\varepsilon     \right)\\
      &\;\;\;\;+\left(\varepsilon,(X_1\wedge X_2\wedge X_3\wedge X\wedge \tr(A\varphi))\cdot\varepsilon
      \right)
      +\mathfrak{S}\left(\varepsilon,\imath_{X_1}(X\wedge [A,\varphi])\wedge X_2 \wedge X_3\cdot\varepsilon     \right)\\
      &\;\;\;\;+\mathfrak{S}\imath_{X_1}\imath_{X_2}(X\wedge \tr(A\varphi))\left(\varepsilon,X_3\cdot\varepsilon \right)\\
    \end{split}
  \end{equation}
  where $\mathfrak{S}=\mathfrak{S}_{X_1,X_2,X_3}$ is the cyclic sum and in 
  the last step we repeatedly used the decomposition \eqref{eq:decompo} and the identity
  \begin{equation}
    2AB=[A,B]+\tr(AB)\Id\;,
  \end{equation}
  which holds for all traceless complex $2\times 2$-matrices. The second
  equation in \eqref{eq:covder} is just \eqref{eq:nabla3} combined with
  definitions \eqref{eq:form1}-\eqref{eq:form5}.
\end{proof}

\begin{corollary}
  \label{cor:Killing}
  Let $\varepsilon$ be a Killing spinor on $(M,g,H,\varphi)$. Then
  \begin{equation}
    \label{eq:extdiff}
    \begin{split}
      d\omega^{(1)}&=4\imath_\kappa H-4\imath_\varphi\omega^{(3)}\;,\\
      d\omega^{(3)}_A(X_0,X_1,X_2,X_3) &=
      -24\operatorname{Skew}\omega^{(3)}_A(\imath_{X_0}\imath_{X_1}H,X_2,X_3)
      -4(\imath_{\tr(A\varphi)}\omega^{(5)})(X_0,X_1,X_2,X_3)\;,\\
      d\omega^{(5)}&=0\;,
    \end{split}
  \end{equation}
  for all $X_0,\ldots, X_3\in\mathfrak{X}(M)$ and $A\in\fsp(1)$, where
  $\operatorname{Skew}=\operatorname{Skew}_{X_0,X_1,X_2,X_3}$ is
  skew-symmetrization on $X_0,\ldots,X_3$ with weight one. In particular
  the Dirac current $\kappa$ is a Killing vector field.
\end{corollary}

\begin{proof}
  By Proposition \ref{prop:auxiliary} we have that
  $\nabla\omega^{(1)}=\tfrac{1}{2}d\omega^{(1)}$. In other words
  $\omega^{(1)}$ is a coclosed conformal Killing $1$-form, $\kappa$ a
  Killing vector field and $d\omega^{(5)}=-d\star\omega^{(1)}=0$. It
  remains to compute
  \begin{equation}
    \begin{split}
      d\omega^{(3)}_A(X_0,X_1,X_2,X_3)&=4\operatorname{Skew}\nabla_{X_0}\omega^{(3)}_A(X_1,X_2,X_3)\\
      &=-24\operatorname{Skew}\omega^{(3)}_A(\imath_{X_0}\imath_{X_1} H,X_2, X_3)+12\operatorname{Skew}\omega^{(3)}_{[\varphi(X_0),A]}(X_1,X_2,X_3)\\
      &\;\;\;\;-12\operatorname{Skew}\omega^{(3)}_{[\varphi(X_0),A]}(X_1,X_2, X_3)
      -4\omega^{(5)}(X_0,X_1, X_2, X_3, \tr(A\varphi))
      \\
      &=-24\operatorname{Skew}\omega^{(3)}_A(\imath_{X_0}\imath_{X_1} H,X_2, X_3)
      -4(\imath_{\tr(A\varphi)}\omega^{(5)})(X_0,X_1, X_2, X_3)\;,
    \end{split}
  \end{equation}
  completing the proof.
\end{proof}

To proceed further, we shall need some algebraic facts on partial and
full skew-symmetrisations of terms of the form
$\alpha(\imath_{X_0}\imath_{X_1} \beta,X_2,X_3)$, where
$\alpha,\beta\in\Omega^3(M)$. Such terms appear in \eqref{eq:covder} and
\eqref{eq:extdiff}, and they will play a crucial role towards the proof
of Theorems~\ref{thm:KSAI} and \ref{thm:KSAII}.
\begin{lemma}\label{lem:sd2}
  Let $\alpha,\beta\in\Omega^3(M)$ and consider the
  associated $4$-form $\left[\alpha\cdot\beta\right]_4\in\Omega^4(M)$
  given by
  \begin{equation}\label{eq:fullskew}
    \left[\alpha\cdot\beta\right]_4(X_0,\ldots,X_3)=\operatorname{Skew}\alpha(\imath_{X_0}\imath_{X_1}\beta,X_2,X_3)\;,
  \end{equation}
  where $X_0,\ldots, X_3\in\mathfrak{X}(M)$ and
  $\operatorname{Skew}=\operatorname{Skew}_{X_0,X_1,X_2,X_3}$ is
  skew-symmetrisation on $X_0,\ldots,X_3$ with weight one. Then
  \begin{enumerate}
  \item[(i)]
    $\left[\alpha\cdot\beta\right]_4=\left[\beta\cdot\alpha\right]_4$
    for all $\alpha,\beta\in\Omega^3(M)$;
  \item[(ii)]
    $\left[\alpha\cdot\beta\right]_4=0$ if both forms are self-dual (or
    antiself-dual).
  \end{enumerate}
\end{lemma}

\begin{proof}
  The first claim follows directly from a simple computation. The second
  claim is also immediate, since we have the decomposition
  \begin{equation}
    \Lambda^3_\pm V\otimes\Lambda^3_\pm V=(\Lambda^3_\pm
    V\otimes\Lambda^3_\pm V)_0\oplus (V\otimes\Lambda^3_\pm
    V)_0\oplus\odot^2_0 V
  \end{equation}
  into irreducible $\fso(V)$-modules
%\Lambda^3_+ V\otimes\Lambda^3_- V&=(\Lambda^3_+ V\otimes\Lambda^3_- V)_0\oplus \Lambda^2 V\oplus \Lambda^0 V\;,
  and therefore any $\fso(V)$-equivariant map from $\Lambda^3_\pm
  V\otimes\Lambda^3_\pm V$ to $\Lambda^4 V$ is necessarily trivial.
\end{proof}

\begin{proposition}
  \label{cor:LieDF}
  Let $\varepsilon$ be a Killing spinor on $(M,g,H,\varphi)$. Then 
  \begin{equation}
    \label{eq:Lienabla1}
    \begin{split}
      d\imath_\kappa H=d\imath_{\varphi}\omega^{(3)}
    \end{split}
  \end{equation}
  and $\mathscr L_{\kappa}\omega^{(1)}= \mathscr
  L_{\kappa}\omega^{(5)}=0$. If $H\in\Omega^3(M)$ is self-dual then
  $\mathscr L_{\kappa} \omega^{(3)}=0$ too.
\end{proposition}

\begin{proof}
  Equation \eqref{eq:Lienabla1} follows by applying the exterior
  derivative to both sides of the first identity in \eqref{eq:extdiff}.
  We recall that $\kappa$ is a Killing vector field by Corollary
  \ref{cor:Killing}, whence $\mathscr L_{\kappa}\omega^{(1)}=0$ and
  \begin{equation}
    \label{eq:Lie5}
    \begin{split}
      \mathscr L_{\kappa}\omega^{(5)}&=-\mathscr L_{\kappa}\star\omega^{(1)}\\
      &=-\star\mathscr L_{\kappa}\omega^{(1)}=0\;.
    \end{split}
  \end{equation}
  Now $d\omega^{(3)}_A=-4\imath_{\tr(A\varphi)}\omega^{(5)}$ if $H$ is
  self-dual, by Corollary \ref{cor:Killing} and Lemma \ref{lem:sd2}.
  Furthermore $\imath_\kappa\omega^{(3)}_A=0$ for all $A\in\fsp(1)$ by
  Lemma~\ref{lem:kwzero}. It then follows
  \begin{equation}
    \begin{split}
      \mathscr L_{\kappa}\omega^{(3)}_A&=\imath_{\kappa}d\omega^{(3)}_A\\
      &=4\imath_{\tr(A\varphi)}\imath_{\kappa}\omega^{(5)}=0\;,
    \end{split}
  \end{equation}
  for all $A\in\fsp(1)$.
\end{proof}

\begin{lemma}
  \label{lem:proportional}
  Let $\alpha,\beta\in\Omega^3_+(M)$ be self-dual $3$-forms, with
  $\beta$ nowhere vanishing. Let us assume that there exists a nowhere
  vanishing null vector field $N\in\mathfrak{X}(M)$ with the property
  that
  \begin{equation}
    \imath_N\alpha=\imath_N\beta=0
    \label{eq:nullcontraction}
  \end{equation}
  and
  \begin{equation}
    \label{eq:proportional}
    \operatorname{skew}\alpha(\imath_{X}\imath_{X_1} \beta,X_2,X_3)=0
  \end{equation}
  for all $X,X_1,X_2,X_3\in\mathfrak{X}(M)$, where
  $\operatorname{skew}=\operatorname{skew}_{X_1,X_2,X_3}$ is
  skew-symmetrization on $X_1,X_2,X_3$ with weight one. Then
  $\alpha=f\beta$ for some $f\in\mathscr C^{\infty}(M)$.
\end{lemma}

\begin{proof}
  It is enough to establish the claim pointwise. We fix $p\in M$ and a
  Witt basis $(\be_+,\be_-,\be_1,\ldots,\be_4)$ of $T_pM$, which we use
  to identify $T_pM$ with $V$ and $N|_p$ with $\be_+$. We then write
  \begin{equation}
    V=\RR\be_+\oplus\RR\be_-\oplus E\;,
  \end{equation}
  where $E=\left\langle \be_1,\ldots,\be_4\right\rangle$ is
  $4$-dimensional euclidean.  It follows from \eqref{eq:nullcontraction}
  that $\alpha=e_+\wedge\widetilde\alpha$ and
  $\beta=e_+\wedge\widetilde\beta$, for some antiself-dual forms
  $\widetilde\alpha,\widetilde\beta\in\Lambda^2_-E$. (A different choice
  of orientations would result in
  $\widetilde\alpha,\widetilde\beta\in\Lambda^2_+E$.)  Now, under the
  isomorphism $\Lambda^2 E\cong \fso(4)$, the module $\Lambda_-^2E$ gets
  identified with the ideal $\fso_-(3)$ of
  \begin{equation}
    \fso(4)=\fso_+(3)\oplus\fso_-(3)
  \end{equation}
  and the Lie brackets on $\fso_-(3)$ with the skew-symmetric operation
  $[-,-]:\Lambda^2_-E\otimes\Lambda^2_-E\to\Lambda^2_-E$ given by
  \begin{equation}
    [\widetilde\alpha,\widetilde\beta](X_1,X_2) =
    \widetilde\alpha(\imath_{X_1}\widetilde\beta,X_2) -
    \widetilde\alpha(\imath_{X_2}\widetilde\beta,X_1)\;,
  \end{equation}
  where $\widetilde\alpha,\widetilde\beta\in\Lambda^2_-E$ and
  $X_1,X_2\in E$.

  Equation~\eqref{eq:proportional} with $X=X_3=\be_-$ leads to
  $[\widetilde\alpha,\widetilde\beta]=0$ and our claim follows from the
  fact that the centraliser of any non-zero element in $\fso_-(3)$ is
  always $1$-dimensional.
\end{proof}
	
\subsection{The Killing superalgebra. Case of self-dual $3$-form.}
\label{sec:KSAself}

Let $(M,g,H,\varphi)$ be a six-dimensional lorentzian spin manifold
$(M,g)$ with spinor bundle $\SS$ which is, in addition, endowed with a
self-dual $3$-form $H$ and a $1$-form $\varphi$ on $M$ with values in
$\fsp(1)$. In this section we shall construct a Lie superalgebra
$\mathfrak{k}=\mathfrak{k}_{\bar 0}\oplus\mathfrak{k}_{\bar 1}$
naturally associated with $(M,g,H,\varphi)$, under appropriate
conditions on $H$ and $\varphi$.
	
Set
\begin{equation}
  \label{eq:KSA}
  \begin{split}
    \mathfrak{k}_{\bar 0}&=\{X\in\mathfrak{X}(M)\mid\mathscr
    L_{X}g=\mathscr L_{X} H=\mathscr L_{X}\varphi=0\}\;,\\
    \mathfrak{k}_{\bar 1}&=\{\varepsilon\in{\mathfrak{S}}(M)\mid
    \eD_X\varepsilon=0\;\;\text{for all}\;\;X\in\mathfrak{X}(M)\}\;,
  \end{split}
\end{equation}
where $\eD$ is the spinor connection introduced in Definition
\ref{def:KillingSpinor}.  We consider the operation
$[-,-]:\fk\otimes\fk\to\fk$  compatible with the parity of $\mathfrak
k=\mathfrak k_{\bar 0}\oplus\mathfrak k_{\bar 1}$ and determined by the
following maps:
\begin{itemize}
\item $[-,-]:\fk_{\bar 0}\otimes\fk_{\bar 0}\to\fk_{\bar 0}$ is
  the usual commutator of vector fields,
\item $[-,-]:\fk_{\bar 1}\otimes\fk_{\bar 1}\to\fk_{\bar 0}$ is a
  symmetric map, with
  $[\varepsilon,\varepsilon]=\kappa(\varepsilon,\varepsilon)$ given by
  the Dirac current of $\varepsilon\in\fk_{\bar 1}$, and
\item $[-,-]:\fk_{\bar 0}\otimes\fk_{\bar 1}\to\fk_{\bar 1}$ is the
  spinorial Lie derivative of Lichnerowicz and Kosmann (see
  \cite{MR0312413}).
\end{itemize}

The fact that this operation actually takes values in $\fk$ is a
consequence of the main Theorem \ref{thm:KSAI} below, where we show that
$[-,-]$ is the bracket of a Lie superalgebra structure on $\fk$.
Assuming that result for the moment we make the following.

\begin{definition}
  The pair $(\fk=\fk_{\bar 0}\oplus\fk_{\bar 1},[-,-])$ is called the
  \emph{Killing superalgebra} associated with  $(M,g,H,\varphi)$.
\end{definition}

Let us briefly recall the main properties of the spinorial Lie
derivative, see \cite{MR0312413} and also, e.g., \cite{JMFKilling}.  The
Lie derivative of a spinor field $\varepsilon$ along a Killing vector
field $X$ is defined by $\mathscr L_{X}\varepsilon =
\nabla_{X}\varepsilon + \sigma(A_X)\varepsilon$, where
$\sigma:\fso(TM)\to\End(\SS)$ is the spin representation and
$A_X=-\nabla X\in\fso(TM)$.  It enjoys the following basic properties,
for all Killing vectors $X,Y$, spinors $\varepsilon$, functions $f$ and
vector fields $Z$:
\begin{enumerate}[label=(\roman*)]
\item $\mathscr L_X$ is a derivation:
  \begin{equation}
    \mathscr L_X(f\varepsilon)=X(f)\varepsilon+f\mathscr L_{X}\varepsilon\;;
  \end{equation}
\item $X\mapsto\mathscr L_X$ is a representation of the Lie algebra of Killing vector fields:
  \begin{equation}
    \mathscr L_X(\mathscr L_Y\varepsilon)-\mathscr L_Y(\mathscr L_X\varepsilon)=\mathscr L_{[X,Y]}\varepsilon\;;
  \end{equation}
\item $\mathscr L_X$ is compatible with Clifford multiplication:
  \begin{equation}
    \mathscr L_X(Z\cdot \varepsilon)=[X,Z]\cdot\varepsilon+ Z\cdot\mathscr L_{X}\varepsilon\;;
  \end{equation}
\item $\mathscr L_X$ is compatible with the Levi-Civita connection:
  \begin{equation}
    \mathscr L_X(\nabla_{Z}\varepsilon)=\nabla_{[X,Z]}\varepsilon+ \nabla_{Z}(\mathscr L_{X}\varepsilon)\;.
  \end{equation}	
\end{enumerate}
		
Using property (iii), it is not difficult to see that the Dirac current
is equivariant under the action of Killing vector fields, namely that
\begin{equation}
  [X,\kappa(\varepsilon,\varepsilon)]=2\kappa(\mathscr L_{X}\varepsilon,\varepsilon)\;,
\end{equation}
for any Killing vector $X$ and spinor $\varepsilon$. It is also clear
from basic properties of Lie derivatives of vector fields that
$[\fk_{\bar 0},\fk_{\bar 0}]\subset\fk_{\bar 0}$ and that for any
$X\in\fk_{\bar 0}$ and $Z\in\mathfrak{X}(M)$ we have
\begin{equation}
  [\eL_X, \eD_Z] = \eD_{[X,Z]},
\end{equation}
since $\eD$ depends solely on the data $(g,H,\varphi)$, which is
preserved by $X \in \fk_{\bar 0}$.  This shows that $[\fk_{\bar
  0},\fk_{\bar  1}]\subset\fk_{\bar 1}$ or, in other words, that the Lie
algebra $\fk_{\bar 0}$ acts on $\fk_{\bar 1}$ via the spinorial Lie
derivative.  It is clear after a moment’s thought that there are still
conditions to be satisfied in order for $\fk=\fk_{\bar 0}\oplus\fk_{\bar
  1}$  to be a Lie superalgebra:
\begin{enumerate}
\item $\kappa(\varepsilon,\varepsilon)\in\fk_{\bar 0}$, and
\item $\mathscr L_{\kappa(\varepsilon,\varepsilon)}\varepsilon=0$,
\end{enumerate}
for all $\varepsilon\in\fk_{\bar1}$. The second equation is equivalent
to the component of the Jacobi identity for $\fk$ with three odd
elements. The rest of this section will be devoted to investigating
(1)-(2). \vskip0.2cm\par
We have already established in Corollary \ref{cor:Killing} that
$\kappa(\varepsilon,\varepsilon)$ is a Killing vector. The following
result provides a more suggestive interpretation of this fact and the
Jacobi identity for three odd elements, in terms of the Spencer complex
considered in \S \ref{sec:calculation-h2-2fp}.

\begin{proposition}
  The first and second cocycle conditions of the Spencer complex are
  equivalent to $\kappa=\kappa(\varepsilon,\varepsilon)$ being a Killing
  vector and $\mathscr L_{\kappa}\varepsilon=0$, for all
  $\varepsilon\in\fk_{\bar 1}$.
\end{proposition}

\begin{proof}
  Recall the cocycle conditions \eqref{eq:co1}-\eqref{eq:co2}. For all
  $\varepsilon\in\fk_{\bar 1}$, $Z\in\mathfrak{X}(M)$, we compute
  \begin{equation}
    \begin{split}
      \nabla_Z \kappa&=\nabla_Z \kappa(\varepsilon,\varepsilon)\\
      &=2\kappa(\nabla_Z\varepsilon,\varepsilon)=2\kappa(\beta_Z\varepsilon,\varepsilon)\\
      &=-\gamma(\varepsilon,\varepsilon)Z\;,
    \end{split}
  \end{equation}
  which says that $\kappa$ is a Killing vector, since $\gamma(\varepsilon,\varepsilon)$ is a section of $\fso(TM)$. Similarly
  \begin{equation}
    \begin{split}
      \mathscr L_{\kappa}\varepsilon&=\nabla_\kappa\varepsilon-\sigma(\nabla \kappa)\varepsilon\\
      &=\beta_\kappa\varepsilon+\sigma(\gamma(\varepsilon,\varepsilon))\varepsilon\\
      &=0\;,
    \end{split}
  \end{equation}
  and the proposition is proved.
\end{proof}

We are now ready to prove the following.

\begin{theorem}
  \label{thm:KSAI}
  Let $(M,g)$ be a lorentzian six-dimensional spin manifold, with
  associated spinor bundle $\SS\to M$ with typical fiber $S$. Let
  $H\in\Omega_+^3(M)$ be a self-dual $3$-form and $\varphi$ a $1$-form
  on $M$ with values in $\fsp(1)$. If
  \begin{itemize}
  \item $dH=0$, and
  \item $d^\star\varphi=0$,
  \end{itemize}
  then there exists a natural structure of Lie superalgebra on the
  direct sum $\fk=\fk_{\bar 0}\oplus\fk_{\bar 1}$ of the spaces
  \eqref{eq:KSA}, called the Killing superalgebra of $(M,g,H,\varphi)$.
\end{theorem}

\begin{proof}
  It remains only to show that
  \begin{equation}
    \mathscr L_{\kappa}H=\mathscr L_{\kappa}\varphi=0\;,
  \end{equation}
  where $\kappa=\kappa(\varepsilon,\varepsilon)$ is the Dirac current of
  a Killing spinor $\varepsilon\in\fk_{\bar 1}$. We recall that
  $d\omega^{(3)}_A=-4\imath_{\tr(A\varphi)}\omega^{(5)}$ by self-duality
  of $H$, Corollary \ref{cor:Killing} and Lemma \ref{lem:sd2}. It
  follows that
  \begin{equation}
    \begin{split}
      d\omega^{(3)}_A&=4\imath_{\tr(A\varphi)}\star\omega^{(1)}\\
      &=4\star(\omega^{(1)}\wedge\tr(A\varphi))\\
      &=-4\imath_\kappa(\star\tr(A\varphi))
    \end{split}
  \end{equation}
  and applying the exterior derivative to both sides yields
  \begin{equation}
    \begin{split}
      0&=d\imath_\kappa(\star\tr(A\varphi))=d\imath_\kappa(\star\tr(A\varphi))+\imath_\kappa(d(\star\tr(A\varphi)))\\
      &=\mathscr L_{\kappa}\star\tr(A\varphi)=\star\mathscr L_{\kappa}\tr(A\varphi)\;,
    \end{split}
  \end{equation}
  for all $A\in\fsp(1)$, or, in other words,
  $\mathscr L_{\kappa}\varphi=0$. We now turn to prove
  $\mathscr L_\kappa H=0$, which is slightly more involved.
  
  First, we use \eqref{eq:Lienabla1} and $\imath_\kappa\omega^{(3)}=0$
  to compute
  \begin{equation}
    \begin{split}
      \imath_\kappa(\mathscr L_\kappa H)&=\imath_\kappa(d\imath_\kappa H)=\imath_\kappa(d\imath_{\varphi}\omega^{(3)})\\
      &=\mathscr L_\kappa\imath_{\varphi}\omega^{(3)}=\imath_{[\kappa,\varphi]}\omega^{(3)}+\imath_{\varphi}\mathscr L_{\kappa}\omega^{(3)}\\
      &=0\;,
    \end{split}
  \end{equation}
  where the last identity follows from
  $\mathscr L_\kappa\varphi=\mathscr L_\kappa\omega^{(3)}=0$. Hence, we
  have self-dual forms $\mathscr L_{\kappa}H$ and $\omega_A^{(3)}$ which
  vanish when evaluated on $\kappa$.
  Moreover, using that $\kappa$ is a Killing vector, we have for all $X\in\mathfrak{X}(M)$:
  \begin{equation}
    \begin{split}
      0&=\nabla_X(\mathscr L_{\kappa}\omega^{(3)}_A)=
      \mathscr L_\kappa (\nabla _X\omega^{(3)}_A)
      -\nabla_{[\kappa,X]}\omega^{(3)}_A\\
      &=-6\operatorname{skew}\omega^{(3)}_A(\imath_{X}\imath_{X_1} \mathscr L_\kappa H,X_2,X_3),
    \end{split}
  \end{equation}
  where the last identity follows from a direct computation using the
  expression~\eqref{eq:covder} of $\nabla\omega^{(3)}_A$, $\mathscr
  L_{\kappa}\varphi=0$ and $\mathscr L_{\kappa}\omega^{(1)}=\mathscr
  L_{\kappa}\omega^{(3)}_A=\mathscr L_{\kappa}\omega^{(5)}=0$, see
  Proposition~\ref{cor:LieDF}. Here $\operatorname{skew}$ is, as usual,
  skew-symmetrisation on $X_1,X_2,X_3$ with weight one.

  Now, let us assume for a contradiction that $\mathscr L_{\kappa}H$ is
  (locally) non-zero. Then Lemma \ref{lem:proportional} applies with
  \begin{equation}
    \alpha=\omega^{(3)}_A\;,\qquad\beta=\mathscr L_{\kappa}H\;,\qquad N=\kappa\;,
  \end{equation} 
  and $\omega^{(3)}_A=f_A\mathscr L_\kappa H$ for some (locally defined)
  function $f_A$, for all $A\in\fsp(1)$. This implies that the $3$-forms
  $\omega^{(3)}_A$, $A\in\fsp(1)$, are pairwise linearly dependent at all
  points $p\in M$, which is absurd by Lemma \ref{lem:linindep}.

  % We then use Proposition \ref{cor:LieDF} again and the last hypothesis of the theorem to get a contradiction:
  % \begin{equation}
  %   \begin{split}
  %     \mathscr L_{K}H&=d\imath_{\varphi}\omega^{(3)}=d(\imath_{\varphi^A}\omega^{(3)}_A)\\
  %     &=d(f_A\imath_{\varphi^A}\mathscr L_K H)=d(f_A\mathscr L_K(\imath_{\varphi^A} H))\\
  %     &=0\;.
  %   \end{split}
  % \end{equation}

  The theorem is proved.
\end{proof}

\subsection{The Killing superalgebra. Case of $H$ not necessarily self-dual.}
\label{sec:KSAnotself}

This section will be devoted to constructing the Killing superalgebra
when the $3$-form is not necessarily self-dual. As we have seen in
Theorem \ref{thm:spencer-R}, it is precisely the introduction of the
R-symmetry which allowed us to relax the self-duality assumption. We
will therefore consider a six-dimensional (connected) lorentzian spin
manifold $(M,g)$ with spinor bundle $\SS=\$_+ \otimes \mathcal H$, endowed with a
$3$-form $H$ and a $1$-form $\varphi$ with values in $\fsp(1)$.

It will turn out that the existence of a Killing superalgebra
$\fkhat=\fkhat_{\bar 0}\oplus\fkhat_{\bar 1}$ extended by R-symmetries
depends not just on some constraints on $H$, $\varphi$ but also on an
algebraic identity relating $\varphi$ with the R-symmetries. Due to
this, we will ultimately restrict our analysis to the case where
$\varphi=0$ (see Theorem~\ref{thm:KSAII}).

To make contact with the notation of Theorem   \ref{thm:spencer-R}, it is convenient to introduce
the bundle morphism
\begin{equation}
  \label{eq:maprho}
  \begin{aligned}
    \rho&:\Lambda^3_+ T^*M\otimes\mathfrak{sp}(\mathcal H) \to
    \mathfrak{sp}(\mathcal H)\;,\\
    \rho&(\omega)=4g(\omega,H^{-})\;,
  \end{aligned}
\end{equation}
where $\mathfrak{sp}(\mathcal H)=M\times\mathfrak{sp}(\Delta)\to M$ is
the trivial rank-three subbundle with fiber $\mathfrak{sp}(\Delta)$ of
the bundle of endomorphisms of $\SS$ and $H=H^++H^-$ the decomposition
of $H$ into self-dual and antiself-dual components.

We set
\begin{equation}
  \label{eq:KSA0}
  \begin{split}
    \fkhat_{\bar 0}&=\mathfrak{k}_{\bar 0}\oplus\mathfrak{R}\;,\;\;\;\text{where}\;\;\\
    \mathfrak{k}_{\bar 0}&=\{X\in\mathfrak{X}(M)\mid\mathscr
    L_{X}g=\mathscr L_{X} H=\mathscr L_{X}\varphi=0\}\;,\\
    \mathfrak{R}&=\{R\in\Gamma(\mathfrak{sp}(\mathcal H))\mid \eD_X R=0\;\;\text{for all}\;\;X\in\mathfrak{X}(M)\}\;,
  \end{split}
\end{equation}
and
\begin{equation}
  \label{eq:KSA1}
  \begin{split}
    \fkhat_{\bar 1}&=\mathfrak{k}_{\bar 1}=\{\varepsilon\in{\mathfrak{S}}(M)\mid
    \eD_X\varepsilon=0\;\;\text{for all}\;\;X\in\mathfrak{X}(M)\}\;.
  \end{split}
\end{equation}
Here $\eD$ is, as usual, the spinor connection introduced in Definition \ref{def:KillingSpinor}.
We note that $\mathfrak{R}$ consists of all the $\eD$-parallel R-symmetry transformations --
we will expand on this condition later on in Proposition \ref{prop:D-parallel}.
We consider the operation $[-,-]$ on $\fkhat=\fkhat_{\bar 0}\oplus\fkhat_{\bar 1}$ determined
by the usual commutator of vector fields, the Lichnerowicz-Kosmann spinorial Lie derivative and the following maps:
\begin{itemize}
\item $[-,-]:\mathfrak R\otimes\fkhat_{\bar 1}\to\fkhat_{\bar 1}$ is the natural action of a R-symmetry transformation on spinor fields;
\item $[-,-]:\fk_{\bar 0}\otimes\mathfrak{R}\to\mathfrak{R}$ is given by
  $
  [X,R]=\mathscr L_X R
  $,
  for all $X\in\fk_{\bar 0}$, $R\in\mathfrak{R}$;
\item $[-,-]:\mathfrak{R}\otimes\mathfrak R\to\mathfrak R$ is the
  commutator of two endomorphisms of the spinor bundle;
\item $[-,-]:\fkhat_{\bar 1}\otimes\fkhat_{\bar 1}\to\fkhat_{\bar 0}$ is
  the symmetric map given by
  \begin{equation}
    [\varepsilon,\varepsilon]=(\kappa(\varepsilon,\varepsilon),\rho(\omega^{(3)}(\varepsilon,\varepsilon)))\;,
    \label{eq:oddodd}
  \end{equation}
  where $\varepsilon\in\fkhat_{\bar 1}$, with associated Dirac current
  $\kappa(\varepsilon,\varepsilon)$ and family of self-dual $3$-forms
  \eqref{eq:form3}.
\end{itemize}
The fact that these maps actually take values in $\fkhat$ and define the
structure of a Lie superalgebra on it depends on appropriate conditions
on $H$ and $\varphi$, which we will now start to detail.

\begin{proposition}
  \label{prop:diffcond}
  The maps just introduced define a Lie superalgebra structure on $\fkhat$ if and only if
  \begin{equation}
    \begin{aligned}
      \label{eq:KSAiff1}
      \mathscr L_{\kappa(\varepsilon,\varepsilon)}H=\mathscr
      L_{\kappa(\varepsilon,\varepsilon)}\varphi&=\mathscr
      L_{\kappa(\varepsilon,\varepsilon)}R=0\;,\\
      \eD_{X}(\rho(\omega^{(3)}(\varepsilon,\varepsilon)))&=0\;,
    \end{aligned}
  \end{equation}
  for all $\varepsilon\in\fkhat_{\bar 1}$, $R\in\mathfrak{R}$ and $X\in\mathfrak{X}(M)$. 
\end{proposition}

\begin{proof}
  We first verify $[\fkhat,\fkhat]\subset \fkhat$:
  \begin{itemize}
  \item $[\mathfrak{k}_{\bar0},\mathfrak{k}_{\bar0}]\subset \mathfrak{k}_{\bar0}$,
    $[\mathfrak{R},\mathfrak{R}]\subset \mathfrak{R}$ and $[\mathfrak{k}_{\bar 0},\fkhat_{\bar 1}]\subset\fkhat_{\bar 1}$ are straightforward;
  \item $[\mathfrak{k}_{\bar0},\mathfrak{R}]\subset \mathfrak{R}$. We first remark that
    the Lichnerowicz-Kosmann Lie derivative acts trivially on any constant (=$\overline\nabla$-parallel, cf. the beginning of \S \ref{sec:killing-spinors}) section of $\mathfrak{sp}(\mathcal H)$. Hence $\mathscr L_X(\Gamma(\mathfrak{sp}(\mathcal H)))\subset \Gamma(\mathfrak{sp}(\mathcal H))$ 
    for any $X\in\mathfrak{k}_{\bar 0}$ and the desired inclusion follows from
    \begin{equation}
      \begin{aligned}
        {}[\mathscr L_XR,\mathscr D_Y]&=[[\mathscr L_X,R],\mathscr D_Y]=[\mathscr L_X,[R,\mathscr D_Y]]+[[\mathscr L_X,\mathscr D_Y],R]\\
        &=[[\mathscr L_X,\mathscr D_Y],R]=[\mathscr D_{[X,Y]},R]=0\;,\\
      \end{aligned}
    \end{equation}
    where $R\in\mathfrak{R}$ and $Y\in\mathfrak{X}(M)$;
  \item $[\mathfrak{R},\fkhat_{\bar 1}]\subset\fkhat_{\bar 1}$ follows from 
    \begin{equation}
      \begin{aligned}
        \eD_{Y}(R(\varepsilon))&=\eD_{Y}R(\varepsilon)+R(\eD_Y(\varepsilon))\\
        &=0\;,
      \end{aligned}
    \end{equation}
    where $R\in\mathfrak{R}$, $\varepsilon\in\fkhat_{\bar 1}$ and $Y\in\mathfrak{X}(M)$;
  \item $[\fkhat_{\bar 1},\fkhat_{\bar 1}]\subset \fkhat_{\bar 0}$. We
    already know that the Dirac current of a Killing spinor is a Killing
    vector field, see Corollary \ref{cor:Killing}. The remaining
    conditions are listed in \eqref{eq:KSAiff1}.
  \end{itemize}

  Assuming $[\fkhat,\fkhat]\subset \fkhat$, we now prove that
  $\fkhat=\fkhat_{\bar 0}\oplus\fkhat_{\bar 1}$ with the operation
  $[-,-]$ is a Lie superalgebra. It is easy to see that $\fkhat_{\bar
    0}=\mathfrak{k}_{\bar 0}\ltimes \mathfrak{R}$ is the Lie algebra
  semidirect sum of $\mathfrak{k}_{\bar 0}$ and $\mathfrak{R}$, acting
  on $\fkhat_{\bar 1}$ via a representation of Lie algebras. It remains
  to show $\fkhat_{\bar 0}$-equivariance of \eqref{eq:oddodd} and the
  Jacobi Identity with three odd elements.

  For all $X\in\mathfrak{k}_{\bar 0}$ and $\varepsilon\in\fkhat_{\bar
    1}$, we compare
  \begin{equation}
    \begin{aligned}
      {}[X,[\varepsilon,\varepsilon]]&=[X,\kappa(\varepsilon,\varepsilon)]+[X,\rho(\omega^{(3)}(\varepsilon,\varepsilon))]\\
      &=2\kappa(\mathscr L_X\varepsilon,\varepsilon)+4\mathscr L_{X}(g(\omega^{(3)}(\varepsilon,\varepsilon),H^{-}))\\
      &=2\kappa(\mathscr L_X\varepsilon,\varepsilon)+4g(\mathscr L_{X}(\omega^{(3)}(\varepsilon,\varepsilon)),H^{-})
    \end{aligned}
  \end{equation}
  with
  \begin{equation}
    \begin{aligned}
      \;\,2[\mathscr L_X\varepsilon,\varepsilon]&=2k(\mathscr L_X\varepsilon,\varepsilon)+8g(\omega^{(3)}(\mathscr L_X\varepsilon,\varepsilon),H^{-})\;,
    \end{aligned}
  \end{equation}
  and deduce that $\mathfrak{k}_{\bar 0}$-equivariance of \eqref{eq:oddodd} follows from the identity $
  \mathscr L_{X}(\omega^{(3)}(\varepsilon,\varepsilon))=2\omega^{(3)}(\mathscr L_X\varepsilon,\varepsilon)
  $. We now check this identity. For all $A\in\fsp(1)$,
  $X_1,X_2,X_3\in\mathfrak{X}(M)$, we compute
  \begin{multline*}
    \mathscr L_{X}(\omega^{(3)}_A(\varepsilon,\varepsilon))(X_1,X_2,X_3)\\
    \qquad\;\;\,\quad-2\omega^{(3)}(\mathscr L_X\varepsilon,\varepsilon)(X_1,X_2,X_3)
    =X\left(\varepsilon,(X_1\wedge X_2\wedge X_3)\cdot A\cdot\varepsilon\right)-\left(\varepsilon,(\mathscr L_{X}X_1\wedge X_2\wedge X_3)\cdot A\cdot\varepsilon\right)-\cdots\\
    \qquad\qquad\qquad\qquad\qquad\qquad\qquad\quad\quad\;\;\;\;\;\;\;-\left(\varepsilon,(X_1\wedge X_2\wedge \mathscr L_{X}X_3)\cdot A\cdot\varepsilon\right)-2\left(\nabla_X\varepsilon,(X_1\wedge X_2\wedge X_3)\cdot A\cdot\varepsilon\right)\\
    \qquad\qquad\qquad-2\left(\sigma(A_X)\varepsilon,(X_1\wedge X_2\wedge X_3)\cdot A\cdot\varepsilon\right)\\
    \qquad\qquad\qquad\qquad\qquad\qquad\qquad\qquad\;\;=-2\left(\varepsilon,(X_1\wedge X_2\wedge X_3)\cdot A\cdot\sigma(A_X)\varepsilon\right)
    -\left(\varepsilon,(A_X(X_1)\wedge X_2\wedge X_3)\cdot A\cdot\varepsilon\right)\\
    \qquad\qquad\qquad\;\;\;\,-\cdots-\left(\varepsilon,(X_1\wedge X_2\wedge A_X(X_3))\cdot A\cdot\varepsilon\right)\\
    =0\qquad\quad\qquad\qquad\qquad\qquad\qquad\qquad\qquad\qquad\qquad\qquad\qquad\;\;\;,
  \end{multline*}
  where $A_X=-\nabla X\in\fso(TM)$, $\sigma:\fso(TM)\to\End(\SS)$ is the
  spin representation and the last equation follows from standard
  Clifford identities and the decomposition $\Lambda^2 S=(\Lambda^1
  V\otimes\odot^2\Delta)\oplus\Lambda^3_+V$.

  Now, for all $R\in\mathfrak{R}$ we consider
  \begin{equation}
    \begin{aligned}
     {}[R,[\varepsilon,\varepsilon]] - 2 [[R,\varepsilon],\varepsilon]
     &= [R,\kappa(\varepsilon,\varepsilon)] +
     [R,\rho(\omega^{(3)}(\varepsilon,\varepsilon))] -
     2\kappa(R(\varepsilon),\varepsilon) -
     2\rho(\omega^{(3)}(R(\varepsilon),\varepsilon))\\
      &=-\mathscr L_{\kappa(\varepsilon,\varepsilon)}R -
      2\kappa(R(\varepsilon),\varepsilon) = -\mathscr
      L_{\kappa(\varepsilon,\varepsilon)}R\\
      &=0\;,
    \end{aligned}
  \end{equation}
  where we used \eqref{eq:KSAiff1}. This concludes the proof of
  $\fkhat_{\bar 0}$-equivariance of \eqref{eq:oddodd}.

  The Jacobi identity with three odd elements is equivalent to the
  second cocycle condition of the extended Spencer complex. Indeed, for
  all $\varepsilon\in\fkhat_{\bar 1}$, we have
  \begin{equation}
    \begin{split}
      [[\varepsilon,\varepsilon],\varepsilon]&=\mathscr L_{\kappa}\varepsilon+\rho(\omega)\varepsilon\\
      &=\nabla_{\kappa}\varepsilon-\sigma(\nabla \kappa)\varepsilon+\rho(\omega)\varepsilon\\
      &=\beta_\kappa\varepsilon+\sigma(\gamma(\varepsilon,\varepsilon))\varepsilon+\rho(\omega)\varepsilon\\
      &=0\;,
    \end{split}
  \end{equation}
  where $\kappa=\kappa(\varepsilon,\varepsilon)$ and
  $\omega=\omega^{(3)}(\varepsilon,\varepsilon)$. The proposition is
  proved.
\end{proof}

It is clear from the definition of the space $\mathfrak R$ and
Proposition~\ref{prop:diffcond} that a better understanding of
$\eD$-parallel R-symmetries is required.

\begin{proposition} 
  \label{prop:D-parallel}
  A R-symmetry transformation $R\in\Gamma(\mathfrak{sp}(\mathcal H))$ is
  $\eD$-parallel if and only if
  \begin{itemize}
  \item it is constant, that is $\overline{\nabla}_{X} R=0$ for all
    $X\in\mathfrak{X}(M)$, and
  \item it pointwise commutes  with $\varphi$, that is
    \begin{equation}
      [R|_p,A]=0
    \end{equation}
    at all points $p\in M$ and for all $A\in\varphi(T_pM)\subset \mathfrak{sp}(1)$.
  \end{itemize}
  In particular, if $\varphi(T_pM)$ has dimension greater than or equal
  to $2$ at some fixed point $p\in M$ then any $\eD$-parallel R-symmetry
  transformation is identically zero.
\end{proposition}

\begin{proof}
  We consider the decomposition of the spinor connection
  \begin{equation}
    \eD_X \varepsilon=   D_X \varepsilon + \Phi_{X}\varepsilon
  \end{equation}
  as sum of the metric connection with skew-symmetric torsion 
  \begin{equation}
    D_X Y = \nabla_X Y + 2 h(X,Y),
  \end{equation}
  where $g(h(X,Y),Z)=H(X,Y,Z)$, and the $\varphi$-dependent endomorphism
  of the spinor bundle
  \begin{equation}
    \Phi_{X}\varepsilon=X \cdot  \varphi\cdot\varepsilon + 2 \varphi
    \cdot X\cdot\varepsilon\;.
  \end{equation}
  We note that $\Phi_X$ is a section of $\mathfrak{sp}(\mathcal
  H)\oplus(\Lambda^2 TM\otimes \mathfrak{sp}(\mathcal H))$.

  A straightforward computation says
  \begin{equation}
    \label{eq:Dpar}
    \begin{aligned}
      {}[\eD_X,R](\varepsilon)&=[D_X,R](\varepsilon)+[\Phi_X,R](\varepsilon)\\
      &=\overline{\nabla}_XR(\varepsilon)-[R,\Phi_X](\varepsilon)\;,
    \end{aligned}
  \end{equation}
  whence the R-symmetry transformation $R$ is $\eD$-parallel if and only if
  \begin{equation}
    \label{eq:DparII}
    \begin{aligned}
      \overline{\nabla}_XR&=[R,\Phi_X]\\
      &=X\cdot[R,\varphi]+2[R,\varphi]\cdot X\\
      &=3g(X,[R,\varphi])-X\wedge [R,\varphi]\;,
    \end{aligned}
  \end{equation}
  for all $X\in\mathfrak{X}(M)$. Equation \eqref{eq:DparII} is an
  identity of endomorphisms of the spinor bundle but note that the LHS
  is a section of $\mathfrak{sp}(\mathcal H)$ whereas the RHS of
  $\mathfrak{sp}(\mathcal H)\oplus(\Lambda^2 TM\otimes
  \mathfrak{sp}(\mathcal H))$. Equation \eqref{eq:DparII} then splits
  into
  \begin{equation}
    \label{eq:Dparsplit}
    \begin{aligned}
      \overline{\nabla}_XR&=3g(X,[R,\varphi])\;,\\
      X\wedge [R,\varphi]&=0\;,
    \end{aligned}
  \end{equation}
  for all $X\in\mathfrak{X}(M)$, which implies $[R,\varphi]=0$ and
  $\overline{\nabla}_XR=0$. The first claim of the proposition is
  proved.

  The last claim follows from the fact that $R$ is constant and the
  centraliser of any non-zero element of $\mathfrak{sp}(1)$ is
  $1$-dimensional.
\end{proof}

\begin{corollary}\label{cor:redundantcond}
  $\mathscr L_{X}R=0$ for all Killing vector fields $X$ and $R\in\mathfrak{R}$.
\end{corollary}

We deduce from Propositions \ref{prop:diffcond} and
\ref{prop:D-parallel} that in general only the \emph{decomposable}
$\varphi:TM\to\mathfrak{sp}(1)$ have an associated Killing superalgebra
extended by R-symmetry transformations (in the sense defined in this
section) and that some additional algebraic conditions on the space of
Killing spinors have to be enforced if $\varphi\neq
 0$ (so that $\rho(\omega^{(3}(\varepsilon,\varepsilon))$ pointwise
 commutes with $\varphi$).
 
We will restrict to $\varphi=0$ in what follows. A deeper understanding
of the decomposable case is an interesting problem, which we leave to
future work.

\begin{theorem} \label{thm:KSAII}
  Let $(M,g,H)$ be a lorentzian six-dimensional spin manifold endowed
  with a $ 3$-form $H\in\Omega^3(M)$ and $D_X Y = \nabla_X Y + 2 h(X,Y)$
  the metric connection with skew-symmetric torsion defined by
  $g(h(X,Y),Z)=H(X,Y,Z)$. Let also $H=H^{+}+H^{-}$ be the decomposition
  of $H$ into self-dual and antiself-dual components. If
  \begin{itemize}
  \item $dH=0$ and
  \item $H^{-}$ is $D$-parallel,
  \end{itemize}
  then there exists a natural structure of Lie superalgebra on the
  direct sum $ \fkhat=\fkhat_{\bar 0}\oplus\fkhat_{\bar 1}$ of the
  spaces \eqref{eq:KSA0} and \eqref{eq:KSA1}. We call it the Killing
  superalgebra extended by R-symmetry transformations associated to
  $(M,g,H)$.
\end{theorem}

\begin{proof}
  Due to Propositions
  \ref{prop:diffcond} and \ref{prop:D-parallel} and Corollary
  \ref{cor:redundantcond}, it remains to show that
  $\mathscr L_{\kappa( \varepsilon,\varepsilon)}H=0$ and that
  $\rho(\omega^{(3)}(\varepsilon, \varepsilon))$ is a constant section
  of $\mathfrak{sp}(\mathcal H)$, for all
  $\varepsilon\in\fkhat_{\bar 1}$. We depart with
  \begin{equation}
    \begin{aligned}
      \mathscr L_{\kappa(\varepsilon,\varepsilon)}H &=
      d\imath_{\kappa(\varepsilon,\varepsilon)}H +
      \imath_{\kappa(\varepsilon,\varepsilon)}dH\\
      &=d\imath_{\kappa(\varepsilon,\varepsilon)}H=0\;,
    \end{aligned}
  \end{equation}
  where the last equation follows from Proposition~\ref{cor:LieDF} with
  $\varphi=0$, and then conclude with
  \begin{equation}
    \begin{aligned}
      \overline{\nabla}_{X}(\rho(\omega^{(3)}(\varepsilon,\varepsilon)))
      &= 4X(g(\omega^{(3)}(\varepsilon,\varepsilon),H^{-}))\\
      &= 4g(D_{X}(\omega^{(3)}(\varepsilon,\varepsilon)),H^{-}) +
      4g(\omega^{(3)}(\varepsilon,\varepsilon),D_{X}H^{-})\\
      &=8g(\omega^{(3)}(D_{X}\varepsilon,\varepsilon),H^{-})\\
      &=0\;,
    \end{aligned}
  \end{equation}
  which holds for all $X\in\mathfrak{X}(M)$.
\end{proof}

\begin{remark}
  Let $D$ (resp.
  $D^+$) be the metric connection with skew-symmetric torsion
  $g(h(X,Y),Z)=H(X,Y,Z)$ (resp. $g(h(X,Y),Z)=H^+(X,Y,Z)$). Then it is
  not difficult to see that $DH^{-}=D^{+}H^{-}$, so that the second
  condition in Theorem \ref{thm:KSAII} is equivalent to $H^-$ being
  $D^{+}$-parallel.
\end{remark}

\section{Killing superalgebras (alternative calculation with some indices)}
 \label{sec:killing-superalgebras}

Let $M$ be a six-dimensional spin manifold equipped with a lorentzian
metric $g$, a three-form $H$ and an $\fsp(1)$-valued one-form $\varphi$.
In addition, let the spinor bundle on $M$ be equipped with a connection
${\hat \nabla}$ whose action on a positive chirality spinor field
$\varepsilon$ is defined, with respect to the basis defined in
section~\ref{sec:conventions}, by
\begin{equation}
  \label{eq:gaugedconnection}
  {\hat \nabla}_\mu \varepsilon^A = \nabla_\mu \varepsilon^A + C_\mu{}^A{}_B \varepsilon^B~,
\end{equation}
where $\nabla$ is the Levi-Civita connection and $C$ is a locally
defined $\fsp(1)$-valued one-form on $M$. For any $\Sp(1)$-valued smooth
function $\lambda$ on $M$, the transformations
\begin{equation}
\label{eq:gaugecovariance}
\varepsilon \mapsto \lambda \varepsilon \; , \quad C_\mu \mapsto  - ( \partial_\mu \lambda ) \lambda^{-1} + \lambda C_\mu \lambda^{-1}~,
\end{equation}
imply ${\hat \nabla}_\mu \varepsilon \mapsto \lambda {\hat \nabla}_\mu \varepsilon$. Furthermore, the curvature 
\begin{equation}
\label{eq:gaugecurvature}
G_{\mu\nu} = \partial_\mu C_\nu - \partial_\nu C_\mu + [ C_\mu , C_\nu ]~,
\end{equation}
of $C$ has the transformation
$G_{\mu\nu} \mapsto \lambda G_{\mu\nu} \lambda^{-1}$.

Now recall that ${\mathfrak{S}}(M)$ denotes the space of sections of the
positive chirality spinor bundle on $M$. In terms of the data above,
motivated by Theorem~\ref{thm:spencer-R}, let us call any $\varepsilon
\in {\mathfrak{S}}(M)$ a \emph{Killing spinor} if
\begin{equation}
  \label{eq:KillingSpinorIndices}
  \eD_\mu \varepsilon^A := {\hat \nabla}_\mu \varepsilon^A - \half
  H_{\mu\nu\rho} \Gamma^{\nu\rho} \epsilon^A + 3 \varphi_\mu{}^A{}_B
  \varepsilon^B - \varphi^{\nu \, A}{}_B \Gamma_{\mu\nu} \varepsilon^B =
  0~.
\end{equation}
Notice that \eqref{eq:KillingSpinorIndices} is invariant under
\eqref{eq:gaugecovariance} provided the background fields transform in
the obvious way:
\begin{equation}
  \label{eq:backgroundgaugecovariance}
  g_{\mu\nu} \mapsto g_{\mu\nu} \; , \quad H_{\mu\nu\rho} \mapsto H_{\mu\nu\rho} \; , \quad \varphi_\mu \mapsto \lambda \varphi_\mu \lambda^{-1}~.
\end{equation}
This manifest local $\Sp(1)$ invariance we have engineered is sometimes
referred to as \lq gauging the R-symmetry' in the physics literature.

Now to the construction of the Killing superalgebra. We define a Killing
superalgebra $\fk$ to be a Lie superalgebra whose odd part
$\fk_{\bar 1}$ is precisely the space of Killing spinors defined by
\eqref{eq:KillingSpinorIndices}. The even part $\fk_{\bar 0}$ must
contain elements which act as endomorphisms of $\fk_{\bar 1}$, so that
we may assign a bracket
$[ \fk_{\bar 0},\fk_{\bar 1}] \subset \fk_{\bar 1}$. There are two
obvious candidates: Killing vectors (acting via the spinorial Lie
derivative) and local $\fsp(1)$ R-symmetries. By definition, both these
transformations are endomorphisms of ${\mathfrak{S}}(M)$. But, as we
will see in a moment, to preserve $\fk_{\bar 1}$ will demand some
additional constraints.

Let ${\mathfrak{X}}(M)$ denote the space of vector fields on $M$ and let
us write the subspace of Killing vectors
\begin{equation}
\label{eq:KillingVectors}
{\mathfrak{K}}(M) = \{ X \in {\mathfrak{X}}(M) \, |\, \eL_X g = 0 \}~,
\end{equation}
where, of course, $\eL_X$ denotes the Lie derivative along $X$. For any
$X \in {\mathfrak{K}}(M)$ and $\varepsilon \in {\mathfrak{S}}(M)$, the
spinorial Lie derivative of $\varepsilon$ along $X$ is defined by
\begin{equation}
\label{eq:SpinorialLieDerivative}
\eL_X \varepsilon = \nabla_X \varepsilon + \tfrac{1}{4} d X^\flat \varepsilon~.
\end{equation}

For any $X \in {\mathfrak{K}}(M)$, $Y \in {\mathfrak{X}}(M)$ and
$\Upsilon \in \Omega^\bullet (M)$, as endomorphisms of
${\mathfrak{S}}(M)$, we have the useful identities
\begin{equation}
  \label{eq:SpinorialLieDerivativeNabla}
  [ \eL_X , \nabla_Y ] = \nabla_{[X,Y]} \; , \quad [ \eL_X , \Upsilon ] = \eL_X \Upsilon~,
\end{equation}
where $[X,Y] = \eL_X Y$ is the Lie bracket on ${\mathfrak{X}}(M)$.

Since $\eL_X$ in \eqref{eq:SpinorialLieDerivative} is clearly not
covariant under the local $\Sp(1)$ transformation $\varepsilon \mapsto
\lambda \varepsilon$, let us define a more appropriate gauged version:
\begin{equation}
  \label{eq:SpinorialLieDerivativeGauged}
  {\hat \eL}_X \varepsilon = {\hat \nabla}_X \varepsilon + \tfrac{1}{4} d X^\flat \varepsilon~,
\end{equation}
for any $X \in {\mathfrak{K}}(M)$ and $\varepsilon \in
{\mathfrak{S}}(M)$, which transforms covariantly ${\hat \eL}_X
\varepsilon \mapsto \lambda {\hat \eL}_X \varepsilon$ under
\eqref{eq:gaugecovariance}. The associated identities in
\eqref{eq:SpinorialLieDerivativeNabla} become
\begin{equation}
  \label{eq:SpinorialLieDerivativeNablaGauged}
  [ {\hat \eL}_X , {\hat \nabla}_Y ] = {\hat \nabla}_{[X,Y]} + G(X,Y) \; , \quad [ {\hat \eL}_X , \Upsilon ] = \eL_X \Upsilon~,
\end{equation}
for any $X \in {\mathfrak{K}}(M)$, $Y \in {\mathfrak{X}}(M)$ and
$\Upsilon \in \Omega^\bullet (M)$, where $G$ is the curvature of $C$
from \eqref{eq:gaugecurvature}. Using these identities together with the
definition of $\eD$ in \eqref{eq:KillingSpinorIndices} then yields
\begin{equation}
  \label{eq:SpinorialLieDerivativeKS}
  [ {\hat \eL}_X , \eD_Y ] \varepsilon = \eD_{[X,Y]} \varepsilon +
  G(X,Y) \varepsilon - \iota_Y ( \eL_X H ) \varepsilon + 3 ( {\hat
    \eL}_X \varphi )(Y) \varepsilon - Y \wedge ( {\hat \eL}_X \varphi )
  \varepsilon~,
\end{equation}
for any $X \in {\mathfrak{K}}(M)$, $Y \in {\mathfrak{X}}(M)$ and
$\varepsilon \in {\mathfrak{S}}(M)$, where ${\hat \eL}_X \varphi = \eL_X
\varphi + [ \iota_X C , \varphi ]$. Consequently, for any $X \in
{\mathfrak{K}}(M)$ and $\varepsilon \in \fk_{\bar 1}$, we see that
${\hat \eL}_X \varepsilon \in \fk_{\bar 1}$ is guaranteed provided
\begin{equation}
  \label{eq:SpinorialLieDerivativeBackgroundData}
  \iota_X G=0 \; , \quad \eL_X H = 0 \; , \quad {\hat \eL}_X \varphi = 0~.
\end{equation}
Henceforth we shall define
\begin{equation}
  \label{eq:KillingVectorsBackgroundPreserving}
{\mathfrak{K}} = \{ X \in {\mathfrak{X}}(M) \, |\, \eL_X g = 0 , \eL_X H
= 0 , {\hat \eL}_X \varphi = 0 , \iota_X G = 0 \}~,
\end{equation}
as a natural subspace of Killing vectors which preserve the background.
Notice that the Bianchi identity $d^{\hat \nabla} G =0$ implies that
${\hat \eL}_X G = d^{\hat \nabla} \iota_X G = 0$, for any
$X \in {\mathfrak{K}}$. Furthermore, if at every point in $M$ the
elements in ${\mathfrak{K}}$ span the tangent space of $M$ (in which
case $M$ is locally homogeneous) then we must have $G=0$.

Now to the local R-symmetries. For any $\fsp(1)$-valued smooth function
$\rho$ on $M$ and any $\varepsilon \in {\mathfrak{S}}(M)$, it is easy to
verify that
\begin{equation}
  \label{eq:RSymmetryKS}
  [ \eD_\mu  , \rho ] \varepsilon = ( {\hat \nabla}_\mu \rho + 3 [ \varphi_\mu , \rho ]  - \Gamma_{\mu\nu} [ \varphi^\nu , \rho]) \varepsilon~.
\end{equation}
Thus, for any $\varepsilon \in \fk_{\bar 1}$, we see that $\rho
\varepsilon \in \fk_{\bar 1}$ is guaranteed provided
\begin{equation}
  \label{eq:RSymmetryBackgroundData}
  {\hat \nabla}_\mu \rho =0 \; , \quad [ \varphi_\mu , \rho ] = 0~.
\end{equation}
Henceforth we shall define
\begin{equation}
  \label{eq:RSymmetriesBackgroundPreserving}
  {\mathfrak{R}} = \{ \rho \in C^\infty (M) \otimes \fsp(1) \, |\, {\hat \nabla} \rho =0 , [ \varphi , \rho ] = 0 \}~,
\end{equation}
as a natural subspace of local R-symmetries which preserve the
background.

Next, if we are to identify $\fk_{\bar 0}$ with (a subspace of)
${\mathfrak{K}} \oplus {\mathfrak{R}}$, we need to define the brackets
and check the Jacobi identities for $\fk$.

The $[{\bar 0}{\bar 0}{\bar 0}]$ component of the Jacobi identity just
says that $\fk_{\bar 0}$ must be a Lie algebra.

For any $X,Y \in {\mathfrak{X}}(M)$, as endomorphisms of the tensor
bundle on $M$, the commutator of Lie derivatives $\eL_X$ and $\eL_Y$
obeys the identity
\begin{equation}
\label{eq:LieXY}
[ \eL_X , \eL_Y ] = \eL_{[X,Y]}~,
\end{equation}
where $[X,Y] = \eL_X Y$ is the Lie bracket on ${\mathfrak{X}}(M)$.
Furthermore, it is easily verified that
\begin{equation}
  \label{eq:LieXYPhi}
[ {\hat \eL}_X , {\hat \eL}_Y ] \varphi = {\hat \eL}_{[X,Y]} \varphi + [
G(X,Y) , \varphi ] \; , \quad [ {\hat \eL}_X , \iota_Y ] G =
\iota_{[X,Y]} G~,
\end{equation}
for any $X,Y \in {\mathfrak{X}}(M)$. Thus, for any $X,Y \in
{\mathfrak{K}}$, \eqref{eq:LieXY} implies $\eL_{[X,Y]} g = 0$ and
$\eL_{[X,Y]} H = 0$, while \eqref{eq:LieXYPhi} implies ${\hat
  \eL}_{[X,Y]} \varphi =0$ (using $\iota_X G = 0$) and $\iota_{[X,Y]} G
=0$ (using ${\hat \eL}_X G =0$ and $\iota_Y G = 0$). It follows that
${\mathfrak{K}}$ is indeed a Lie algebra with respect to the Lie bracket
of vector fields.

For any $\rho , \rho^\prime \in C^\infty (M) \otimes \fsp(1)$, as
endomorphisms of ${\mathfrak{S}}(M)$, clearly the commutator $[ \rho ,
\rho^\prime ] = \rho \rho^\prime - \rho^\prime \rho$ obeys
\begin{equation}
  \label{eq:DPhiRho}
  {\hat \nabla} [ \rho , \rho^\prime ]  = [ {\hat \nabla} \rho ,
  \rho^\prime ] + [ \rho ,  {\hat \nabla} \rho^\prime ] \; , \quad [
  \varphi , [ \rho , \rho^\prime ] ] = [ [ \varphi ,\rho ] , \rho^\prime
  ] + [ \rho , [ \varphi , \rho^\prime ] ]~.
\end{equation}
Whence, any $\rho , \rho^\prime \in {\mathfrak{R}}$ have $[ \rho , \rho^\prime ] \in {\mathfrak{R}}$ and ${\mathfrak{R}}$ is clearly a Lie algebra with respect to the commutator of endomorphisms.

So ${\mathfrak{K}} \oplus {\mathfrak{R}}$ is certainly a Lie algebra if
we define $[ {\mathfrak{K}} , {\mathfrak{R}} ] =0$. Indeed, even if we
had defined $[ X , \rho ] = {\hat \nabla}_X \rho$, for any
$X \in {\mathfrak{K}}$ and $\rho \in {\mathfrak{R}}$, the condition
${\hat \nabla} \rho =0$ in \eqref{eq:RSymmetriesBackgroundPreserving}
would force us to take $[ {\mathfrak{K}} , {\mathfrak{R}} ] =0$.

The $[{\bar 0}{\bar 0}{\bar 1}]$ component of the Jacobi identity says
that $\fk_{\bar 0}$ must act on $\fk_{\bar 1}$ as a $\fk_{\bar
  0}$-module. For any $X,Y \in {\mathfrak{K}}(M)$ and $\varepsilon \in
{\mathfrak{S}}(M)$, one finds that the commutator of gauged spinorial
Lie derivatives in \eqref{eq:SpinorialLieDerivativeGauged} obeys the
identity
\begin{equation}
  \label{eq:LieXYSpinorial}
  [ {\hat \eL}_X , {\hat \eL}_Y ] \varepsilon = {\hat \eL}_{[X,Y]} \varepsilon + G(X,Y) \varepsilon~.
\end{equation}
Whence, for any $X \in {\mathfrak{K}}$ and $\varepsilon \in \fk_{\bar
  1}$, the bracket
\begin{equation}
\label{eq:XEpsilonBracket}
[ X , \varepsilon ] = {\hat \eL}_{X} \varepsilon~,
\end{equation}
defines $\fk_{\bar 1}$ as a ${\mathfrak{K}}$-module (since $\iota_X G
=0$). Moreover, for any $\rho \in {\mathfrak{R}}$ and $\varepsilon \in
\fk_{\bar 1}$, the bracket
\begin{equation}
  \label{eq:RhoEpsilonBracket}
[ \rho , \varepsilon ] = \rho \varepsilon~,
\end{equation}
clearly defines $\fk_{\bar 1}$ as an ${\mathfrak{R}}$-module by
restricting local $\fsp(1)$ endomorphisms of the spinor bundle. Finally,
combining \eqref{eq:XEpsilonBracket} and \eqref{eq:RhoEpsilonBracket},
we see that
\begin{equation}
  \label{eq:XRhoEpsilon}
  [ X , [ \rho , \varepsilon ] ] - [ \rho , [ X , \varepsilon ] ] = [ {\hat \eL}_X , \rho ] \varepsilon = ( {\hat \nabla}_X \rho ) \varepsilon = 0~,
\end{equation}
as required, for any $X \in {\mathfrak{K}}$, $\rho \in {\mathfrak{R}}$
and $\varepsilon \in \fk_{\bar 1}$ (since ${\hat \nabla} \rho = 0$).

This has established that $\fk_{\bar 1}$ is indeed a representation of
the Lie algebra ${\mathfrak{K}} \oplus {\mathfrak{R}}$ with respect to
the action defined by \eqref{eq:XEpsilonBracket} and
\eqref{eq:RhoEpsilonBracket}.

In order to check the remaining Jacobi identities for $\fk$, we must
first specify a bracket $[ \fk_{\bar 1} , \fk_{\bar 1} ] \subset
\fk_{\bar 0}$. Since the odd-odd bracket for $\fk$ is symmetric, it is
sufficient to define
\begin{equation}
  \label{eq:Epsilon2Bracket}
  [ \varepsilon , \varepsilon ] = ( \kappa (\varepsilon) , \vartheta (\varepsilon) )~,
\end{equation}
for all $\varepsilon \in \fk_{\bar 1}$, such that $\kappa (\varepsilon)
\in {\mathfrak{K}}$ and $\vartheta (\varepsilon) \in {\mathfrak{R}}$.
The bracket of two different $\varepsilon , \varepsilon^\prime \in
\fk_{\bar 1}$ is then defined by polarisation:
\begin{equation}
  \label{eq:EpsilonPrimeBracket}
[ \varepsilon , \varepsilon^\prime ] = \half ( [ \varepsilon +
\varepsilon^\prime , \varepsilon + \varepsilon^\prime ] - [ \varepsilon
, \varepsilon ] - [ \varepsilon^\prime , \varepsilon^\prime ] )~,
\end{equation}
and it is convenient to define $\kappa (\varepsilon ,
\varepsilon^\prime) = \half ( \kappa (\varepsilon + \varepsilon^\prime)
- \kappa (\varepsilon) - \kappa (\varepsilon^\prime) )$ and $\vartheta
(\varepsilon , \varepsilon^\prime ) = \half ( \vartheta (\varepsilon +
\varepsilon^\prime) - \vartheta (\varepsilon) - \vartheta
(\varepsilon^\prime) )$.

Guided again by Theorem~\ref{thm:spencer-R}, let us consider the following choices:
\begin{equation}
  \label{eq:Epsilon2Data}
\kappa (\varepsilon)^\mu = \epsilon_{AB} {\overline \varepsilon}^A
\Gamma^\mu \varepsilon^B \; , \quad \vartheta (\varepsilon)^{AB} =
\tfrac{2}{3} H^{\mu\nu\rho} \, {\overline \varepsilon}^A
\Gamma_{\mu\nu\rho} \varepsilon^B~.
\end{equation}
For a given $\varepsilon \in \fk_{\bar 1}$, it will sometimes be
convenient to drop the parenthetical $\varepsilon$ in
\eqref{eq:Epsilon2Data} and write $\kappa^\mu = \epsilon_{AB} {\overline
  \varepsilon}^A \Gamma^\mu \varepsilon^B$ and $\omega^{AB}_{\mu\nu\rho}
= {\overline \varepsilon}^A \Gamma_{\mu\nu\rho} \varepsilon^B$ for the
Killing spinor bilinears.

Clearly \eqref{eq:Epsilon2Data} defines $\kappa \in {\mathfrak{X}}(M)$
and $\vartheta \in C^\infty (M) \otimes \fsp(1)$. However, for $\kappa
\in {\mathfrak{K}}$ and $\vartheta \in {\mathfrak{R}}$, we require all
of the following conditions to be satisfied:
\begin{equation}
  \label{eq:KappaThetaBackgroundPreserving}
\eL_\kappa g = 0 \; , \quad \eL_\kappa H = 0 \; , \quad {\hat
  \eL}_\kappa \varphi = 0 \; ,\quad  \iota_\kappa G = 0 \; , \quad {\hat
  \nabla} \vartheta =0 \; , \quad [ \varphi , \vartheta ] = 0~.
\end{equation}
We shall return to the important matter of checking whether these
conditions are actually satisfied in a moment but first let us just
assume that they are and move on to confirm the remaining Jacobi
identities.

The $[{\bar 0}{\bar 1}{\bar 1}]$ component of the Jacobi identity says
that the odd-odd bracket on $\fk$ must define a $\fk_{\bar
  0}$-equivariant map $\fk_{\bar 1} \otimes \fk_{\bar 1} \rightarrow
\fk_{\bar 0}$. This means that the $\fk_{\bar 1} \otimes \fk_{\bar 1}
\rightarrow {\mathfrak{K}}$ part must be ${\mathfrak{K}}$-equivariant
and ${\mathfrak{R}}$-invariant (since $[ {\mathfrak{R}} , {\mathfrak{K}}
] =0$) while the $\fk_{\bar 1} \otimes \fk_{\bar 1} \rightarrow
{\mathfrak{R}}$ part must be ${\mathfrak{R}}$-equivariant and
${\mathfrak{K}}$-invariant (since $[ {\mathfrak{K}} , {\mathfrak{R}} ]
=0$).

For any $X \in {\mathfrak{K}}(M)$ and $\varepsilon \in
{\mathfrak{S}}(M)$, we have the identities
\begin{equation}
\label{eq:LieXEpsilon2K}
[ X , \kappa ]  = 2 \kappa ({\hat \eL}_X \varepsilon , \varepsilon) \; ,
\quad {\hat \nabla}_X \vartheta = 2 \vartheta ( {\hat \eL}_X
\varepsilon,\varepsilon) + \tfrac{2}{3} ( \eL_X H )_{\mu\nu\rho}
\omega^{\mu\nu\rho}~.
\end{equation}
The first identity above guarantees that $\fk_{\bar 1} \otimes \fk_{\bar
  1} \rightarrow {\mathfrak{K}}$ is ${\mathfrak{K}}$-equivariant.
Moreover, if $X \in {\mathfrak{K}}$ and $\vartheta \in {\mathfrak{R}}$
then $\eL_X H =0$ and ${\hat \nabla \vartheta} =0$, in which case the
second identity in \eqref{eq:LieXEpsilon2K} says that $\fk_{\bar 1}
\otimes \fk_{\bar 1} \rightarrow {\mathfrak{R}}$ is indeed
${\mathfrak{K}}$-invariant.

For any $\rho \in C^\infty (M) \otimes \fsp(1)$ and $\varepsilon \in
{\mathfrak{S}}(M)$, we also have the identities
\begin{equation}
  \label{eq:LieRhoEpsilon2K}
[ \rho , \vartheta ]  = 2 \vartheta ( \rho \varepsilon,\varepsilon) \; ,
\quad  \kappa ( \rho \varepsilon,\varepsilon) =0~.
\end{equation}
The first identity above shows that $\fk_{\bar 1} \otimes \fk_{\bar 1}
\rightarrow {\mathfrak{R}}$ is ${\mathfrak{R}}$-equivariant while the
second identity shows that $\fk_{\bar 1} \otimes \fk_{\bar 1}
\rightarrow {\mathfrak{K}}$ is ${\mathfrak{R}}$-invariant.

Whence, at least if $\kappa \in {\mathfrak{K}}$ and $\vartheta \in
{\mathfrak{R}}$, we have shown that the $[{\bar 0}{\bar 1}{\bar 1}]$
component of the Jacobi identity is satisfied.

The final $[{\bar 1}{\bar 1}{\bar 1}]$ component of the Jacobi identity
is equivalent (via polarisation) to the condition
\begin{equation}
  \label{eq:Epsilon3}
  [ [ \varepsilon , \varepsilon ] , \varepsilon] = 0~,
\end{equation}
for all $\varepsilon \in \fk_{\bar 1}$. To examine this condition more
closely, it is worth noting the identity
\begin{equation}
  \label{eq:NablaKappa}
\nabla_\mu \kappa_\nu = 2 H_{\mu\nu\rho} \kappa^\rho -2
\varphi^\rho_{AB} \omega^{AB}_{\mu\nu\rho}~,
\end{equation}
which can be derived from the definition of $\kappa$ in
\eqref{eq:Epsilon2Data} using the Killing spinor equation
\eqref{eq:KillingSpinorIndices}. Using the brackets defined by
\eqref{eq:XEpsilonBracket}, \eqref{eq:RhoEpsilonBracket},
\eqref{eq:Epsilon2Bracket} and \eqref{eq:Epsilon2Data}, the left hand
side of \eqref{eq:Epsilon3} reads
\begin{equation}
  \label{eq:Epsilon32}
( {\hat \eL}_\kappa \varepsilon + \vartheta \varepsilon )^A = \kappa^\mu
{\hat \nabla}_\mu \varepsilon^A + \tfrac{1}{4} ( \nabla_\mu \kappa_\nu )
\Gamma^{\mu\nu} \varepsilon^A + \tfrac{1}{3} ( H^{\mu\nu\rho} - {\tilde
  H}^{\mu\nu\rho} ) \omega_{\mu\nu\rho}{}^A{}_B \varepsilon^B ~.
\end{equation}
We have used the identity $H^{\mu\nu\rho} \omega^{AB}_{\mu\nu\rho} =
H^{\mu\nu\rho} {\tilde \omega}^{AB}_{\mu\nu\rho} = - {\tilde
  H}^{\mu\nu\rho}  \omega^{AB}_{\mu\nu\rho}$ to identify the
contribution of $\vartheta^A{}_B$ in the third term of
\eqref{eq:Epsilon32} with $\rho(s,s)^A{}_B$ in
Theorem~\ref{thm:spencer-R}. Moreover, notice that \eqref{eq:NablaKappa}
allows us to identify the contribution of $\nabla_\mu \kappa_\nu$ in the
second term of \eqref{eq:Epsilon32} with $\gamma(s,s)_{\mu\nu}$ in
Theorem~\ref{thm:spencer-R}. Finally, the Killing spinor equation
\eqref{eq:KillingSpinorIndices} allows us to identify the contribution
of ${\hat \nabla}_\mu \varepsilon^A$ in the first term of
\eqref{eq:Epsilon32} with $( \beta_\mu s)^A$ in
Theorem~\ref{thm:spencer-R}. The vanishing of \eqref{eq:Epsilon32} is
therefore precisely equivalent to the second cocycle condition that was
already established in the proof of Theorem~\ref{thm:spencer-R}.

In summary, $\fk$ is a indeed a Lie superalgebra with respect to the
brackets we have chosen above provided every Killing spinor has $\kappa$
and $\vartheta$ in \eqref{eq:Epsilon2Data} obeying all the conditions in
\eqref{eq:KappaThetaBackgroundPreserving}. Let us now return to resolve
these conditions. For simplicity, we shall assume henceforth that the
connection $C$ is flat (i.e. $G=0$).

The condition $\eL_\kappa g =0$ in
\eqref{eq:KappaThetaBackgroundPreserving} follows immediately from
\eqref{eq:NablaKappa} (since \eqref{eq:NablaKappa} implies $\nabla
\kappa^\flat = \half d \kappa^\flat)$.

The condition $[ \varphi_\mu , \vartheta ] =0$ in
\eqref{eq:KappaThetaBackgroundPreserving} implies that, at each point in
$M$, either $\vartheta = 0$ or else $\varphi_\mu$ must be proportional
to $\vartheta$. Having assumed that $C$ is flat, if the condition ${\hat
  \nabla}_\mu \vartheta =0$ in \eqref{eq:KappaThetaBackgroundPreserving}
is satisfied, then we can always fix a gauge (i.e. for an appropriate
local $\Sp(1)$ transformation) in which $\vartheta$ is constant. So
either $\vartheta$ is identically zero on $M$ or else $\varphi_\mu =
\psi_\mu \vartheta$, for some $\psi \in \Omega^1 (M)$.

To make further progress, we now require the identity 
\begin{equation}
\label{eq:NablaOmega}
\begin{split}
{\hat \nabla}_\mu \omega^{AB}_{\nu\rho\sigma} =\; & 6 H_{\mu
  [\nu}{}^\tau \omega^{AB}_{\rho\sigma ]\tau} - 6 \varphi_\mu^{(A}{}_C
\omega^{B)C}_{\nu\rho\sigma} - 6 \varphi_{[\nu}^{(A}{}_C
\omega^{B)C}_{\rho\sigma ]\mu} \\
&+ 6 g_{\mu [\nu} ( \varphi_\rho^{AB} \kappa_{\sigma ]} + \varphi^{\tau
  \, (A}{}_C \omega^{B)C}_{\rho\sigma ]\tau} ) -
\varepsilon_{\mu\nu\rho\sigma\tau\theta} \varphi^{\tau \, AB}
\kappa^\theta~, 
\end{split}
\end{equation}
which can be derived (with some effort) from the definition of $\omega$
below \eqref{eq:Epsilon2Data} using the Killing spinor equation
\eqref{eq:KillingSpinorIndices}. Parentheses around indices denote
symmetrisation while brackets denote skew-symmetrisation (with weight
one in both cases). Skew-symmetrising $[\mu\nu\rho\sigma]$ in
\eqref{eq:NablaOmega} gives a useful subsidiary identity
\begin{equation}
  \label{eq:DOmega}
  {\hat \nabla}_{[\mu} \omega^{AB}_{\nu\rho\sigma ]} = 6
  H^-_{[\mu\nu}{}^\tau \omega^{AB}_{\rho\sigma ]\tau} -
  \varepsilon_{\mu\nu\rho\sigma\tau\theta} \varphi^{\tau \, AB} \kappa^\theta~,  
\end{equation}
where $H^\pm = \half ( H \pm {\tilde H}) \in \Omega^3_\pm (M)$ denotes
the self-dual and anti-self-dual projections of $H$. The fact that only
$H^-$ appears in \eqref{eq:DOmega} is due to the identity
\begin{equation}
  \label{eq:SD1}
  X^\pm_{[ \mu\nu}{}^\tau Y^\pm_{\rho\sigma ]\tau} = 0~,  
\end{equation}
which holds for any $X^\pm , Y^\pm \in \Omega^3_\pm (M)$ with the same
chirality (c.f. Lemma~\ref{lem:sd2}).

The identity \eqref{eq:DOmega} defines $d^{\hat \nabla} \omega$. Having
assumed that $G=0$, acting with $\star d^{\hat \nabla}$ on $d^{\hat
  \nabla} \omega$ must give zero. Using \eqref{eq:DOmega} together with
\eqref{eq:NablaOmega} and \eqref{eq:NablaKappa} to evaluate this
operation yields (after some simplification) the following expression
for the (gauged) Lie derivative of $\varphi$ along $\kappa$:
\begin{equation}
  \label{eq:LKappaPhi}
  \begin{split}
    {\hat \eL}_\kappa \varphi^{AB}_{\mu} =\; & \kappa_\mu {\hat \nabla}^\nu \varphi_\nu^{AB} + 2 H^-_{\mu\nu\rho} \varphi^{\nu\, AB} \kappa^\rho + \varphi_\mu^{(A}{}_C \vartheta^{B)C} + H^{-\, \nu\rho\sigma} \varphi_\sigma^{(A}{}_C \omega^{B)C}_{\mu\nu\rho} \\
    &+\tfrac{1}{6} ( \nabla_\mu H^-_{\nu\rho\sigma} - 6 H^+_{\mu [\nu}{}^\tau H^-_{\rho\sigma ]\tau} ) \omega^{AB\, \nu\rho\sigma}  - g^{\sigma\tau} ( \nabla_\tau H^-_{\nu\rho\sigma} - 6 H^+_{\tau [\nu}{}^\theta H^-_{\rho\sigma ]\theta} ) \omega_\mu^{AB\, \nu\rho}~. 
  \end{split}
\end{equation}
Notice that the term $\varphi_\mu^{(A}{}_C \vartheta^{B)C}$ in the first
line vanishes as a consequence of the condition $[ \varphi , \vartheta ]
=0$ in \eqref{eq:KappaThetaBackgroundPreserving}. The remaining terms in
\eqref{eq:LKappaPhi} would vanish identically if $\varphi$ and $H$ obey
\begin{equation}
  \label{eq:DStarPhi}
  {\hat \nabla}^\nu \varphi_\nu^{AB} = 0 \; , \quad H^-_{\mu\nu\rho}
  \varphi^{\rho\, AB} =0 \; , \quad \nabla_\mu H^-_{\nu\rho\sigma} - 6
  H^+_{\mu [\nu}{}^\tau H^-_{\rho\sigma ]\tau} =0~.
\end{equation}
To articulate the third condition in \eqref{eq:DStarPhi} more easily,
let us define the connection $\nabla^+_X Y = \nabla_X Y + 2 h^+ (X,Y)$,
with skew-symmetric torsion defined by $g(h^+(X,Y),Z) = H^+ (X,Y,Z)$,
for all $X,Y,Z \in {\mathfrak{X}} (M)$. The third condition in
\eqref{eq:DStarPhi} just says that $\nabla^+ H^- =0$.

The action of ${\hat \nabla}$ on $\vartheta$ can be evaluated using the
definition \eqref{eq:Epsilon2Data} together with the identity
\eqref{eq:NablaOmega}. After some simplification, and applying the
condition $[ \varphi , \vartheta ] =0$, this gives
\begin{equation}
  \label{eq:DTheta}
  {\hat \nabla}_\mu \vartheta^{AB} = ( \nabla_\mu H^-_{\nu\rho\sigma} -
  6 H^+_{\mu [\nu}{}^\tau H^-_{\rho\sigma ]\tau} ) \omega^{AB\,
    \nu\rho\sigma} + 12 H^-_{\mu\nu\rho} \varphi^{\nu\, AB} \kappa^\rho
  -12 H^{-\, \nu\rho\sigma} \varphi_\sigma^{(A}{}_C
  \omega^{B)C}_{\mu\nu\rho}~.
\end{equation}
So ${\hat \nabla} \vartheta =0$ if $\varphi$ and $H$ obey the second two
conditions in \eqref{eq:DStarPhi}.

To summarise, thus far we have shown that the three conditions on
$\varphi$ and $H$ in \eqref{eq:DStarPhi}, together with $[ \varphi ,
\vartheta ] =0$ and $G=0$, are sufficient to guarantee that all the
conditions except $\eL_\kappa H =0$ in
\eqref{eq:KappaThetaBackgroundPreserving} are satisfied.

Taking the exterior derivative of the exact two-form defined by
\eqref{eq:NablaKappa} (i.e. $d^2 \kappa^\flat =0$) and using the
identity \eqref{eq:NablaOmega} provides us with the following expression
for the Lie derivative of $H$ along $\kappa$:
\begin{equation}
  \label{eq:LKappaH}
\eL_\kappa H_{\mu\nu\rho} = - 4 \kappa^\sigma \nabla_{[\mu}
H_{\nu\rho\sigma]} + 3 \omega^{AB}_{[\mu\nu}{}^\sigma ( {\hat
  \nabla}_{\rho ]} \varphi_\sigma + 2 [ \varphi_{\rho ]} ,
\varphi_\sigma ]  + 2 H_{\rho ]\sigma\tau} \varphi^\tau )_{AB}  +24
H^-_{[\mu\nu}{}^\tau \omega^{AB}_{\rho\sigma ] \tau} \varphi^\sigma_{AB}~.
\end{equation}
Even if we demand that $H$ is closed (so that the first term in
\eqref{eq:LKappaH} vanishes identically), in general the three
conditions on $\varphi$ and $H$ in \eqref{eq:DStarPhi} will not be
sufficient to guarantee that $\eL_\kappa H = 0$. In order to proceed, we
will now consider two special cases that yield distinct branches of
solutions of all the conditions in
\eqref{eq:KappaThetaBackgroundPreserving}. That is, of course, not to say that all solutions must necessarily lie on one of these two branches. It is merely a simplifying assumption we shall make in order to find interesting solutions. Having said that, in Section~\ref{sec:maxim-supersymm-back}, we will discover that all the maximally supersymmetric backgrounds do actually lie on one of these two branches of solutions.  

The first branch is defined by taking
\begin{equation}
\label{eq:BranchOne}
d^{\hat \nabla} {\star} \varphi = 0 \; , \quad H^- = 0 \; , \quad dH = 0~. 
\end{equation}
The first two conditions above ensure \eqref{eq:DStarPhi} are satisfied
and therefore all the conditions except $\eL_\kappa H = 0$ in
\eqref{eq:KappaThetaBackgroundPreserving} are guaranteed. Notice that
$H$ being self-dual implies that $\vartheta =0$ identically (i.e.
$[\fk_{\bar 1},\fk_{\bar 1}] \subset {\mathfrak{K}}$). From
\eqref{eq:LKappaH}, we see that
\begin{equation}
\label{eq:LKappaHBranchOne}
\eL_\kappa H_{\mu\nu\rho} = 3 \omega^{AB}_{[\mu\nu}{}^\sigma ( {\hat
  \nabla}_{\rho ]} \varphi_\sigma + 2 [ \varphi_{\rho ]} ,
\varphi_\sigma ]  + 2 H_{\rho ]\sigma\tau} \varphi^\tau )_{AB}~,
\end{equation}
and it is not obvious that the right hand side is zero. To prove that
this is in fact the case, notice that $\eL_\kappa \kappa = [ \kappa ,
\kappa ] =0$ and ${\hat \eL}_\kappa \omega = ( \iota_\kappa d^{\hat
  \nabla} + d^{\hat \nabla}  \iota_\kappa ) \omega = 0$ (using
$\iota_\kappa d^{\hat \nabla} \omega = 0$ from \eqref{eq:DOmega} and
$\iota_\kappa \omega =0$ from Lemma~\ref{lem:kwzero}). Therefore,
because we have already ensured that ${\hat \eL}_\kappa \varphi = 0$,
using \eqref{eq:NablaKappa} to evaluate $\eL_\kappa \nabla \kappa^\flat$
and \eqref{eq:NablaOmega} to evaluate ${\hat \eL}_\kappa {\hat \nabla}
\omega$, we deduce that
\begin{equation}
  \label{eq:LKappaHBranchOneConditions}
\kappa^\mu \eL_\kappa H_{\mu\nu\rho} = 0 \; , \quad \eL_\kappa H_{\mu [
  \nu}{}^\tau \omega^{AB}_{\rho\sigma ]\tau} = 0~.
\end{equation}
But since $\eL_\kappa H$ and $\omega^{AB}$ are self-dual three-forms,
the identity \eqref{eq:SD1} implies that the second condition above is
equivalent to
\begin{equation}
  \label{eq:LKappaHBranchOneConditionsEquivalent}
  \omega^{AB}_{\mu [ \nu}{}^\tau {\eL_\kappa H}_{\rho\sigma ]\tau} = 0~.
\end{equation}
So we have self-dual three-forms $\eL_\kappa H$  and $\omega^{AB}$
obeying \eqref{eq:LKappaHBranchOneConditionsEquivalent}, with
$\iota_\kappa \eL_\kappa H = 0$ (from
\eqref{eq:LKappaHBranchOneConditions}) and $\iota_\kappa \omega^{AB} =
0$ (from Lemma~\ref{lem:kwzero}). Furthermore, the Killing spinor
bilinears $\kappa$ and $\omega^{AB}$ are nowhere vanishing and $\kappa$
is everywhere null. From Lemma~\ref{lem:proportional} (identifying
$\alpha = \eL_\kappa H$, $\beta = \omega^{AB}$ and $N = \kappa$), it
therefore follows that $\eL_\kappa H$ must equal some locally defined
function multiplying $\omega^{AB}$, for any choice of $A$ and $B$. In
particular, we must have
\begin{equation}
  \label{eq:LKappaHBranchOneProportional}
\eL_\kappa H = f_1 \omega^{11} = f_2 \omega^{12} = f_3 \omega^{22}~,
\end{equation}
for some locally defined functions $f_1$, $f_2$ and $f_3$. However, at any point in $M$, we know from Lemma~\ref{lem:linindep} that $\omega^{11}$, $\omega^{12}$ and $\omega^{22}$ are linearly independent. Therefore \eqref{eq:LKappaHBranchOneProportional} can only be true if $\eL_\kappa H =0$, as required. So the
conditions \eqref{eq:BranchOne} for this branch do indeed imply
\eqref{eq:KappaThetaBackgroundPreserving} and therefore guarantee the
existence of a Killing superalgebra. Since $\vartheta = 0$ here, it is
possible to define a Killing superalgebra $\fk$ with $\fk_{\bar 0} =
{\mathfrak{K}}$, i.e. ignoring ${\mathfrak{R}}$ completely. The Killing
superalgebras in this case are therefore naturally associated with the
Spencer cohomology calculation that led to Theorem~\ref{thm:spencer},
where we ignored the R-symmetry.   

% However,
% making use of the Fierz identity \eqref{eq:Fierz}, it is a simple matter
% to deduce the identity
% \begin{equation}
%   \label{eq:Omega2Identity}
%   \omega^{AB}_{\mu [ \nu}{}^\tau \omega^{CD}_{\rho\sigma ]\tau} = \tfrac{1}{3} \kappa_\mu (  \epsilon^{A(C} \omega^{D)B}_{\nu\rho\sigma} + \epsilon^{B(C} \omega^{D)A}_{\nu\rho\sigma} )~.
% \end{equation}
% Consequently,
% \begin{equation}
% \label{eq:Omega2Identity111222}
% \omega^{11}_{\mu [ \nu}{}^\tau \omega^{12}_{\rho\sigma ]\tau} = -
% \tfrac{1}{3} \kappa_\mu \omega^{11}_{\nu\rho\sigma} \; , \quad
% \omega^{22}_{\mu [ \nu}{}^\tau \omega^{12}_{\rho\sigma ]\tau} =
% \tfrac{1}{3} \kappa_\mu \omega^{22}_{\nu\rho\sigma} \; , \quad
% \omega^{11}_{\mu [ \nu}{}^\tau \omega^{22}_{\rho\sigma ]\tau} = -
% \tfrac{2}{3} \kappa_\mu \omega^{12}_{\nu\rho\sigma}~,
% \end{equation}
% are all nowhere vanishing. Therefore substituting, say, the relation
% $\eL_\kappa H = f_2 \omega^{12}$ from
% \eqref{eq:LKappaHBranchOneProportional} back into
% \eqref{eq:LKappaHBranchOneConditionsEquivalent} with $A=B=1$ would imply
% (using the first identity in \eqref{eq:Omega2Identity111222}) that $f_2
% =0$. 

The second branch is defined by taking
\begin{equation}
\label{eq:BranchTwo}
\varphi = 0 \; , \quad \nabla^+ H^- = 0 \; , \quad dH = 0~. 
\end{equation}
The first two conditions above ensure that \eqref{eq:DStarPhi} are
satisfied while the third condition ensures that $\eL_\kappa H = 0$.
Thus, all the conditions in \eqref{eq:KappaThetaBackgroundPreserving}
are satisfied and the existence of a Killing superalgebra is guaranteed.
Of course, if $H^- \neq 0$, we could have $\vartheta \neq 0$ for one or
more Killing spinors on this branch, in which case ${\mathfrak{R}}$ must
be included in the Killing superalgebra, just as we would expect from
Theorem~\ref{thm:spencer-R}.

\section{Maximally supersymmetric backgrounds}
\label{sec:maxim-supersymm-back}

We now investigate which geometries admit the ``maximal number of Killing
spinors''; that is, for which the dimension of the space of solutions to
equation~\eqref{eq:KillingSpinor} is maximal.  This is equivalent to
demanding that the curvature of the connection
\begin{equation}
  \label{eq:superconnection}
  \eD_X = \nabla_X - \iota_X H + X \cdot \varphi + 2 \varphi \cdot X
\end{equation}
vanishes.  As usual $\varphi$ is a one-form with values in the Lie
algebra of the R-symmetry and $H$ a $3$-form.  We will not take
$H$ to be self-dual at this moment, but will comment on any
simplifications which result from that assumption.

\subsection{The curvature of the superconnection}
\label{sec:curv-superc}

Let us write the connection as $\eD = \nabla - \beta$, where
\begin{equation}
  (\beta_\mu \varepsilon)^A = \tfrac12
  H_{\mu\rho\sigma}\Gamma^{\rho\sigma} \varepsilon^A - 3 \varphi_\mu{}^A{}_B
  \varepsilon^B + \varphi^{\sigma\, A}{}_B \Gamma_{\mu\sigma} \varepsilon^B.
\end{equation}
The curvature $\eR$ of $\eD$ is given by
\begin{equation}
  \eR_{\mu\nu}{}^A{}_B = \tfrac14 R_{\mu\nu\alpha\beta}
  \Gamma^{\alpha\beta} \delta^A{}_B + \nabla_\mu \beta_\nu{}^A{}_B -
  \nabla_\nu\beta_\mu{}^A{}_B - [\beta_\mu,\beta_\nu]^A{}_B,
\end{equation}
which can be decomposed into components
\begin{equation}
  \label{eq:curvature}
  \eR_{\mu\nu}{}^A{}_B = \tfrac12 \eT_{\mu\nu\alpha\beta}
  \Gamma^{\alpha\beta} \delta^A{}_B + \eU_{\mu\nu}{}^A{}_B + \tfrac12
  \eV_{\mu\nu\alpha\beta}{}^A{}_B \Gamma^{\alpha\beta}.
\end{equation}
The tensor $\eT$ is skew-symmetric in the first and last pairs of
indices, so it is a section through a bundle associated to the
representation $\Lambda^2V \otimes \Lambda^2V$ of $\fso(V)$. Similarly
$\eU$ (resp. $\eV$) is a section through a bundle associated to the
representation $\Lambda^2 V\otimes\fsp(1)$ (resp.
$\Lambda^2 V\otimes\Lambda^2 V\otimes\fsp(1)$) of
$\fso(V)\oplus\fsp(1)$.

The explicit expressions of the components are
\begin{multline}
  \label{eq:T-component}
  \eT_{\mu\nu\alpha\beta} = \tfrac12 R_{\mu\nu\alpha\beta} + \nabla_\mu
  H_{\nu\alpha\beta} - \nabla_\nu H_{\mu\alpha\beta} + 2 \left(
    \left<H_{\mu\alpha}, H_{\nu\beta}\right> -
    \left<H_{\mu\beta},H_{\nu\alpha}\right>\right)\\
  + \left( g_{\mu\alpha}
    \left<\varphi_\beta,\varphi_\nu\right> - g_{\mu\beta}
    \left<\varphi_\alpha,\varphi_\nu\right> - g_{\nu\alpha}
    \left<\varphi_\beta,\varphi_\mu\right> + g_{\nu\beta}
    \left<\varphi_\alpha,\varphi_\mu\right>\right)\\
  - \left<\varphi^\lambda,\varphi_\lambda\right> (g_{\mu\alpha} g_{\nu\beta} -
  g_{\mu\beta} g_{\nu\alpha}),
\end{multline}
with $\left<H_{\mu\nu},H_{\alpha\beta}\right> := H_{\mu\nu}{}^\lambda
H_{\alpha\beta\lambda}$ and $\left<\varphi_\mu,\varphi_\nu\right> = \varphi_\mu{}^{AB}
\varphi_{\nu\, AB}$, and, omitting the R-symmetry indices, 
\begin{equation}
  \label{eq:U-component}
  \eU_{\mu\nu} = - 3 \left(\nabla_\mu \varphi_\nu - \nabla_\nu
    \varphi_\mu\right) - 8 [\varphi_\mu, \varphi_\nu],
\end{equation}
and
\begin{multline}
  \label{eq:V-component}
  \eV_{\mu\nu\alpha\beta} = \nabla_\mu \varphi_\beta g_{\nu\alpha} -
  \nabla_\mu \varphi_\alpha g_{\nu\beta} - \nabla_\nu \varphi_\beta g_{\mu\alpha} +
  \nabla_\nu \varphi_\alpha g_{\mu\beta} + 4 \left( H_{\mu\nu\alpha} \varphi_\beta -
    H_{\mu\nu\beta} \varphi_\alpha \right)\\
  - 2 \left( H_{\mu\beta\lambda} g_{\nu\alpha} - H_{\mu\alpha\lambda} g_{\nu\beta}
    - H_{\nu\beta\lambda} g_{\mu\alpha} + H_{\nu\alpha\lambda} g_{\mu\beta} \right) \varphi^\lambda -
  \epsilon_{\mu\nu\alpha\beta\rho\sigma} [\varphi^\rho,\varphi^\sigma]\\
  + 3\left( [\varphi_\mu,\varphi_\beta] g_{\nu\alpha} - [\varphi_\mu,\varphi_\alpha]
    g_{\nu\beta} -  [\varphi_\nu,\varphi_\beta] g_{\mu\alpha} + [\varphi_\nu,\varphi_\alpha]
    g_{\mu\beta} \right).
\end{multline}

\subsection{Zero curvature conditions}
\label{sec:zero-curv-cond}

We begin our analysis of the zero curvature conditions.  Let us start
with the condition $\eT_{\mu\nu\alpha\beta} = 0$.  The tensor $\eT$ is a section
through a bundle associated to $\Lambda^2V \otimes
\Lambda^2V$ and this representation decomposes as follows
into irreducible components under $\fso(V)$:
\begin{equation}
  \begin{split}
    \Lambda^2 V \otimes \Lambda^2 V &= \odot^2(\Lambda^2 V) \oplus
    \Lambda^2(\Lambda^2 V)\\
    &= \left( \Lambda^4 V \oplus \Lambda^0V \oplus \odot^2_0 V \oplus W
    \right) \oplus \left( \Lambda^2 V \oplus (V \otimes \Lambda^3_+V)_0
      \oplus (V \otimes \Lambda^3_- V)_0\right),
  \end{split}
\end{equation}
where $(V \otimes \Lambda^3_\pm V)_0$ is the kernel of the natural contraction $V
\otimes \Lambda^3_\pm V \to \Lambda^2 V$ and $W$ is the module of Weyl curvature tensors.

Setting the $\Lambda^2(\Lambda^2 V)$ component to zero, we find
\begin{equation}
  \begin{split}
    0 &= \eT_{\mu\nu\alpha\beta} - \eT_{\alpha\beta\mu\nu}\\
    &= \nabla_\mu H_{\nu\alpha\beta} - \nabla_\nu H_{\mu\alpha\beta} -
    \nabla_\alpha H_{\beta\mu\nu} + \nabla_\beta H_{\alpha\mu\nu},
  \end{split}
\end{equation}
which we can rewrite as
\begin{equation}
  \label{eq:1}
  \nabla_{[\mu} H_{\nu]\alpha\beta} =   \nabla_{[\alpha} H_{\beta]\mu\nu}.
\end{equation}

Skew-symmetrising in the last three indices we find
\begin{equation}
  \label{eq:T-component-skew-3}
  \begin{split}
    0 &= \eT_{\mu[\nu\alpha\beta]}\\
    &= \nabla_\mu H_{\nu\alpha\beta} - \tfrac13 \left( \nabla_\nu
      H_{\alpha\beta\mu}+ \nabla_\alpha
      H_{\beta\nu\mu} + \nabla_\beta
      H_{\nu\alpha\mu}\right) + \tfrac43
    \left(\left<H_{\mu\alpha},H_{\nu\beta}\right> -
      \left<H_{\mu\beta},H_{\nu\alpha}\right> -
      \left<H_{\mu\nu},H_{\alpha\beta}\right> \right),
  \end{split}
\end{equation}
whereas completely skew-symmetrising gives
\begin{equation}
  \label{eq:T-component-skew-all}
  \begin{split}
    0 &= \eT_{[\mu\nu\alpha\beta]}\\
    &= \nabla_{[\mu} H_{\nu]\alpha\beta} + \nabla_{[\alpha}
    H_{\beta]\mu\nu} + \tfrac43 \left(\left<H_{\mu\alpha},H_{\nu\beta}\right> -
      \left<H_{\mu\beta},H_{\nu\alpha}\right> -
      \left<H_{\mu\nu},H_{\alpha\beta}\right> \right).
  \end{split}
\end{equation}
Comparing the two equations  \eqref{eq:T-component-skew-3} and   \eqref{eq:T-component-skew-all} we see that
\begin{equation}
  \label{eq:eq3}
  \nabla_{[\mu} H_{\nu]\alpha\beta} + \nabla_{[\alpha} H_{\beta]\mu\nu}
  = 2 \nabla_\mu H_{\nu\alpha\beta},
\end{equation}
which together with equation~\eqref{eq:1} gives
\begin{equation}
  \label{eq:eq4}
  \nabla_\mu H_{\nu\alpha\beta} + \nabla_\nu H_{\mu\alpha\beta} = 0,
\end{equation}
which says that the covariant derivative of $H$ is a $4$-form. In other words $\nabla H=\tfrac14 dH$ and $H$ is a
coclosed conformal Killing $3$-form, or a Killing $3$-form in the
nomenclature of \cite{SemmelmanCKForms}.

The totally skew-symmetric component then finally gives
\begin{equation}
  \label{eq:eq5}
  \nabla_\mu H_{\nu\alpha\beta} = \tfrac23 \left(
    \left<H_{\mu\nu},H_{\alpha\beta}\right>  -
    \left<H_{\mu\alpha},H_{\nu\beta}\right> +
    \left<H_{\mu\beta},H_{\nu\alpha}\right> \right).
\end{equation}

The algebraic curvature tensor components ($\Lambda^0V \oplus \odot^2_0
V \oplus W$) of $\eT$ give the Riemann tensor in terms of $H$ and $\varphi$.

The vanishing of the $\eU$-component of the curvature gives the equation
\begin{equation}
  \label{eq:U-comp-zero}
  0 = \tfrac12\eU_{\mu\nu} = -3 \nabla_{[\mu} \varphi_{\nu]} - 4 [\varphi_\mu, \varphi_\nu].
\end{equation}

Next let us consider the vanishing of the $\eV$-component of the
curvature.  Totally skew-symmetrising and omitting the R-symmetry
indices, we obtain
\begin{equation}
  \label{eq:V-skew}
    0 = \eV_{[\mu\nu\alpha\beta]} = 8 H_{[\mu\nu\alpha} \varphi_{\beta]}
    - \epsilon_{\mu\nu\alpha\beta\rho\sigma}
    [\varphi^\rho,\varphi^\sigma],
\end{equation}
which can be rewritten as
\begin{equation}
  \label{eq:V-skew-too}
  [\varphi_\mu, \varphi_\nu] = \widetilde H_{\mu\nu\lambda} \varphi^\lambda,
\end{equation}
with $\widetilde H$ the Hodge dual of $H$.

The component along $\Lambda^2(\Lambda^2 V)$ gives
\begin{equation}
  \label{eq:l2l2v}
  \begin{split}
    0 &= \tfrac12 (\eV_{\mu\nu\alpha\beta} - \eV_{\alpha\beta\mu\nu})\\
    &= \nabla_{[\mu} \varphi_{\beta]} g_{\nu\alpha} - \nabla_{[\mu}
    \varphi_{\alpha]} g_{\nu\beta} - \nabla_{[\nu} \varphi_{\beta]} g_{\mu\alpha} +
    \nabla_{[\nu} \varphi_{\alpha]} g_{\mu\beta}\\
    & \quad {} + 2 \left( H_{\mu\nu\alpha} \varphi_\beta - H_{\mu\nu\beta} \varphi_\alpha
      - H_{\alpha\beta\mu} \varphi_\nu + H_{\alpha\beta\nu} \varphi_\mu\right)\\
    & \quad {} - 2 \left( H_{\mu\beta\lambda} g_{\nu\alpha} - H_{\mu\alpha\lambda} g_{\nu\beta}
    - H_{\nu\beta\lambda} g_{\mu\alpha} + H_{\nu\alpha\lambda} g_{\mu\beta} \right) \varphi^\lambda\\
  & \quad {} + 3\left( [\varphi_\mu,\varphi_\beta] g_{\nu\alpha} - [\varphi_\mu,\varphi_\alpha]
    g_{\nu\beta} -  [\varphi_\nu,\varphi_\beta] g_{\mu\alpha} + [\varphi_\nu,\varphi_\alpha]
    g_{\mu\beta} \right).
  \end{split}
\end{equation}
Using equations~\eqref{eq:U-comp-zero} and \eqref{eq:V-skew-too}, we may rewrite this equation as
\begin{equation}
  \label{eq:l2l2v-also}
  \begin{split}
    0 &= 2 \left( H_{\mu\nu\alpha} \varphi_\beta - H_{\mu\nu\beta} \varphi_\alpha
      - H_{\alpha\beta\mu} \varphi_\nu + H_{\alpha\beta\nu} \varphi_\mu\right)\\
    & \quad {} - 2 \left( H_{\mu\beta\lambda} g_{\nu\alpha} - H_{\mu\alpha\lambda} g_{\nu\beta}
    - H_{\nu\beta\lambda} g_{\mu\alpha} + H_{\nu\alpha\lambda} g_{\mu\beta}\right) \varphi^\lambda\\
  & \quad {} + \tfrac53 \left(\widetilde H_{\mu\beta\lambda} g_{\nu\alpha} - \widetilde H_{\mu\alpha\lambda} g_{\nu\beta}
    - \widetilde H_{\nu\beta\lambda} g_{\mu\alpha} + \widetilde
    H_{\nu\alpha\lambda} g_{\mu\beta} \right) \varphi^\lambda.
  \end{split}
\end{equation}
Contracting with $g^{\nu\alpha}$, we find
\begin{equation}
  \label{eq:trace}
  \tfrac53 \widetilde H_{\mu\beta\lambda} \varphi^\lambda = H_{\mu\beta\lambda} \varphi^\lambda,
\end{equation}
and reinserting this into equation~\eqref{eq:l2l2v-also}, we arrive at
\begin{equation}
  \label{eq:l2l2v-too}
    2 \left( H_{\mu\nu\alpha} \varphi_\beta - H_{\mu\nu\beta} \varphi_\alpha -
      H_{\alpha\beta\mu} \varphi_\nu + H_{\alpha\beta\nu} \varphi_\mu\right) =
    \left( H_{\mu\beta\lambda} g_{\nu\alpha} - H_{\mu\alpha\lambda}
      g_{\nu\beta} - H_{\nu\beta\lambda} g_{\mu\alpha} +
      H_{\nu\alpha\lambda} g_{\mu\beta}\right) \varphi^\lambda.
\end{equation}

\begin{lemma}
  \label{lem:HA}
  Equation~\eqref{eq:l2l2v-too} is equivalent to $H_{\mu\nu\alpha}
  \varphi_\beta = 0$.
\end{lemma}

\begin{proof}
  Let $p \in M$ and suppose that $\varphi|(p)\neq 0$, so that some component
  $\varphi^{AB}$ is different from zero at $p$.  We will let $a =\varphi^{AB}(p)$ and
  show that $H(p) = 0$.

  Equation~\eqref{eq:l2l2v-too} for the component $a$ becomes
  \begin{equation}
    \label{eq:HAHA}
    2 \left( H_{\mu\nu\alpha} a_\beta - H_{\mu\nu\beta} a_\alpha -
      H_{\alpha\beta\mu} a_\nu + H_{\alpha\beta\nu} a_\mu\right) =
    \left( H_{\mu\beta\lambda} g_{\nu\alpha} - H_{\mu\alpha\lambda}
      g_{\nu\beta} - H_{\nu\beta\lambda} g_{\mu\alpha} +
      H_{\nu\alpha\lambda} g_{\mu\beta}\right) a^\lambda.
  \end{equation}
  Skew-symmetrising in $[\mu\nu\alpha]$, we find
  \begin{equation}
    \label{eq:HAHA2}
    H_{\mu\nu\alpha} a_\beta + H_{\beta[\mu\nu} a_{\alpha]} =
    g_{\beta[\mu} H_{\nu\alpha]\lambda} a^\lambda.
  \end{equation}
  Contracting equation~\eqref{eq:HAHA} with $a^\beta$, we find that
  \begin{equation}
    \label{eq:HAA}
    H_{\mu\nu\alpha} a^2 = 0,
  \end{equation}
  whereas contracting equation~\eqref{eq:HAHA2} with $a^\mu$ and using
  equation~\eqref{eq:HAA}, we arrive at
  \begin{equation}
    \label{eq:AHA}
    a_\beta H_{\nu\alpha\lambda} a^\lambda = a_{[\nu}
    H_{\alpha]\beta\lambda} a^\lambda.
  \end{equation}
  Let's multiply this equation by $a_\mu$ and skew-symmetrise in
  $[\mu\nu\alpha]$ to obtain
  \begin{equation}
    a_\beta a_{[\mu} H_{\nu\alpha]\lambda}a^\lambda = 0,
  \end{equation}
  which, since $a\neq 0$, is equivalent to
  \begin{equation}
    \label{eq:AHA2}
    a_{[\mu} H_{\nu\alpha]\lambda}a^\lambda = 0.
  \end{equation}
  Adding equations~\eqref{eq:AHA} and \eqref{eq:AHA2}, we arrive at
  \begin{equation}
    \label{eq:AcontractH}
    H_{\nu\alpha\lambda} a^\lambda = 0.
  \end{equation}
  Inserting this into equation~\eqref{eq:HAHA2}, we get
  \begin{equation}
    \label{eq:HAplusHA}
    H_{\mu\nu\alpha} a_\beta + H_{\beta[\mu\nu} a_{\alpha]} = 0,
  \end{equation}
  and into equation~\eqref{eq:HAHA},
  \begin{equation}
    \label{eq:HAplusHAplusHAplusHA}
    H_{\mu\nu\alpha} a_\beta - H_{\mu\nu\beta} a_\alpha -
    H_{\alpha\beta\mu} a_\nu + H_{\alpha\beta\nu} a_\mu = 0.
  \end{equation}
  Subtracting the two equations, we find
  \begin{equation}
    \label{eq:HAsym}
    H_{\mu\nu(\alpha} a_{\beta)} = 0.
  \end{equation}
  Now let's contract this equation with $b^\beta$, where $b$ is a vector such that
  $b^\mu a_\mu = 1$, to find
  \begin{equation}
    H_{\mu\nu\alpha} + H_{\mu\nu\beta} a_\alpha b^\beta = 0.
  \end{equation}
  Contracting this equation with $b^\alpha$ now, we find
  $H_{\mu\nu\alpha}b^\alpha = 0$, which inserted in the previous
  equation, gives $H_{\mu\nu\alpha} = 0$, as desired.
\end{proof}

We have proved most of the following.

\begin{proposition}\label{prop:branchespreliminary}
  Let $(M,g,H,\varphi)$ be a (connected) lorentzian six-dimensional spin
  manifold endowed with a $3$-form $H$ and a $1$-form $\varphi$ with
  values in $\fsp(1)$. Assume $(M,g,H,\varphi)$ is maximally
  supersymmetric, that is, the dimension of the space of solutions to
  equation~\eqref{eq:KillingSpinor} is maximal. Then either (1) $H=0$ or
  (2) $\varphi=0$. Moreover:
  \begin{enumerate}
  \item In the first case $\varphi=a\otimes R$ is decomposable, where
    $a\in\Omega^1(M)$ is a parallel $1$-form on $M$ and $R\in\fsp(1)$ a
    fixed element of the R-symmetry algebra;
  \item In the second case $H$ is a coclosed conformal Killing $3$-form
    (i.e., $\nabla H=\tfrac14 dH$) with covariant derivative given by
    \begin{equation}
      \nabla_\mu H_{\nu\alpha\beta} = \tfrac23 \left(
        \left<H_{\mu\nu},H_{\alpha\beta}\right>  -
        \left<H_{\mu\alpha},H_{\nu\beta}\right> +
        \left<H_{\mu\beta},H_{\nu\alpha}\right> \right).
    \end{equation}
  \end{enumerate}
  In both cases, the algebraic curvature tensor components ($\Lambda^0V
  \oplus \odot^2_0 V \oplus W$) of the tensor \eqref{eq:T-component}
  give the Riemann curvature tensor in terms of $\varphi$ and,
  respectively, $H$.
\end{proposition}

\begin{proof}
  Lemma~\ref{lem:HA} shows that any point $p \in M$, either $\varphi=0$ or $H=0$ (or
  possibly both).  In particular, this means that at all points,
  $H_{\mu\nu\lambda}\varphi^\lambda = 0$, so that also $\widetilde
  H_{\mu\nu\lambda}\varphi^\lambda = 0$, $[\varphi_\mu, \varphi_\nu] = 0$ and $\nabla_{[\mu}
  \varphi_{\nu]} = 0$.

  The component of $\eV$ along $\odot^2(\Lambda^2V)$ gives, after using
  equation~\eqref{eq:V-skew},
  \begin{equation}
    \label{eq:s2l2v}
    \nabla_{(\mu} \varphi_{\beta)} g_{\nu\alpha} - \nabla_{(\mu}
    \varphi_{\alpha)} g_{\nu\beta} - \nabla_{(\nu} \varphi_{\beta)} g_{\mu\alpha} +
    \nabla_{(\nu} \varphi_{\alpha)} g_{\mu\beta} = 0,
  \end{equation}
  which upon contraction with $g^{\nu\alpha}$ gives
  \begin{equation}
    4 \nabla_{(\mu} \varphi_{\beta)} + \nabla^\lambda \varphi_\lambda g_{\mu\beta} = 0.
  \end{equation}
  Contracting with $g^{\mu\beta}$ we find $\nabla^\lambda \varphi_\lambda = 0$
  and inserting back into the previous equation,
  \begin{equation}
    \nabla_{(\mu} \varphi_{\beta)} = 0.
  \end{equation}
  Together with $\nabla_{[\mu} \varphi_{\nu]} = 0$, we conclude that $\varphi$ is
  parallel.  Parallel sections of vector bundles are determined by their
  value at any given point, hence if $\varphi=0$ at any point, it is
  identically zero.  In other words, either $\varphi$ is identically zero or, by
  Lemma \ref{lem:HA}, $H$ is identically zero.  (Of course, it is possible that both
  are identically zero, which corresponds to the trivial (flat)
  background.)

  In the first case $\varphi=a\otimes R$ for some parallel
  $a\in\Omega^1(M)$ and constant $R\in\fsp(1)$, since
  $\varphi:TM\to\fsp(1)$ has a $1$-dimensional range at any point $p\in
  M$ (due to $[\varphi_\mu, \varphi_\nu] = 0$) and it is parallel. The
  rest is clear.
\end{proof}
In summary, we have two branches of nontrivial backgrounds, which we
will analyse in turn.

\subsection{First branch: $H = 0$}
\label{sec:branch-Hzero}

Here $\varphi\neq 0$ and $\varphi_\mu{}^A{}_B = a_\mu R^A{}_B$, where
$R$ is a fixed element of the R-symmetry algebra $\fsp(1)$ and the
one-form $a$ is parallel. Without loss of generality we can normalise
$R$ so that $R^{AB}R_{AB} = 1$ or, equivalently, that $\tr(R^2) = -1$.
The fact that the $1$-form $a$ is parallel can be seen also from the
vanishing of the $\eT$-component~\eqref{eq:T-component} of the curvature
of the spinor connection, which in this branch becomes
\begin{equation}
  \label{eq:riemann-branch-1}
  R_{\mu\nu\alpha\beta} = 2 a^2 (g_{\mu\alpha} g_{\nu\beta} -
  g_{\mu\beta} g_{\nu\alpha}) - 2 (g_{\mu\alpha} a_\beta a_\nu -
  g_{\mu\beta} a_\alpha a_\nu - g_{\nu\alpha} a_\beta a_\mu +
  g_{\nu\beta} a_\alpha a_\mu),
\end{equation}
and noticing that $R_{\mu\nu\alpha\beta} a^\beta = 0$.  

The causal type of a parallel vector is constant, so we may distinguish
between three cases depending on whether $a$ is null, spacelike or
timelike. The discussion breaks up naturally into two cases, depending
on whether or not the squared norm $a^2$ of $a$ vanishes. In all cases,
it follows from the expression~\eqref{eq:riemann-branch-1} of the
Riemann tensor that the Weyl tensor vanishes and hence that all
geometries are conformally flat.

\subsubsection{$a^2\neq 0$}
\label{sec:a-nonnull}

Since $a$ is parallel, nowhere-vanishing and $a^2 \neq 0$, the de~Rham
decomposition theorem says that $(M,g)$ is locally isometric to a
product: $M = N \times \RR$, where $\RR$ is either timelike or spacelike
according to the causal type of $a$. To understand the geometry of $N$,
we rewrite the Riemann tensor in equation~\eqref{eq:riemann-branch-1} as
follows:
\begin{equation}
  \label{eq:riemann-branch-1-time}
  R_{\mu\nu\alpha\beta} =
 % 2 a^2 \left(\left(g_{\mu\alpha} -
 %      \frac{a_\alpha a_\mu}{a^2}\right) \left(g_{\nu\beta} -
 %      \frac{a_\beta a_\nu}{a^2}\right) - \left(g_{\mu\beta} -
 %      \frac{a_\beta a_\mu}{a^2}\right) \left(g_{\nu\alpha} -
 %      \frac{a_\alpha a_\nu}{a^2}\right) \right).
  2 a^2 \left( h_{\mu\alpha} h_{\nu\beta} - h_{\mu\beta} h_{\nu\alpha} \right),
\end{equation}
where the tensor $h_{\mu\nu} := g_{\mu\nu} - \frac{a_\mu a_\nu}{a^2}$.
Note that $h$ coincides with the induced metric on the distribution
$a^\perp$ perpendicular to $a$, which is the tangent bundle of $N$.  The
above form of the Riemann tensor makes it evident that $N$ has constant
sectional curvature, with Ricci tensor
\begin{equation}
  R_{\mu\beta} := R_\mu{}^\nu{}_{\nu\beta} = -8 a^2 h_{\mu\beta}.
\end{equation}
Therefore if $a$ is timelike, so that $a^2<0$, $(M,g)$ is locally
isometric to $(\RR,-dt^2) \times \Sph^5$, where $\Sph^5$ is a round
$5$-sphere with scalar curvature $-40 a^2$; whereas if $a$ is spacelike,
$(M,g)$ is locally isometric to $(\RR,dt^2) \times \AdS_5$, with
$\AdS_5$ the anti-de~Sitter spacetime with scalar curvature $-40 a^2$.

\subsubsection{$a$ is null}
\label{sec:a-null}

If $a$ is non-zero and null, $(M,g)$ is a Brinkmann space, with  Riemann curvature tensor given by
\begin{equation}
  \label{eq:riemann-branch-1-null}
  R_{\mu\nu\alpha\beta} = -2(g_{\mu\alpha} a_\beta a_\nu -
  g_{\mu\beta} a_\alpha a_\nu - g_{\nu\alpha} a_\beta a_\mu +
  g_{\nu\beta} a_\alpha a_\mu).
\end{equation}
It is clear by inspection of the above expression for the Riemann
curvature tensor, that the metric is both conformally flat and scalar
flat.  Furthermore, $R(a^\perp,a^\perp)=0$, so that the transverse
geometry is flat and since $a$ and $g$ are both parallel, so is the Riemann
tensor. Hence $(M,g)$ is locally isometric to a (possibly
decomposable) Cahen--Wallach plane wave \cite{CahenWallach}, with
metric 
\begin{equation}
  g = 2dx^+dx^- + \sum_{i,j=1}^4 B_{ij} x^i x^j (dx^-)^2 +
  \sum_{i=1}^4 (dx^i)^2,
\end{equation}
where the parallel null vector is $a = \partial_+$.  The only
nonzero components of the Weyl tensor of the metric $g$ are
\begin{equation}
  W_{-ij-} = B_{ij} - \tfrac14 (\tr B) \delta_{ij},
\end{equation}
so that $g$ is conformally flat if and only if $B$ is a scalar matrix.
From the explicit form  of the Riemann curvature tensor
\eqref{eq:riemann-branch-1-null} we see that $B$ is nonzero and up to a
local diffeomorphism we may write the metric down as
\begin{equation}
  \label{eq:ppwave}
  g_{\pm} = 2 dx^+dx^- \pm \tfrac14 \sum_{i=1}^4(x^i)^2(dx^-)^2 +
  \sum_{i=1}^4 (dx^i)^2.
\end{equation}
Comparing the Riemann tensor of $g_\pm$ with
\eqref{eq:riemann-branch-1-null} precisely selects the metric $g_-$ (see
also equation (13) in \cite{CFOSchiral}). The metric in this background
is (locally) isometric to the plane wave in \cite{Meessen}; but the
Killing spinors here obey a different equation than those of $d=6$
$(1,0)$ supergravity. In other words, we are in the curious situation
where the same geometry is maximally supersymmetric with respect to two
different notions of Killing spinors.

\subsection{Second branch: $\varphi=0$}
\label{sec:branch-Azero}

In the second branch, $\varphi=0$ and the Killing spinors satisfy
\begin{equation}
  \label{eq:KSbranch2}
  \eD_X \varepsilon = \nabla_X \varepsilon - \iota_X H \cdot \varepsilon
  = 0.
\end{equation}
We may understand this equation as saying that Killing spinors are
parallel with respect to (the spin lift of) the metric connection 
\begin{equation}
  D_X Y = \nabla_X Y + 2 h(X,Y)
\end{equation} 
with skew-symmetric torsion $g(h(X,Y),Z)=H(X,Y,Z)$. Since the
representation of $\fso(V)$ on $\Sigma_+$ is faithful, maximal
supersymmetry exactly amounts to the flatness of $D$ and, using
Proposition \ref{prop:branchespreliminary}, this condition is equivalent
to
\begin{equation}
\label{eq:flatnessD}
\begin{split}
  \nabla_\mu H_{\nu\alpha\beta} &= \tfrac23 \left(
    \left<H_{\mu\nu},H_{\alpha\beta}\right>  -
    \left<H_{\mu\alpha},H_{\nu\beta}\right> +
    \left<H_{\mu\beta},H_{\nu\alpha}\right> \right)\\
  R_{\mu\nu\alpha\beta} &=-2\nabla_\mu
  H_{\nu\alpha\beta} +2 \nabla_\nu H_{\mu\alpha\beta} -4 \left(
    \left<H_{\mu\alpha}, H_{\nu\beta}\right> -
    \left<H_{\mu\beta},H_{\nu\alpha}\right>\right)\\
    &=-\tfrac83\left<H_{\mu\nu},H_{\alpha\beta}\right>
-\tfrac43\left<H_{\mu\alpha},H_{\nu\beta}\right>
+\tfrac43\left<H_{\mu\beta},H_{\nu\alpha}\right>.
\end{split}
\end{equation}
If the torsion $3$-form $H$ is, in addition, parallel with respect to
$\nabla$ (equivalently it is closed) then $H$ obeys the Jacobi identity.
More precisely, it is possible to use \eqref{eq:flatnessD} to see that
in this case $(M,g)$ is locally isometric to a Lie group with a
bi-invariant lorentzian metric, $H$ is the Cartan $3$-form of the group
and $D$ the parallelising connection (see \cite[\S 2.3]{MR2436237}).
Furthermore $H$ is also $D$-parallel. We recall that the general
classification of lorentzian Lie groups is due to Medina
\cite{MedinaLorentzian}.

Now, it is a deep result of Cahen and Parker that a simply connected,
complete indecomposable \emph{lorentzian} manifold $(M,g)$ with a flat
metric connection with skew-torsion $H$ satisfies $\nabla H=0$
automatically \cite{MR0461388}. They also show that the assumption of
indecomposability can be relaxed.

In summary, a maximally supersymmetric background $(M,g,H)$ in the
second branch of our classification is (up to local isometry) a Lie
group with a bi-invariant lorentzian metric. The Lie algebra of such a
Lie group is a six-dimensional Lie algebra with a lorentzian
ad-invariant inner product and these have been listed in
\cite{CFOSchiral, MR2436237}. The corresponding backgrounds are:
\begin{enumerate}
\item $\RR^{5,1}$,
\item $\RR^{2,1}\times \Sph^3$,
\item $\RR^3\times \AdS_3$,
\item $\AdS_3\times \Sph^3$,
\item the plane wave \eqref{eq:ppwave} in \cite{Meessen}.
\end{enumerate}
We emphasise that the Cartan $3$-form $H$ may be chosen self-dual or
antiself-dual in cases (1), (4) and (5) above. Solutions with self-dual
Cartan $3$-form correspond to the maximally supersymmetric backgrounds
of $d=6$ $(1,0)$ supergravity.

We have proved the following classification result. We recall that in
our conventions $S$ is an irreducible representation of $\Spin(V)$ of
quaternionic dimension $2$.

\begin{theorem}
  \label{thm:maxsusy}
  Let $(M,g)$ be a lorentzian six-dimensional spin manifold, with
  associated spinor bundle $\SS\to M$ with typical fiber $S$. Let
  $H\in\Omega^3(M)$ be a $3$-form and $\varphi$ a $1$-form on $M$ with
  values in $\fsp(1)$. Let also
  \begin{equation}
    \eD_X \varepsilon:= \nabla_X \varepsilon - \iota_X H \cdot
    \varepsilon + 3\varphi(X)\cdot\varepsilon-X\wedge
    \varphi\cdot\varepsilon\;,
  \end{equation}
  be the associated spinor connection.
  If $\eD$ is flat, then $(M,g,H,\varphi)$ is locally isometric to one of the following:
  \begin{enumerate}[label=(\roman*)]
  \item $\varphi = H = 0$ and $(M,g)$ Minkowski spacetime;
  \item $H=0$ and $\varphi = a\otimes R$ for some parallel $1$-form $a$
    and fixed element $R $ of the R-symmetry algebra $\fsp(1)$ with
    $\tr(R^2) = -1$. Depending on the causal type of $a$:
    \begin{enumerate}[label=(\alph*)]
    \item $(M,g)$ the product $(\RR,-dt^2) \times \Sph^5$, with $\Sph^5$
      the round $5$-sphere of scalar curvature $-40 a^2$, for $a^2<0$;
    \item $(M,g)$ the product $(\RR,dt^2) \times \AdS_5$, with $\AdS_5$
      the five-dimensional anti-de~Sitter spacetime with scalar
      curvature $-40 a^2$, for $a^2>0$;
    \item and if $a^2 =0$, then $(M,g)$ a conformally flat lorentzian
      symmetric plane wave with metric
      \begin{equation}
        g = 2dx^+dx^- - \tfrac14 \sum_{i=1}^4(x^i)^2(dx^-)^2 +
        \sum_{i=1}^4 (dx^i)^2;
      \end{equation}
    \end{enumerate}
  \item  $\varphi=0$ and $H$ is the parallel Cartan $3$-form of a
    six-dimensional Lie group with bi-invariant lorentzian metric. If
    $H$ is in addition self-dual, these are the maximally supersymmetric
    backgrounds of $d=6$ $(1,0)$ supergravity.
  \end{enumerate}
\end{theorem}

\subsection{Killing superalgebras and filtered deformations}
\label{sec:KSAfd}

A natural question is whether the Killing spinors in the maximally
supersymmetric backgrounds of Theorem~\ref{thm:maxsusy} generate always
a Lie superalgebra and if different backgrounds have different
associated Killing superalgebras (in particular, if the plane wave that
is maximally supersymmetric in two different senses has different
Killing superalgebras for the different types of Killing spinors).

Concerning existence, there is clearly nothing to check for the trivial
Minkowski background. In the first branch, where $H=0$ and
$\varphi=a\otimes R$ is parallel (in particular coclosed), this follows
from Theorem \ref{thm:KSAI}. In the second branch, the existence is
guaranteed from the fact that $H=H^++H^-$ is closed and $D$-parallel, so
that Theorem \ref{thm:KSAII} applies. When $H$ is self-dual, Theorem
\ref{thm:KSAI} is actually sufficient, giving rise to an ideal
$\fk=\fk_{\bar 0}\oplus\fk_{\bar 1}$ of the general extended Killing
superalgebra $\fkhat=\fkhat_{\bar 0}\oplus\fkhat_{\bar 1}$.

The fact that maximally supersymmetric backgrounds are distinguished by
their Killing superalgebras is a consequence of the general theory of
filtered deformations (of subalgebras) of the Poincaré superalgebra,
possibly extended by R-symmetries. The full line of arguments is as for
the eleven-dimensional case \cite{Figueroa-OFarrill:2015rfh,
  Figueroa-O'Farrill:2015utu, Figueroa-OFarrill:2016khp} and
four-dimensional case \cite{deMedeiros:2016srz}, and we will not
replicate it here. We will review the main ingredients in the simpler
case of "maximally supersymmetric" filtered deformations, which is
enough for the purpose of this paper, and comment on any adjustment
which results from the presence of R-symmetries.

Let 
\begin{equation}
\label{eq:gradedA}
\fa=\fa_{0}\oplus\fa_{-1}\oplus\fa_{-2}
\end{equation}
be a $\ZZ$-graded Lie subalgebra of the $d{=}6$ $(1,0)$ Poincaré superalgebra 
\begin{equation}
\begin{split}
  \fp &= \fp_{-2} \oplus \fp_{-1} \oplus \fp_0\\
 &= V \oplus S \oplus \fso(V)\;,
  \end{split}
\end{equation}
which satisfies $\fa_{-1}=S$ and $\fa_{-2}=V$. The fact that
$\fa_{-1}=S$ means we have maximal supersymmetry, whereas $\fa_{-2}=V$
(which is forced by the local homogeneity theorem in
\cite{Figueroa-O'Farrill:2013aca}) means we are
describing (locally) homogeneous geometries. 

\begin{definition}
  A \emph{filtered deformation} of $\fa$ is a Lie superalgebra $\fg$
  supported on the same underlying vector space of $\fa$ whose Lie
  brackets have nonnegative total degree, with the zero-degree
  components coinciding with the Lie brackets of $\fa$.
\end{definition}

If we do not wish to mention the subalgebra $\fa$ of $\fp$ explicitly,
we simply say that $\fg$ is a (maximally supersymmetric) filtered
subdeformation of $\fp$. The notion of a filtered subdeformation $\fg$
of the extended Poincaré superalgebra $\fphat$ can be introduced in a
completely analogous way. We note that the spin group $\Spin(V)$
naturally acts on $\fp$ by $0$-degree Lie superalgebra automorphisms, so
that any element $g\in\Spin(V)$ sends a graded subalgebra of $\fp$ into
an (isomorphic) graded subalgebra of $\fp$ and filtered subdeformations
into filtered subdeformations. A similar observation holds for the
action of $\Spin(V)\times\Sp(1)$ on $\fphat$.

The $\ZZ_2$-grading of $\fa$ is compatible with the $\ZZ$-grading, in
that $\fa_{\overline 0}=\fa_0\oplus\fa_{-2}$ and
$\fa_{\overline 1}=\fa_{-1}$. In particular the components of the Lie
brackets of a filtered subdeformation $\fg$ of $\fp$ or $\fphat$ have
even (nonnegative) degree, which is at most four.

First, we wish to localise the Killing superalgebra associated to a
maximally supersymmetric background $(M,g,H,\varphi)$ of Theorem
\ref{thm:maxsusy} at a point $p\in M$. The construction is parallel to
that given in four-dimensions \cite[\S 3.3]{deMedeiros:2016srz} and
eleven-dimensions \cite[\S 3.1, \S
3.2]{Figueroa-OFarrill:2016khp}.

First of all, elements of the Killing superalgebra may be identified
with parallel sections of the supervector bundle
$\eE = \eE_{\bar 0} \oplus \eE_{\bar 1}$,
\begin{equation}
  \eE_{\bar 0} = TM \oplus \fso(TM)\qquad\text{and}\qquad\eE_{\bar 1} =
  \SS\;.
\end{equation}
For extended Killing superalgebras, one needs to set
$\eE_{\bar 0}=TM\oplus\fso(TM)\oplus \mathfrak{sp}(\mathcal H)$. It is
clear that Killing spinors and R-symmetries are parallel w.r.t. the
connection $\eD$, whereas it is a well-known fact that any Killing
vector is identified with a parallel section of $TM\oplus\fso(TM)$ by
the so-called Killing transport.

Hence, any element of the (extended) Killing superalgebra is determined
by the value at $p\in M$ of the corresponding parallel section of $\eE$
and the Killing superalgebra itself defines a graded subspace $\fa$ of
the (extended) Poincaré superalgebra. It is not difficult to see that
\begin{equation}
  \fa=\begin{cases}
    V\oplus S\oplus\fh\quad\text{if}\;\;H^{-}=0\;, \\
    V\oplus S\oplus (\fh\oplus\fr)\quad\text{if}\;\;H^{-}\neq0\;,
  \end{cases}
\end{equation} 
for some subalgebra $\fh$ of $\fso(V)$.  Tracking back the Lie algebra
structure of the Killing superalgebra yields the following Lie brackets
on $\fa$: 
\begin{equation}
  \label{eq:KSAasFDII}
  \begin{split}
    [L,M]&= LM - ML\\
    [L,A]&=0\\
    [A,B]&=AB-BA\\
    [L,s]&=\tfrac12 \omega_L\cdot s\\
    [A,s]&=As\\
    [L,v]&=Lv+\underbrace{[L,X_v]-X_{Lv}}_{\text{element of}\;\fh}\\
    [A,v]&=0\\
    [s,s]&=\kappa(s,s) + \underbrace{\gamma(s,s) -
      X_{\kappa(s,s)}}_{\text{element of}\; \fh}+\underbrace{\rho(s,s)}_{\text{element of}\; \fr}\\
    [v,s]&=\underbrace{\beta_v s + 
      \tfrac12 \omega_{X_v}\cdot s}_{\text{element of}\; S}\\
    [v,w]&=\underbrace{X_v w - X_w v}_{\text{element of}\; V} +
    \underbrace{[X_v,X_w] + R(v,w) - X_{X_v w- X_w v}}_{\text{element of}\; \fh},
  \end{split}
\end{equation}
for all $L,M\in\fso(V)$, $s\in S$, $v,w\in V$ and $A,B\in\fr$. Here
$X:V\to\fso(V)$ is a linear map which geometrically corresponds to the
choice of a basis of $T_pM$ consisting of (the values at the point of)
some Killing vectors, $R$ is the Riemann curvature tensor and the maps
$\beta$, $\gamma$ and $\rho$ are as determined in Theorems
\ref{thm:spencer} and \ref{thm:spencer-R}. If $H^-=0$ then $\fr$ is not
included in $\fa$ and the Lie brackets in \eqref{eq:KSAasFDII} involving
elements $A,B\in\fr$ need to be disregarded. We also note that $\rho=0$
in this case. It is clear from \eqref{eq:KSAasFDII} that the Killing
superalgebra (resp. the extended Killing superalgebra) is isomorphic to
a filtered subdeformation $\fg$ of $\fp$ (resp. $\fphat$), where the
components of positive degree of the Lie brackets are given by the
underbraced elements and the underlying graded Lie superalgebra is $\fa$
with its natural Lie brackets.

Now, one can easily check (compare e.g. with Lemma 3 of
\cite{Figueroa-OFarrill:2016khp}) that the space of cocycles
$Z^{d,2}(\fa_-,\fa)=0$ for all even $d\geq 4$ and that $\fa$ is a full
prolongation of degree $2$ in the sense of \cite{Cheng:1999cy}, that is
$H^{d,1}(\fa_-,\fa)=0$ for all $d\geq 2$ (see also Lemma \ref{lem:h21}). It follows that the
infinitesimal filtered deformation is completely determined by (the
orbit under $\Spin(V)\times \Sp(1)$ of) the corresponding cohomology
class in $H^{2,2}(\fa_-,\fa)$. Associated to the natural inclusion
$\imath$ of $\fa$ into $\fp$ or $\fphat$ there is a map in cohomology
\begin{equation}
  \label{eq:mapadm}
  \begin{split}
    \imath_*&:H^{2,2}(\fa_-,\fa)\to \Lambda^3 V \oplus \left(V \otimes \odot^2\Delta\right)\;,
  \end{split}
\end{equation}
which is easily seen to be injective (remember that
$\fa_-=\fp_-=\fphat_-=V\oplus S$ by maximal supersymmetry).

Hence, the infinitesimal filtered deformation is completely determined
by (the orbit of) $\varphi|_p\in V \otimes \odot^2\Delta$ in the first
branch and $H|_p\in \Lambda^3 V$ in the second branch. With some more
effort, we may show that this data do actually allow to recover the
whole filtered deformation, up to isomorphisms of filtered
subdeformations of $\fp$ or $\fphat$. The claim relies on
$Z^{4,2}(\fa_-,\fa)=0$, the fact that $\fa$ is a full prolongation of
degree $2$ and the general results of \cite{Cheng:1999cy} -- we refer to
\cite[Proposition 10]{deMedeiros:2016srz} for more details.

We have seen that the maximally supersymmetric backgrounds
$(M,g,H,\varphi)$ of Theorem \ref{thm:maxsusy} have different associated
Killing superalgebras, up to isomorphisms of filtered subdeformations of
$\fp$ if $H^{-}=0$ and $\fphat$ if $H^{-}\neq 0$. It is certainly
possible to give an explicit description of these Lie superalgebras
using \eqref{eq:KSAasFDII}, in the spirit of
\cite[\S5]{deMedeiros:2016srz} and
\cite[\S4]{Figueroa-OFarrill:2015rfh}, but we don't do that for the sake of brevity.

\section*{Acknowledgments}

The research of JMF is supported in part by the grant ST/L000458/1
``Particle Theory at the Higgs Centre'' from the UK Science and
Technology Facilities Council. The research of AS is supported by the
project "Lie superalgebra theory and its applications" of the University
of Bologna and partly supported by the Project Prin 2015 "Moduli spaces
and Lie Theory" and Project "RFO14CANTN RFO 2014".

This work was initiated during a visit of JMF to the University of
Stavanger in July 2016 (supported by the Research Council of Norway,
Toppforsk grant no.~250367, held by Sigbjørn Hervik) and it is his
pleasure to thank Paul and Sigbjørn for the invitation and the
hospitality. The third author AS would like to thank José and the
University of Edinburgh for hospitality in November 2017 during the
final stages of this work.

\providecommand{\href}[2]{#2}\begingroup\raggedright\endgroup

\end{document}